\pdfoutput=1
\documentclass[12pt, a4paper, oneside]{article}
\usepackage{latexsym}
\usepackage{amssymb}
\usepackage{mathrsfs}
\usepackage{graphicx}
\usepackage[latin1]{inputenc}
\usepackage{natbib}
\usepackage{multirow}
\usepackage{url}
\usepackage{dcolumn}
\usepackage{amsthm}
\usepackage{amsmath}
\usepackage{array}
\usepackage{amsfonts}
\usepackage{bm}
\usepackage[german]{varioref}
\usepackage{psfrag}
\usepackage{caption}
\usepackage{rotating}
\usepackage{oldgerm}
\usepackage{enumerate}
\usepackage{here}	
\usepackage{pdflscape}
\usepackage{listings}
\usepackage{subcaption}
\usepackage[english]{babel}
\usepackage[normalem]{ulem}
\usepackage[lined,boxed,linesnumbered]{algorithm2e}
\usepackage{color}
\usepackage{footmisc}
\usepackage{enumitem}
\usepackage{dsfont} % fuer die Indikatorfunktion
\usepackage{tabularx}
\usepackage{longtable}

\newcolumntype{L}[1]{>{\raggedright\arraybackslash}p{#1}} % linksb�ndig mit Breitenangabe
\newcolumntype{C}[1]{>{\centering\arraybackslash}p{#1}} % zentriert mit Breitenangabe
\newcolumntype{R}[1]{>{\raggedleft\arraybackslash}p{#1}} % rechtsb�ndig mit Breitenangabe

\long\def\symbolfootnote[#1]#2{\begingroup%
\def\thefootnote{\fnsymbol{footnote}}\footnote[#1]{#2}\endgroup}

\numberwithin{equation}{section}
\newtheorem{Theorem}{Theorem}[section]
\newtheorem{Corollary}[Theorem]{Corollary}
\newtheorem{Proposition}[Theorem]{Proposition}
\newtheorem{Lemma}[Theorem]{Lemma}
\newtheorem{Definition}[Theorem]{Definition}

\newtheorem{example}[Theorem]{Example}
\newtheorem{asu}{}

\allowdisplaybreaks[1]
\newcommand{\CCC}{\mathcal{C}}

\newcommand{\HHH}{\mathcal{H}}

\newcommand{\GGG}{\mathcal{G}}
\newcommand{\III}{\mathcal{I}}

\newcommand{\NNN}{\mathcal{N}}

\newcommand{\SSS}{\mathcal{S}}

\newcommand{\UUU}{\mathcal{U}}
\newcommand{\VVV}{\mathcal{V}}

\newcommand{\coloneqq}{:=}
\newcommand{\eqqcolon}{=:}

\newcommand{\dotcup}{\mathbin{\dot{\cup}}}

\DeclareMathOperator{\Variance}{Var}           % variance-covariance matrix
\DeclareMathOperator*{\argmax}{arg\,max}
\DeclareMathOperator{\ancestors}{an}
\DeclareMathOperator{\descendants}{de}
\DeclareMathOperator{\nondescendants}{nd}
\DeclareMathOperator{\parents}{pa}

\newcommand{\indep}{\mathrel{{\perp}\hspace*{-0.6em}{\perp}}}
\newcommand{\given}{\mathrel{|}}
\newcommand{\condindep}[3]{#1 \indep #2 \given #3}
\DeclareMathOperator{\diag}{diag}
\begin{document}

\title{\textbf{Representing sparse Gaussian DAGs as sparse R-vines allowing for non-Gaussian dependence}}
  \author{Dominik M\"uller\thanks{
  		Dominik M\"uller is Ph.D. student at the Department of Mathematics, Technische Universit\"at M\"unchen,
  		Boltzmannstra\ss{}e 3, 85748 Garching, Germany (E-mail: \href{mailto:dominik.mueller@ma.tum.de}{dominik.mueller@ma.tum.de})}\hspace{.2cm}
  	and
  	Claudia Czado\thanks{
  		Claudia Czado is Associate Professor at the Department of Mathematics, Technische Universit\"at M\"unchen,
  		Boltzmannstra\ss{}e 3, 85748 Garching, Germany (E-mail: \href{mailto:cczado@ma.tum.de}{cczado@ma.tum.de})} \\
  }
  \maketitle

%%%%%%%%%%%%%%%%%%%%%%%%%%%%%%%%%%%%%%%%%%%%%%%%%%%%%%%%%%%

\begin{abstract}
	Modeling dependence in high dimensional systems has become an increasingly important topic. Most approaches rely on the assumption of a multivariate Gaussian distribution such as statistical models on directed acyclic graphs (DAGs). They are based on modeling conditional independencies and are scalable to high dimensions. In contrast, vine copula models accommodate more elaborate features like tail dependence and asymmetry, as well as independent modeling of the marginals. This flexibility comes however at the cost of exponentially increasing complexity for model selection and estimation. We show a novel connection between DAGs with limited number of parents and truncated vine copulas under sufficient conditions. This motivates a more general procedure exploiting the fast model selection and estimation of sparse DAGs while allowing for non-Gaussian dependence using vine copulas. We demonstrate in a simulation study and using a high dimensional data application that our approach outperforms standard methods for vine structure estimation.
\end{abstract}

\noindent%

%%%%%%%%%%%%%%%%%%%%%%%%%%%%%%%%%%%%%%%%%%%%%%%%%%%%%%%%%%%%%%%%%%%%%%%%%%%

%--------------------------------------------------------------------------
\section{Introduction}\label{sec:introduction}
%--------------------------------------------------------------------------
In many areas of natural and social sciences, high dimensional data are collected for analysis. For all these data sets the dependence between the variables in addition to the marginal behaviour needs to be taken into account. While there exist many easily applicable univariate models, dependence models in $d$ dimensions often come with high complexity. Additionally, they put restrictions on the associated marginal distributions, such as the  multivariate Student-t and Gaussian distribution. The latter is also the backbone of statistical models on directed acyclic graphs (DAGs) or Bayesian Networks (BNs), see \citet{Lauritzen1996} and \citet{KollerFriedman2009}. Based on the Theorem of \citet{Sklar1959}, the pair copula construction (PCC) of \citet{Aasetal2009} allows for more flexible $d$ dimensional models. More precisely, the building blocks are the marginal distributions and (conditional) bivariate copulas which can be chosen independently. The resulting models, called \textit{regular vines} or \textit{R-vines} \citep{kuro:joe:2010} are specified by a sequence of $d-1$ linked trees, the R-vine structure. The edges of the trees are associated with bivariate parametric copulas. When the trees are specified by star structures we speak of C-vines, while line structures give rise to D-vines. However, parameter estimation and model selection for R-vine models can be cumbersome, see \citet{cza:2010} and \citet{jaworski2012}. In particular, the sequential approach of \citet{dissmann-etal} builds the R-vine structure from the first tree to the higher trees. Since choices in lower trees put restrictions on higher trees, the resulting model might not be overall optimal in terms of goodness-of-fit. Thus, \citet{dissmann-etal}
model the stronger (conditional) pairwise dependencies in lower trees compared to weaker ones. To reduce model complexity, pair copulas in the trees $k+1$ to $d-1$ can be set to the independence copula resulting in $k$-\textit{truncated} R-vines \citep{brech-cc-2012a}. Another sequential model selection approach is the Bayesian approach of \citet{GruberCzado20152}, while 
\citet{GruberCzado2015} contains a full Bayesian analysis. Both methods are computationally demanding and thus not scalable to high dimensions.\\
Since DAGs are scalable to high dimensions, attempts were made to relate DAGs to R-vines. For example, \citet{bauer:2011} and \citet{BauerCzado2016} provide a PCC to the density factorization of a DAG. While this approach maintains the structure of the DAG, some of the conditional distribution functions in the PCC can not be calculated recursively and thus require high dimensional integration. This limits the applicability in high dimensions dramatically. \citet{Pircalabelu2015} approximate each term in the DAG density factorization by a quotient of a C-vine and a D-vine. However, this yields in general no consistent joint distribution. Finally, \citet{NIPS2010_3956} uses copulas to generalize the density factorization of a DAG to non-Gaussian dependence by exchanging conditional normal densities with copula densities. Yet, the dimension of these copulas is not bounded, inheriting the drawbacks of higher dimensional copula models, i.\,e.\ lack of flexibility and high computational effort.\\
Our goal is to ultimately use the multitude of fast algorithms for estimating sparse Gaussian DAGs in high dimensions to efficiently calculate sparse R-vines. Thus, once a DAG has been selected, we compute an R-vine which represents a similar decomposition of the density as the DAG. This new decomposition allows us to replace Gaussian copula densities and marginals by non-Gaussian pair copula families and arbitrary marginals. We attain this without the drawback of possible higher-dimensional integration as in the approach of \citet{BauerCzado2016}. However, we still exploit conditional independences described by the DAG facilitating parsimony of the R-vine. To attain this, we first build a theoretically sound bridge between DAG with at most $k$ parents, called $k$-DAGs and $k$-truncated R-vines. Since the class of $k$-truncated R-vines is much smaller than the class of $k$-DAGs, such an appealing exact representation will not exist for most DAGs. Yet, we can prove under sufficient conditions when it does and determine special classes of $k$-DAGs which have expressions as $k$-truncated C- and D-vines. Next, we give strong necessary conditions on arbitrary $k$-DAGs to check whether an exact representation as $k$-truncated R-vine exists. If not, we obtain a smallest possible truncation level $k^\prime > k$. All the previous results motivate a more general procedure to find sparse R-vines based on $k$-DAGs, attaining our final goal, to find a novel approach to estimate high dimensional sparse R-vines. The presented method is also independent of the sequential estimation of pair copula families and parameters as used by \citet{dissmann-etal}. Thus, error propagation in later steps caused by misspecification in early steps is prevented. By allowing the underlying DAG model to have at most $k$ parents, we control for a specific degree of sparsity.
%We do this by first selecting a Gaussian DAG structure for a data set and utilize the detected (conditional) independencies to construct a R-vine structure which matches these independencies. In this case the found R-vine structure has the independence copula as pair copula for each DAG implied conditional independence The resulting R-vine structure is then used to fit an R-vine with this structure but with possible non Gaussian pair copulas. While the DAG structure is selected assuming a joint Gaussian distribution, this way of proceeding avoids a sequential tree selection procedure where misspecification of lower trees propagate misspecification in higher trees. This might allow to find better fitting R-vines than are provided by the sequential approach of \citet{dissmann-etal}. Sparsity is induced by limiting the maximal number of parents of each node in the DAG. We give non trivial sufficient conditions under which a DAG with at most $k$ parents can be represented as $k$-truncated R-vine structure. Then, we provide necessary conditions for this representation and develop heuristics to find a sparse $k^\prime$-truncated R-vine structure for $k^\prime \geq k$ representing an arbitrary DAG with at most $k$ parents. The proposed methods are finally applied to a high dimensional data example.\\
The paper is organized as follows: Sections \ref{sec:modelingvines} and \ref{sec:dags} introduce R-vines and DAGs, respectively. Section \ref{sec:representation} contains the main result where we first demonstrate that each (truncated) R-vine can be represented by a DAG non-uniquely. The converse also holds true for $1$-DAGs, i.\,e.\ \textit{Markov trees}. We prove a representation of DAGs as R-vines under sufficient conditions and propose necessary conditions. Afterwards, we develop a general procedure to compute sparse R-vines representing $k$-DAGs. There, we propose a novel technique combining several DAGs. In Section \ref{sec:simstudy}, a high dimensional simulation study shows the efficiency of our approach. We conclude with a high dimensional data application in Section \ref{sec:application} and summarize our contribution. Additional results are contained in an online supplement.
%--------------------------------------------------------------------------
\section{Dependence Modeling with R-vines}\label{sec:modelingvines}
%--------------------------------------------------------------------------
Consider a random vector $\mathbf{X} = \left(X_1,\ldots,X_d\right)$ with joint density function $f$ and joint distribution function $F$. The famous Theorem of \citet{Sklar1959} allows to separate the univariate marginal distribution functions $F_1,\ldots,F_d$ from the dependency structure such that $F\left(x_1,\ldots,x_d\right) = \CCC\left(F_1\left(x_1\right),\ldots,F_d\left(x_d\right)\right)$,
where $\CCC$ is an appropriate $d$-dimensional copula. For continuous $F_i$, $\CCC$ is unique. The corresponding joint density function $f$ is given as
\begin{equation}\label{eq:copuladensity}
	f\left(x_1,\ldots,x_d\right) = \left[\prod_{i=1}^d~f_i\left(x_i\right)\right] \times c\left(F_1\left(x_1\right),\ldots,F_d\left(x_d\right)\right),
\end{equation}
where $c$ is a $d$-dimensional copula density. This representation relies on an appropriate $d$-dimensional copula, which might be cumbersome and analytically not tractable. As shown by \citet{Aasetal2009}, $d$-dimensional copula densities may be decomposed into $d\left(d-1\right)/2$ bivariate (conditional) copula densities. Its backbone, the \textit{pair copulas} can flexibly represent important features like positive or negative tail dependence or asymmetric dependence. The \textit{pair-copula-construction} (PCC) in $d$ dimensions itself is not unique. However, the different possible decompositions may be organized to represent a valid joint density using \textit{regular vines} (R-vines), see \citet{BedfordCooke2001} and \citet{BedfordCooke2002}. To construct a statistical model, a \textit{vine tree sequence} stores which bivariate (conditional) copula densities are present in the presentation of a $d$-dimensional copula density. 
More precisely, such a sequence in $d$ dimensions is defined by $\VVV = \left(T_1,\ldots,T_{d-1}\right)$ such that
\vspace{-2mm}
\begin{enumerate}[label={(\roman*})]
	\itemsep0em 
	\item $T_1$ is a tree with nodes $V_1=\left\{1,\ldots,d\right\}$ and edges $E_1$,
	\item for $i \ge 2$, $T_i$ is a tree with nodes $V_i = E_{i-1}$ and edges $E_i$,
	\item if two nodes in $T_{i+1}$ are joined by an edge, the corresponding edges in $T_i$ must share a common node (proximity condition).
	\label{eq:proximitycondition}
\end{enumerate}
\vspace{-2mm}
Since edges in a tree $T_{i-1}$ become nodes in $T_i$, denoting edges in higher order trees is complex. For example, edges $\left\{a,c\right\}, \left\{a,b\right\} \in E_1$ are nodes in $T_2$ and an edge in $T_2$ between these nodes is denoted $\left\{\left\{c,a\right\},\left\{a,b\right\}\right\} \in E_2$. To shorten this set formalism, we introduce the following. For a node $f \in V_i$ we call a node $e \in V_{i-1}$ an \textit{m-child} of $f$ if $e$ is an element of $f$. If $e \in V_1$ is reachable via inclusions $e \in e_1 \in \ldots \in f$, we say $e$ is an \textit{m-descendant} of $f$. We define the \textit{complete union} $A_e$ of an edge $e$ by $A_e \coloneqq \left\{j \in V_1|\exists \ e_1 \in E_1,\ldots,e_{i-1}\in E_{i-1}: j \in e_1 \in \ldots \in e_{i-1} \in e\right\}$ where the \textit{conditioning set} of an edge $e=\left\{a,b\right\}$ is defined as $D_e \coloneqq A_a \cap A_b$ and
$C_e \coloneqq C_{e,a} \cup C_{e,b} \mbox{ with } C_{e,a} \coloneqq A_a \setminus D_e \mbox{ and } C_{e,b} \coloneqq A_b \setminus D_e$ is the \textit{conditioned set}. Since $C_{e,a}$ and $C_{e,b}$ are singletons, $C_e$ is a doubleton for each $e,a,b$, see \citet[p.\ 96]{KurowickaCooke2006}. For edges $e \in E_i,\ 1 \le i \le d-1$, we define the set of bivariate copula densities corresponding to $j\left(e\right),\ell\left(e\right)|D\left(e\right)$ by $\mathcal{B}\left(\VVV\right) = \left\{c_{j\left(e\right),\ell\left(e\right);D\left(e\right)}|e \in E_i, 1 \le i \le d-1\right\}$ with the conditioned set $j\left(e\right),\ell\left(e\right)$ and the conditioning set $D\left(e\right)$. Denote sub vectors of $\mathbf{x}=\left(x_1,\ldots,x_d\right)^T$ by $\mathbf{x}_{D\left(e\right)} \coloneqq \left(\mathbf{x}_j\right)_{j \in D\left(e\right)}$. With the PCC, Equation \eqref{eq:copuladensity} becomes
\begin{equation}\label{eq:vinedensity}
	f\left(x_1,\ldots,x_d\right) = \left[\prod_{i=1}^{d}~f_i\left(x_i\right)\right]
	\times \left[\prod_{i=1}^{d-1}~\prod_{e\in E_i}~c_{j\left(e\right),\ell\left(e\right);D\left(e\right)}\left(F\left(x_{j\left(e\right)}|x_{D\left(e\right)}\right),F\left(x_{\ell\left(e\right)}|x_{D\left(e\right)}\right)\right)\right].
\end{equation}
By referring to bivariate \textit{conditional} copulas, we implicitly take into account the \textit{simplifying assumption}, which states that the two-dimensional conditional copula density $c_{13;2}\left(F_{1|2}\left(x_1|x_2\right),F_{3|2}\left(x_3|x_2\right);x_2\right)$
is independent of the conditioning value $X_2 = x_2$, see \citet{stoeber-vines} for a detailed discussion. Henceforth, in our considerations we assume the simplifying assumption. We define the parameters of the bivariate copula densities $\mathcal{B}\left(\VVV\right)$ by $\theta\left(\mathcal{B}\left(V\right)\right)$. This determines the R-vine copula $\left(\VVV,\mathcal{B}\left(\VVV\right),\theta\left(\mathcal{B}\left(\VVV\right)\right)\right)$. 
A convenient way to represent R-vines uses lower triangular $d\times d$ matrices, see \citet{dissmann-etal}. 

\begin{example}[R-vine in 6 dimensions]\label{ex:exvine1}
	The R-vine tree sequence in Figure \ref{fig:exvine1:1} is given by the R-vine matrix $M$ as follows. Edges in $T_1$ are pairs of the main diagonal and the lowest row, e.\,g.\ $\left(2{,}1\right)$, $\left(6{,}2\right)$, $\left(3{,}6\right)$, etc. $T_2$ is described by the main diagonal and the second last row conditioned on the last row, e.\,g.\ $6{,}1|2$; $3{,}2|6$, etc. Higher order trees are characterized similarly. For a column $p$ in $M$, only entries of the main-diagonal right of $p$, i.\,e.\ values in $M_{p+1,p+1},\dots,M_{d,d}$ are allowed and no entry must occur more than once in a column.\\
	\begin{tabular}{cc}
		\begin{minipage}[l]{0.25\textwidth}
			\begin{equation*}
				\left(
				\begin{array}{cccccc}
					4&&&&&\\
					1&5&&&&\\
					3&1&3&&&\\
					6&3&1&6&&\\
					2&6&2&1&2&\\
					5&2&6&2&1&1
				\end{array}
				\right)
			\end{equation*}
			\begin{center}
				R-vine matrix $M$
			\end{center}
		\end{minipage}
		&
		\begin{minipage}[r]{0.75\textwidth}
			%		\begin{figure}[h]
			\centering
			\includegraphics[width=0.35\textwidth]{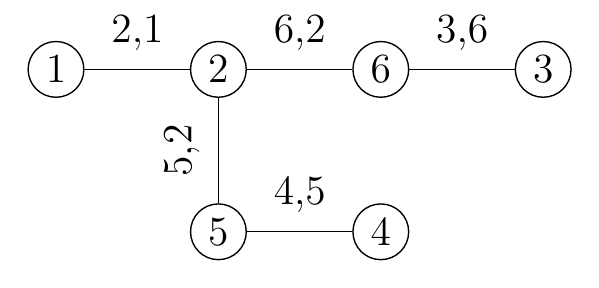}
			\includegraphics[width=0.35\textwidth]{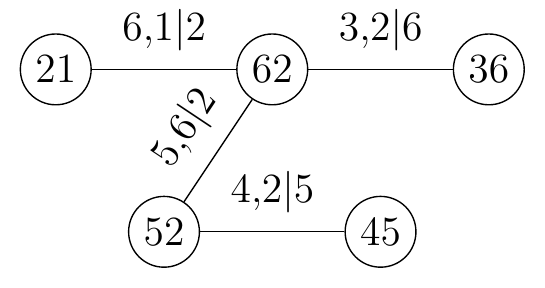}\\
			\includegraphics[width=0.35\textwidth]{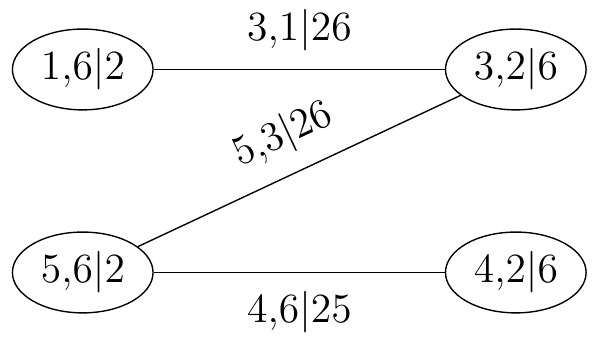}
			\includegraphics[width=0.25\textwidth]{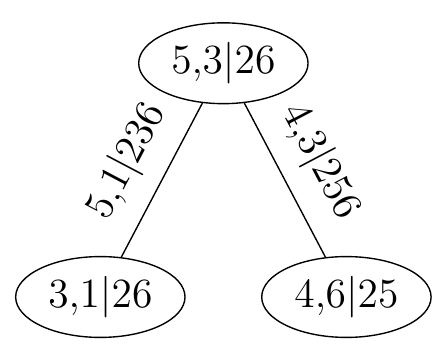}
			\includegraphics[width=0.3\textwidth]{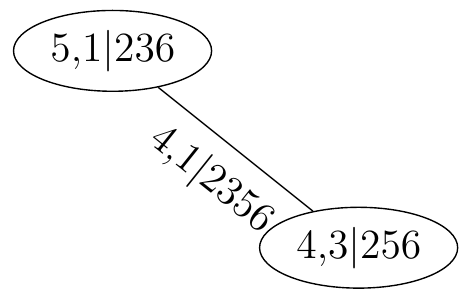}
			\captionof{figure}{R-vine trees $T_1, T_2$ (top), $T_3, T_4, T_5$ (bottom), left to right.}
			\label{fig:exvine1:1}
			%		\end{figure}
		\end{minipage}
	\end{tabular}
	\\[0.5cm]
	%	\begin{equation*}
	%	M=\left(
	%	\begin{array}{cccccc}
	%	4&&&&&\\
	%	1&5&&&&\\
	%	6&1&3&&&\\
	%	3&6&1&6&&\\
	%	5&3&6&1&2&\\
	%	2&2&2&2&1&1
	%	\end{array}
	%	\right)
	%	\end{equation*}
	See Table \ref{table:exvine1:1} for a non exhaustive list of \textit{m-children} and \textit{m-descendants} of edges in the R-vine trees $T_1,T_2,T_3$. For the complete list, see Appendix \ref{sec:appendix_examples}, Example \ref{ex:appendix_example}.\\
	%		\begin{figure}[h]
	%			\centering
	%			\includegraphics[width=0.18\textwidth]{Vine1n}
	%			\includegraphics[width=0.15\textwidth]{Vine2n}
	%			\includegraphics[width=0.18\textwidth]{Vine3n}
	%			\includegraphics[width=0.18\textwidth]{Vine4n}
	%			\includegraphics[width=0.18\textwidth]{Vine5n}
	%			\caption{R-vine trees $T_1, T_2, T_3$ (top), $T_4, T_5$ (bottom), left to right.}
	%			\label{fig:exvine1:1}
	%		\end{figure}
	$$
	\begin{array}{p{0.04\textwidth}p{0.48\textwidth}p{0.18\textwidth}p{0.22\textwidth}}
	tree  & edge $e$ & m-children of $e$ & m-descendants of $e$\\
	\hline \hline
	$T_1$ & $2{,}1$ & $1{,}2$ & $1{,}2$ \\
	%	& $6{,}2$ & $2{,}6$ & $2{,}6$ \\
	%	& $3{,}6$ & $3{,}6$ & $3{,}6$ \\
	\hline
	$T_2$ & $6{,}1|2=\left\{\left\{2,1\right\},\left\{6,2\right\}\right\}$ & $\left\{2,1\right\};\left\{6,2\right\}$ & $1{,}2{,}6$ \\
	& $3{,}2|6=\left\{\left\{3,6\right\},\left\{6,2\right\}\right\}$ & $\left\{6,2\right\};\left\{3,6\right\}$ & $2{,}6{,}3$ \\
	\hline
	$T_3$ &$ 3{,}1|26=\left\{\left\{\left\{2,1\right\},\left\{6,2\right\}\right\},\left\{\left\{2,1\right\},\left\{6,2\right\}\right\}\right\} $&$ 6{,}1|2$; $ 3{,}2|6 $& $1{,}2{,}6{,}3$ \\
	\hline
	\end{array}
	$$
	\captionof{table}{Exemplary edges, m-children and m-descendants in the R-vine trees $T_1,T_2,T_3$.}
	\label{table:exvine1:1}
	With $c_{i,j|k} \coloneqq c_{i,j;k}\left(F\left(x_i|x_k\right),F\left(x_j|x_k\right)\right)$, $\mathbf{x} = \left(x_1,\dots,x_{6}\right)$, $f_i \coloneqq f_i(x_i)$, the density becomes
	\begin{equation*}
		\begin{aligned}
			f\left(\mathbf{x}\right) = & f_1 \times f_2 \times f_3 \times f_4 \times f_5 \times f_6 \times c_{2,1} \times c_{6,2} \times c_{3,6} \times c_{5,2}\times c_{4,5} \times c_{6,1|2} \times c_{3,2|6} \\
			&\times c_{5,6|2} \times c_{4,2|5} \times c_{3,1|26} \times c_{5,3|26} \times c_{4,6|25} \times c_{5,1|236} \times c_{4,3|256}  \times c_{4,1|2356}.
		\end{aligned}
	\end{equation*}
\end{example}
As we model $d\left(d-1\right)/2$ edges, the model complexity is increasing quadratically in $d$. We can ease this by only modeling the first $k$ trees and assuming (conditional) independence for the remaining $d-1-k$ trees. Thus, the model complexity increases linearly. This \textit{truncation} is discussed in detail by \citet{brech-cc-2012a}. Generally, for $k \in \left\{1,\ldots,d-2\right\}$, a $k$-truncated R-vine is an R-vine where each pair copula density $c_{j\left(e\right),\ell\left(e\right);D\left(e\right)}$ assigned to an edge $e \in \left\{E_{k+1},\ldots,E_{d-1}\right\}$ is represented by the independence copula density $c^\perp\left(u_1,u_2\right) \equiv 1$. In a $k$-truncated R-vine, Equation \eqref{eq:vinedensity} becomes
\begin{equation*}
	f\left(x_1,\ldots,x_d\right) = \left[\prod_{i=1}^{d}~f_i\left(x_i\right)\right] \times \left[\prod_{i=1}^{k}~\prod_{e\in E_i}~c_{j\left(e\right),\ell\left(e\right);D\left(e\right)}\left(F\left(x_{j\left(e\right)}|x_{D\left(e\right)}\right),F\left(x_{\ell\left(e\right)}|x_{D\left(e\right)}\right)\right)\right].
\end{equation*}
In Example \ref{ex:exvine1}, we obtain a $k$-truncated R-vine by setting $c_{i,j|D} = c^\perp$ whenever $\left|D\right| \ge k$. The most complex part of estimating an R-vine copula is the structure selection. To solve this, \citet{dissmann-etal} \label{page:dissmann} suggest to calculate a maximum spanning tree with edge weights set to absolute values of empirical Kendall's $\tau$. The intuition is to model strongest dependence in the first R-vine trees. After selecting the first tree, pair copulas and parameters are chosen by maximum likelihood estimation for each edge. Based on the estimates, pseudo-observations are derived from the selected pair-copulas. Kendall's $\tau$ is estimated for these pseudo-observations to find a maximum spanning tree by taking into account the proximity condition. Thus, higher order trees are dependent on the structure, pair copulas and parameters of lower order trees. Hence, this sequential greedy approach is not guaranteed to lead to optimal results in terms of e.\,g.\ log-likelihood, AIC or BIC. \citet{GruberCzado2015} developed a Bayesian approach which allows for simultaneous selection of R-vine structure, copula family and parameters to overcome the disadvantages of sequential selection. However, this approach comes at the cost of higher computational effort and is not feasible in high dimensional set-ups, i.\,e.\ for more than ten dimensions.
%--------------------------------------------------------------------------
\section{Graphical models}\label{sec:dags}
%-------------------------------------------------q-------------------------
\subsection{Graph theory}\label{subsec:graphtheory}
We introduce necessary graph theory from \citet[pp.\ 4--7]{Lauritzen1996}. A comprehensive list with examples is given in Appendix \ref{sec:appendix_definitionsgraphtheory}. Let $V \neq \emptyset$ be a finite set, the \textit{node set} and let $E \subseteq \left\{\left(v,w\right)|\left(v,w\right) \in V \times V \mbox{ with } v \neq w\right\}$ be the \textit{edge set}. We define a graph $\GGG = \left(V,E\right)$ as a pair of node set and edge set. An edge $\left(v,w\right)$ is undirected if $\left(v,w\right)\in E \Rightarrow \left(w,v\right)\in E$, and $\left(v,w\right)$ is \textit{directed} if $\left(v,w\right)\in E \Rightarrow \left(w,v\right)\notin E$. A directed edge $\left(v,w\right)$ is called an \textit{arrow} and denoted $v \rightarrow w$ with $v$ the \textit{tail} and $w$ the \textit{head}. The existence of a directed edge between $v$ and $w$ without specifying the orientation is denoted by $v \leftrightarrow w$ and no directed edge between $v$ and $w$ regardless of orientation is denoted by $v \nleftrightarrow w$. If a graph only contains undirected edges, it is an \textit{undirected} graph and if it contains only directed edges, it is a \textit{directed} graph. We will not consider graphs with both directed and undirected edges. A weighted graph is a graph $\GGG=\left(V,E\right)$ with \textit{weight function} $\mu$ such that $\mu: E \rightarrow \mathbb{R}$. By replacing all arrows in a directed graph $\GGG$ by undirected edges, we obtain the \textit{skeleton} $\GGG^s$ of $\GGG$. Let $\GGG = \left(V,E\right)$ be a graph and define a \textit{path} of length $k$ from nodes $\alpha$ to $\beta$ by a sequence of distinct nodes $\alpha = \alpha_0, \ldots ,\alpha_k = \beta$ such that $\left(\alpha_{i-1},\alpha_i\right) \in E$ for $i=1,\ldots,k$. This applies to both undirected and directed graphs. A \textit{cycle} is defined as path with $\alpha=\beta$. A graph without cycles is called \textit{acyclic}. In a directed graph, a \textit{chain} of length $k$ from $\alpha$ to $\beta$ is a sequence of distinct nodes $\alpha=\alpha_0, \ldots, \alpha_k=\beta$ with $\alpha_{i-1} \rightarrow \alpha_i$ or $\alpha_{i} \rightarrow \alpha_{i-1}$ for $i=1,\ldots,k$. Thus, a directed graph may contain a chain from $\alpha$ to $\beta$ but no path from $\alpha$ to $\beta$. A graph $\HHH=\left(W,F\right)$ is a \textit{subgraph} of $\GGG = \left(V,E\right)$ if $W \subseteq V$ and $F \subseteq E$. We speak of an \textit{induced subgraph} $\HHH=\left(W,F\right)$ if $W \subseteq V$ and $F=\left\{\left(v,w\right)|\left(v,w\right) \in W \times W \mbox{ with } v \neq w\right\} \cap E$, i.\,e.\ $\HHH$ contains a subset of nodes of $\GGG$ and all the edges of $\GGG$ between these nodes. If $\GGG=\left(V,E\right)$ is undirected and a path from $v$ to $w$ exists for all $v,w \in V$, we say that $\GGG$ is \textit{connected}. If $\GGG=\left(V,E\right)$ is directed we say that $\GGG$ is weakly connected if a path from $v$ to $w$ exists for all $v,w \in V$ in the skeleton $\GGG^s$ of $\GGG$. If an undirected graph is connected and acyclic, it is a \textit{tree} and has $d-1$ edges on $d$ nodes. For $\GGG$ undirected, $\alpha,\beta \in V$, a set $C \subseteq V$ is said to be an $\left(\alpha,\beta\right)$ \textit{separator} in $\GGG$ if all paths from $\alpha$ to $\beta$ intersect $C$. $C$ is said to \textit{separate} $A$ from $B$ if it is an $\left(\alpha,\beta\right)$ separator for every $\alpha \in A$, $\beta \in B$. 

\subsection{Directed acyclic graphs (DAGs)}\label{subec:dags}
Let $\GGG=\left(V,E\right)$ be a directed acyclic graph (DAG). If there exists a path from $w$ to $v$, we write $w >_\GGG v$. Denote a disjoint union by $\dotcup$, and define the \textit{parents} $\parents\left(v\right) \coloneqq \left\{w \in V|w \rightarrow v\right\}$, \textit{ancestors} $\ancestors\left(v\right) \coloneqq \left\{w \in V| w >_\GGG v \right\}$, \textit{descendants} $\descendants\left(v\right) \coloneqq \left\{w \in V|v >_\GGG w\right\}$ and \textit{non-descendants} $\nondescendants\left(v\right) \coloneqq V \setminus \left(\descendants\left(v\right) \dotcup \parents\left(v\right) \dotcup v\right)$. We see $V = v \dotcup \parents \left(v\right) \dotcup \descendants\left(v\right) \dotcup \nondescendants\left(v\right)$ for all $v \in V$. $A \subseteq V$ is \textit{ancestral} if $\parents\left(v\right) \subseteq A$ for all $v \in A$, with $An\left(A\right)$ the smallest ancestral set containing $A$. Let $k_v \coloneqq \left|\parents\left(v\right)\right| \mbox{ and } k \coloneqq \max_{v \in V}k_v$ for all $v \in V$. A DAG with at most $k$ parents is called $k$-DAG. For each DAG $\GGG$ there exists a \textit{topological ordering}, see \cite{AnderssonPerlman1998}. This is formalized by an ordering function $\eta$. Let $V = \left\{v_1,\ldots,v_d\right\}$ and $\eta:\ V \rightarrow \left\{1,\ldots,d\right\}$ such that for each pair $v_i{,}v_j \in V$ we have $\eta\left(v_i\right) < \eta\left(v_j\right) \Rightarrow v_j \not >_\GGG v_i$, i.\,e.\ there is no path from $v_j$ to $v_i$ in $\GGG$. An ordering $\eta$ always exists, but is not necessarily unique. By $\left\{\eta^{-1}\left(1\right),\ldots,\eta^{-1}\left(d\right)\right\}$, we refer to $V$ ordered increasingly according to $\eta$ and by $\left\{\eta^{-1}\left(d\right),\ldots,\eta^{-1}\left(1\right)\right\}$ we refer $V$ ordered decreasingly according to $\eta$. A \textit{v-structure} in $\GGG$ is a triple of nodes $\left(u,v,w\right) \in V$ where $u \rightarrow v$ and $w \rightarrow v$ but $u \nleftrightarrow w$. The \textit{moral graph} $\GGG^m$ of a DAG $\GGG$ is the skeleton $\GGG^s$ of $\GGG$ with an additional undirected edge $\left(u,w\right)$ for each \textit{v-structure} $\left(u,v,w\right)$. As for undirected graphs, separation can also be defined for DAGs, called \textit{d-separation}\label{page:dsep}. Let $\GGG = \left(V,E\right)$ be an DAG. A chain $\pi$ from $a$ to $b$ in $\GGG$ is \textit{blocked} by a set of nodes $S$, if it contains a node $\gamma \in \pi$ such that either
\begin{enumerate}
	\itemsep0em 
	\item $\gamma \in S$ and arrows of $\pi$ do not meet head-to-head at $\gamma$ (i.\,e.\ at $\gamma$ there is no v-structure with nodes of $\pi$), or
	\item $\gamma \notin S$ nor has $\gamma$ any descendants in $S$, and arrows of $\pi$ do meet head-to-head at $\gamma$ (i.\,e.\ at $\gamma$ there is a v-structure with nodes of $\pi$).
\end{enumerate}
A chain that is not blocked by $S$ is \textit{active}. Two subsets $A$ and $B$ are \textit{d-separated} by $S$ if all chains from $A$ to $B$ are blocked by $S$.

\begin{example}[DAG in 6 dimensions]\label{ex:exdag1}
	Table \ref{table:exdag1} displays the topological ordering function, parents, descendants and non-descendants for all $v \in V$ of the DAG $\GGG_1$ in Figure \ref{fig:exdag1}.
	\begin{tabular}{cc}
		\begin{minipage}[l]{0.25\textwidth}
			\centering
			\includegraphics[width=0.9\textwidth, trim={0.0cm 0.0cm 0.0cm 0.0cm},clip]{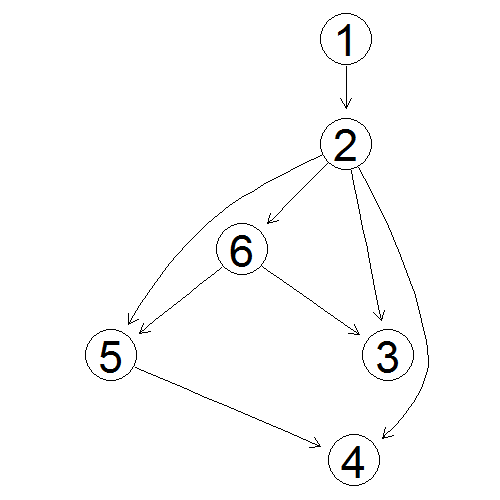}
			\captionof{figure}{DAG $\GGG_1$}
			\label{fig:exdag1}
		\end{minipage}
		&
		\begin{minipage}[r]{0.75\textwidth}
			\begin{table}[H]
				\centering
				\begin{tabular}{rrrrr}
					$v$ & $\eta\left(v\right)$ & $\parents\left(v\right)=\left\{w_1^v,w_2^v\right\}$ & $\descendants\left(v\right)$ & $\nondescendants\left(v\right)$ \\
					\hline \hline
					1 & 1 & - & 2,3,4,5,6 & - \\
					2 & 2 & 1 & 3,4,5,6 & - \\
					3 & 4 & 6,2 & - & 1,4,5\\
					4 & 6 & 5,2 & - & 1,3,6 \\
					5 & 5 & 6,2 & 4 & 1,3\\
					6 & 3 & 2 & 3,4,5 & 1\\
					\hline
				\end{tabular}
				\captionof{table}{Properties of DAG $\GGG_1$.}
				\label{table:exdag1}
			\end{table}
		\end{minipage}
	\end{tabular}
	\\[0.5cm]
	%	\begin{figure}[h]
	%		\centering	
	%		\includegraphics[width=0.3\textwidth, trim={0.0cm 0.0cm 0.0cm 0.0cm},clip]{exdag}
	%		\caption{Example \ref{ex:exdag1}: DAG $\GGG_1$}
	%		\label{fig:exdag1}
	%	\end{figure}
	%	\begin{table}[H]
	%		\centering
	%		\begin{tabular}{rrrrr}
	%			$v$ & $\eta\left(v\right)$ & $\parents\left(v\right)=\left\{w_1^v,w_2^v\right\}$ & $\descendants\left(v\right)$ & $\nondescendants\left(v\right)$ \\
	%			\hline \hline
	%			1 & 1 & - & 2,3,4,5,6 & - \\
	%			2 & 2 & 1 & 3,4,5,6 & - \\
	%			3 & 4 & 6,2 & - & 1,4,5\\
	%			4 & 6 & 6,2 & - & 1,3,6 \\
	%			5 & 5 & 6,2 & 4 & 1,3\\
	%			6 & 3 & 2 & 3,4,5 & 1\\
	%			\hline
	%		\end{tabular}
	%		\caption{Properties of $\GGG_1$.}
	%		\label{table:exdag1}
	%	\end{table}
	A high value of $\eta\left(v\right)$ corresponds to more non-descendants. $\eta$ is not unique since $3 \nleftrightarrow 5$ in $\GGG_1$. Hence, a topological ordering for $\GGG_1$ is also $\left\{\eta^{-1}\left(1\right),\ldots,\eta^{-1}\left(6\right)\right\}=\left\{1,2,6,5,3,4\right\}$.
\end{example}
\subsection{Markov properties on graphs}\label{subsec:markovproperties}
Let $V=\left\{1,\ldots,d\right\}$ and consider a random value $\mathbf{X}=\left(X_1,\dots,X_d\right) \in \mathbb{R}^d$ distributed according to a probability measure $P$. For $I \subseteq V$ define $\mathbf{X}_I \coloneqq \left(X_v\right)_{v \in I}$ and denote the conditional independence of the random vectors $\mathbf{X}_A$ and $\mathbf{X}_B$ given $\mathbf{X}_C$ by $\condindep {A}{B}{C}$. Let $\GGG= \left(V,E\right)$ be a DAG, then $P$ obeys the \textit{local directed Markov property} according to $\GGG$ if
\begin{equation}\label{eq:directedlocalmarkov}
	\condindep{v}{\nondescendants\left(v\right)}{\parents\left(v\right)} \mbox{ for all } v \in V.
\end{equation}
\begin{example}[Example \ref{ex:exdag1} cont.]
	The local directed Markov property \eqref{eq:directedlocalmarkov} for the DAG $\GGG_1$ in Figure \ref{fig:exdag1} gives $\condindep {4}{1,3,6}{2,5}$; $\condindep {5}{1,3}{2,6}$; $\condindep {3}{1,4,5}{2,6}$ and $\condindep {6}{1}{2}$.
\end{example}

From \citet[p.\ 51]{Lauritzen1996}, $P$ has the local directed Markov property according to $\GGG$ if and only if it has the \textit{global directed Markov property} according to $\GGG$, which states that for $A,B,C \subseteq V$ we have that $\condindep ABC \mbox{ if } A \mbox{ and } B \mbox{ are separated by } C \mbox{ in } \left(\GGG_{An\left(A\dotcup B \dotcup C\right)}\right)^m.$ Thus, inferring conditional independences using this property requires undirected graphs. To use directed graphs, we can employ the d-separation. \citet[p.\ 48]{Lauritzen1996} showed that for a DAG $\GGG= \left(V,E\right)$ and $A,B,C \subseteq V$ disjoint sets, $C$ \textit{d-separates} $A$ from $B$ in $\GGG$ if and only if $C$ separates $A$ from $B$ in $\left(\GGG_{An\left(A\dotcup B\dotcup C\right)}\right)^m$. The conditional independences drawn from a DAG can be exploited using the following Proposition, see \citet[p.\ 33]{Whittaker1990}.
\begin{Proposition}[Conditional independence]\label{prop:CI_lemmas}
	If $\left(\mathbf{X},\mathbf{Y},\mathbf{Z}_1,\mathbf{Z}_2\right)$ is a partitioned random vector with joint density $f_{X,Y,Z_1,Z_2}$, then the following expressions are equivalent:
	\begin{enumerate}
		\itemsep0em 
		\item $\condindep {\mathbf{Y}} {\left(\mathbf{Z}_1,\mathbf{Z}_2\right)} \mathbf{X}$,
		\item $\condindep {\mathbf{Y}} {\mathbf{Z}_2} {\left(\mathbf{X}, \mathbf{Z}_1\right)}$ and $\condindep {\mathbf{Y}}{\mathbf{Z_1}}{\mathbf{X}}$.
	\end{enumerate}
\end{Proposition}
%\begin{proof}
%	To show $(i) \Rightarrow (ii)$, we refer to the block independence Lemma in \citet[p.\ 33]{Whittaker1990} for the first assertion. To show the second part of $(ii)$ with $\mathbf{Z}_2 \in \mathbb{R}^p$ we have
%	\begin{equation*}
%	\begin{aligned}
%	f_{\mathbf{Y,Z_1|X}}(\mathbf{y,z_1;x})&=\int_{\mathbb{R}^p}f_{\mathbf{Y,Z_1,Z_2|X}}(\mathbf{y,z_1,z_2;x})~d\mathbf{z_2} = \int_{\mathbb{R}^p}f_{\mathbf{Y|X}}(\mathbf{y;x})f_{\mathbf{Z_1,Z_2|X}}(\mathbf{z_1,z_2;x})~d\mathbf{z_2}\\
%	& = f_{\mathbf{Y|X}}(\mathbf{y;x}) \int_{\mathbb{R}^p}f_{\mathbf{Z_1,Z_2|X}}(\mathbf{z_1,z_2;x})~d\mathbf{z_2} = f_{\mathbf{Y|X}}(\mathbf{y;x}) f_{\mathbf{Z_1|X}}(\mathbf{z_1;x}),
%	\end{aligned}
%	\end{equation*}
%	using $(i)$ in the second line. Direction $(ii) \Rightarrow (i)$ is given in \citet[p.\ 35]{Whittaker1990}.
%\end{proof}
To estimate DAGs, a specific distribution is assumed. For continuous data, it is most often the multivariate Gaussian. There exists a multitude of algorithms, see \citet{bnlearnpaper}, which are applicable also in high dimensions.  While we are aware that assuming Gaussianity might be too restrictive for describing the data adequately, we consider the estimated DAG as proxy for an R-vine. An R-vine is however not restricted to Gaussian pair copulas or marginals, relaxing the severe restrictions which come along with DAG models.
%--------------------------------------------------------------------------
\section{Representing DAGs as R-vines}\label{sec:representation}
%--------------------------------------------------------------------------
First, we show that each Gaussian R-vine has a representation as a Gaussian DAG. Second, we demonstrate that the converse also holds for the case of $1$-DAGs, i.\,e.\ \textit{Markov-trees}. For the case $k \ge 2$, a representation of $k$-DAGs as $k$-truncated R-vines is not necessarily possible. We prove under sufficient conditions when such a representation exists. Finally, we derive necessary conditions to infer if an R-vine representation of a $k$-DAG is possible and which truncation level $k^\prime > k$ can be attained at best.
\subsection{Representing truncated R-vines as DAGs}\label{subsec:representationrvinesdags}
To establish a connection between $k$-truncated Gaussian R-vines and DAGs, we follow \cite{BrechmannJoe2014} using \textit{structural equation models (SEMs)}. Define a SEM corresponding to a Gaussian R-vine with structure $\VVV$, denoted by $\SSS\left(\VVV\right)$. Let $\VVV = T_1,\dots,T_{d-1}$ be an R-vine tree sequence and assume without loss of generality $\left\{1,2\right\} \in T_1$ and for $j = 3,\dots,d$ denote the edges in $T_1$ by $\left\{j, \kappa_1\left(j\right)\right\}$. The higher order trees contain edges $j,\kappa_i\left(j\right)|\kappa_1\left(j\right),\dots,\kappa_{i-1}\left(j\right) \in T_i$ for $i = 2,\dots,d-1$. Based on this R-vine, define $\SSS\left(\VVV\right)$ by
\begin{equation}\label{eq:sem}
	\begin{aligned}
		X_1 &= \psi_1 \epsilon_1,\\
		X_2 &= \varphi_{21} X_1 + \psi_2 \epsilon_2,\\
		X_i &= \sum_{j = 1}^{i-1}~\varphi_{i \kappa_j\left(i\right)} X_{\kappa_j\left(i\right)} + \psi_i \epsilon_i,
	\end{aligned}
\end{equation}
with $\epsilon_i \sim \NNN\left(0,1\right)$ i.i.d. and $\psi_i$ such that $\Variance\left(X_i\right) = 1$ for $i = 1,\dots,d$. From $\SSS\left(\VVV\right)$ we obtain a graph $\GGG = \left(V = \left\{1,\dots,d\right\}, E = \emptyset\right)$ and add a directed edge $X_{\kappa_j\left(i\right)} \rightarrow X_i$ for each $i \in 2,\dots,d$ and $j = 1,\dots,i$. In other words, each conditioned set of the R-vine yields an arrow. By the structure of $\SSS\left(\VVV\right)$, $\GGG$ is a DAG. By \cite{PetersBuehlmann2014}, the joint distribution of $\left(X_1,\dots,X_d\right)$ is uniquely determined by $\GGG$ and it is Markov with respect to $\GGG$. Additionally, if the R-vine is $k$-truncated, we have at most $k$ summands on the right hand side and thus, obtain a $k$-DAG. Furthermore, $\GGG$ has a topological ordering $1,\dots,d$. We show that it is possible for two different R-vines to have the same DAG representation.

\begin{example}[Different $2$-truncated R-vines with same DAG representation in $4$ dimensions]\label{ex:semunique}
	Consider the following two $2$-truncated R-vines and their $2$-DAG representation.\\
	\begin{tabular}{lll}
		\begin{minipage}[l]{0.3\textwidth}
			\centering
			\includegraphics[width=0.4\textwidth, trim={0.1cm 0.1cm 0.1cm 0.1cm},clip]{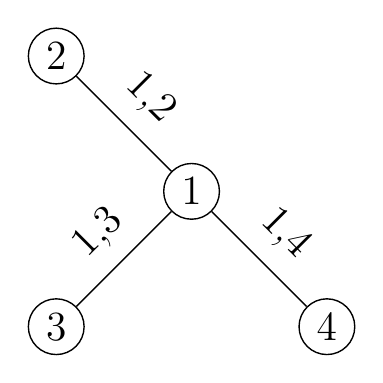}
			\includegraphics[width=0.4\textwidth, trim={0.1cm 0.1cm 0.1cm 0.1cm},clip]{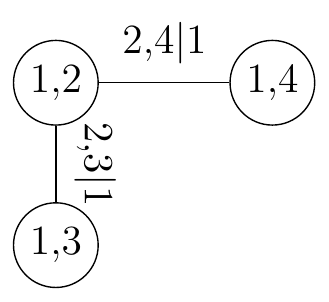}
			\captionof{figure}{R-vine $\VVV_1$.}
			\label{fig:exsemunique:vine1}
		\end{minipage}
		&
		\begin{minipage}[l]{0.3\textwidth}
			\centering
			\includegraphics[width=0.4\textwidth, trim={0.1cm 0.1cm 0.1cm 0.1cm},clip]{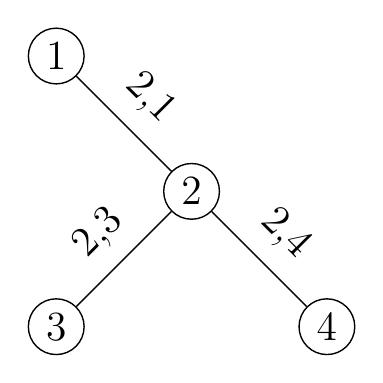}
			\includegraphics[width=0.4\textwidth, trim={0.1cm 0.1cm 0.1cm 0.1cm},clip]{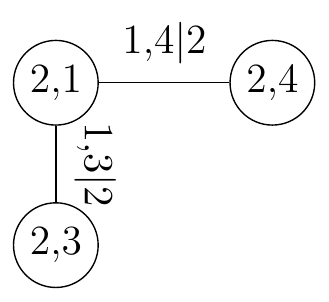}\\
			\captionof{figure}{R-vine $\VVV_2$.}
			\label{fig:exsemunique:vine2}
		\end{minipage}
		&
		\begin{minipage}[l][][l]{0.3\textwidth}
			\centering
			\includegraphics[width=0.425\textwidth, trim={0.1cm 0.1cm 0.1cm 0.1cm},clip]{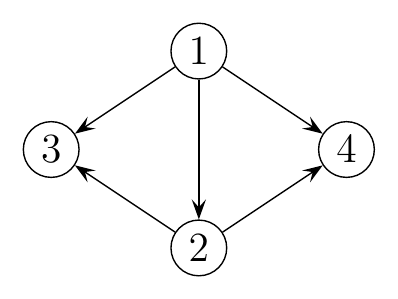}
			\hspace{1.5cm}
			\captionof{figure}{DAG $\GGG_2$ of $\VVV_1$, $\VVV_2$.}
			\label{fig:exsemunique:dag}
		\end{minipage}
	\end{tabular}\\[0.5cm]
	Since the conditioned sets of $\VVV_1$ and $\VVV_2$ in their first two trees are the same, both R-vines have the same DAG representation $\GGG_2$. Assuming fixed SEM coefficients $\varphi$, both R-vines also have different correlation matrices. Yet, both correlation matrices are belonging to distributions which are Markov with respect to $\GGG_2$.
\end{example}

Since two R-vines may have the same representing DAG, inferring an R-vine from a DAG uniquely is not necessarily possible. We  formalize an R-vine representation of a DAG.

\begin{Definition}[R-vine representation of DAG]\label{def:rvinerep}
	Let $\GGG = \left(V,E\right)$ be a $k$-DAG. A $k$-truncated R-vine representation of $\GGG$ is an R-vine tree sequence $\VVV\left(\GGG\right)= \left(T_1,\dots,T_{d-1}\right)$ such that $T_{k+1},\dots,T_{d-1}$ contain edges $j\left(e\right),\ell\left(e\right)|D\left(e\right)$ where $\condindep{j\left(e\right)}{\ell\left(e\right)}{D\left(e\right)}$ by $\GGG$.
\end{Definition}
We first consider the case of representing Markov-Trees, i.\,e.\ $1$-DAGs. Afterwards, the representation of general $k$-DAGs for $k \ge 2$ is evaluated.

\subsection{Representing Markov Trees as 1-truncated R-vines}\label{subsec:representingmarkovtrees}
\begin{Proposition}[Representing Markov Trees]\label{prop:dagk1}
	Let $\GGG = \left(V,E\right)$ be a $1$-DAG. There exists a $1$-truncated R-vine representation $\VVV\left(\GGG\right)$ of $\GGG$. If $\left|E\right| = d-1$, $T_1 = \GGG^s = \GGG^m$.
\end{Proposition}

See Appendix \ref{subsec:appendix_proof_prop_markovtree} for the proof and Appendix \ref{sec:appendix_algorithms} for an implementation of the algorithm \texttt{RepresentMarkovTreeRVine}. Next, we consider the general case for $k$-DAGs with $k \geq 2$.

\subsection{Representing $\mathbf{k}$-DAGs as $\mathbf{k}$-truncated R-vines under sufficient conditions}\label{subsec:mainresult}
First, we introduce the assumptions of our main theorem and their interpretation. Next, the proof follows with some illustrations. Let $\GGG = \left(V,E\right)$ be an arbitrary $k$-DAG. We impose under which assumptions an incomplete R-vine tree sequence $\left(T_1=\left(V,E_1\right),\ldots,T_k=\left(V_k,E_k\right)\right)$ is part of a $k$-truncated R-vine representation $\VVV\left(\GGG\right)$ of $\GGG$.
\begin{asu}\label{ass:1}
	For all $v,w \in V$ with $w \in \parents\left(v\right)$, there exists an $i \in \left\{1,\dots,k\right\}$ and $e \in E_i$ such that $j(e) = v,\ k(e) = w$. Here, $\parents\left(v\right)$ is specified by the DAG $\GGG$.
\end{asu}
\begin{asu}\label{ass:2}
	The main diagonal of the R-vine matrix $M$ of $T_1,\ldots,T_k$ can be written as decreasing topological ordering of the DAG $\GGG$, $\left\{\eta^{-1}\left(d\right),\ldots,\eta^{-1}\left(1\right)\right\}$ from the top left to bottom right.
\end{asu}
We illustrate such an R-vine satisfying \ref{ass:1} and \ref{ass:2} by the Examples \ref{ex:exdag1} and \ref{ex:exvine1}.
\begin{example}[Example \ref{ex:exdag1} cont.]\label{ex:exdag1contproof}
	Denote $\parents\left(v\right) = \left\{w_1^v,w_2^v\right\} \ \mbox{for} \ v \in \left\{4,5,3\right\}$ and $w_1^6 = \parents\left(6\right),\ w_1^2=\parents\left(2\right)$. The values $w_1^v,w_2^v$ of $M$ for each $v \in V$ are given in Table \ref{table:exdag1}. The corresponding R-vine can be seen in Figure \ref{fig:exvine1:1} of Example \ref{ex:exvine1}, $T_1$ and $T_2$.
	\begin{equation*}
		M=\left(
		\begin{array}{cccccc}
			4&&&&&\\
			&5&&&&\\
			&&3&&&\\
			&&&6&&\\
			w_2^4& w_2^5& w_2^7 & 1 & 2 & \\
			w_1^4& w_1^5& w_1^7 & w_1^6 & w_1^2 & 1
		\end{array}
		\right)
	\end{equation*}
\end{example}
\ref{ass:1} links each conditioned set in an edge in one of the first $k$ R-vine trees to an arrow in the DAG $\GGG$. We have seen this property in the representation of R-vines as DAGs in Section \ref{subsec:representationrvinesdags}. Note that in a (not truncated) R-vine, each pair $j\left(e\right),\ell\left(e\right) \in 1,\dots,d$ occurs exactly once as conditioned set, see \citet[p.\ 96]{KurowickaCooke2006}. \ref{ass:2} maps the topological ordering of $\GGG$ onto the conditioned sets of the R-vine tree such that
\begin{equation}\label{eq:ordering}
	j\left(e\right) \not>_\GGG \ell\left(e\right) \mbox{ for each } e \in E_1,\ldots,E_{d-1}.
\end{equation}
This can be seen as for a column $p$, the elements $M_{p+1,p},\dots,M_{d,p}$ must occur as a diagonal element to the right of $p$, i.\,e.\ as diagonal entries in a column $p+1,\ldots,d$. By definition of topological orderings, we obtain \eqref{eq:ordering}. To interpret \ref{ass:2}, recall that in a DAG we have $\condindep{v}{\nondescendants\left(v\right)}{\parents\left(v\right)}$. For higher R-vine trees $T_{k+1},\dots,T_{d-1}$ we want to truncate, \ref{ass:1} assures that all parents $\parents\left(v\right)$ are in the conditioning set for these trees. \ref{ass:2} gives us that only pairs of $v,w$ for $w \in \nondescendants\left(v\right)$ are in the conditioned sets in these trees. This holds true since the later a node occurs in the topological ordering, the more non-descendants it has. Thus, \ref{ass:2} maps the structure of DAG $\GGG$ and the R-vine $\VVV\left(\GGG\right)$. 
\begin{Theorem}[Representing DAGs as truncated R-vines]\label{thm:transformation}
	Let $\GGG = \left(V,E\right)$ be a $k$-DAG. If there exists an incomplete R-vine tree sequence $\VVV\left(\GGG\right)=\left( T_1=\left(V,E_1\right),\ldots,T_k=\left(V_k,E_k\right)\right)$ such that \ref{ass:1} and \ref{ass:2} hold, then $\VVV\left(\GGG\right)$ can be completed with trees $T_{k+1},\ldots,T_{d-1}$ which only contain independence copulas. In particular, these independence pair copulas encode conditional independences derived from the $k$-DAG $\GGG$ by the local directed Markov property.
\end{Theorem}
The main benefit now is that we can use the R-vine structure instead of the DAG structure, which is most often linked to the multivariate Gaussian distribution. For the proof, we first present two lemmas. These and the proof itself will be continuously illustrated.

\begin{Lemma}[]\label{lemma:prooflemma1}
	Let $\GGG$ be a $k$-DAG and $T_1,\ldots,T_k$ an R-vine tree sequence satisfying \ref{ass:1} and \ref{ass:2}. For each $j\left(e\right),\ell\left(e\right)|D\left(e\right) \mbox{ with } e \in E_{k+1},\ldots,E_{d-1}$, we have $\ell\left(e\right) \in \nondescendants\left(j\left(e\right)\right)$.
\end{Lemma}
\begin{proof}
	Consider an arbitrary edge $j\left(e\right),\ell\left(e\right)|D\left(e\right) \mbox{ for } e \in E_{k+1},\ldots,E_{d-1}.$
	We have $\ell\left(e\right) \notin \parents\left(j\left(e\right)\right)$, since conditioned sets in an R-vine tree sequence are unique and all conditioned sets of the form $j\left(e\right),\ell\left(e\right)$ with $\ell\left(e\right) \in \parents\left(j\left(e\right)\right)$ occurred already in the first $k$ trees by \ref{ass:1}. Additionally, $\ell\left(e\right) \notin \descendants\left(j\left(e\right)\right)$, since otherwise would violate \ref{ass:2}, as $\ell\left(e\right) >_\GGG j\left(e\right)$. Finally, $\ell\left(e\right) \neq j\left(e\right)$, since the two elements of a conditioned set must be distinct. Thus, we have $\ell\left(e\right) \notin \left(\parents\left(j\left(e\right)\right) \dotcup \descendants\left(j\left(e\right)\right) \dotcup j\left(e\right)\right) = V\setminus \nondescendants\left(j\left(e\right)\right)$ and hence $\ell\left(e\right) \in \nondescendants\left(j\left(e\right)\right)$.
\end{proof}
\begin{example}[Example \ref{ex:exdag1contproof} cont.]\label{ex:exdag1contproof1}
	Illustrating Lemma \ref{lemma:prooflemma1}, consider the R-vine matrix $M$ of Example \ref{ex:exdag1contproof} and column $3$. To complete $M$, we need to fill in e.\,g.\ $M_{4,3}$. Valid entries can come from the main diagonal of $M$ right of $3$, i.\,e.\ $\left\{M_{4,4},M_{5,5},M_{6,6}\right\}=\left\{6,2,1\right\}$. Since $\parents\left(3\right)= \left\{2,6\right\}$ and by \ref{ass:1}, the edges in the first two R-vine trees are $\left\{3,6\right\}$ and $3,2|6$, the only remaining entry is $M_{4,3}=1$. This can only be a non-descendant of $3$ because of \ref{ass:2}.
\end{example}

\begin{Lemma}[]\label{lemma:prooflemma2}
	Let $\GGG$ be a $k$-DAG and $T_1,\ldots,T_k$ an R-vine tree sequence satisfying \ref{ass:1} and \ref{ass:2}. For each $j\left(e\right),\ell\left(e\right)|D\left(e\right) \mbox{ with } e \in E_{k+1},\ldots,E_{d-1}$ we have $D\left(e\right) \subseteq \left\{\parents\left(j\left(e\right)\right) \cup \nondescendants\left(j\left(e\right)\right)\right\}$.
\end{Lemma}
\begin{proof}
	Consider $j\left(e^\prime\right),k\left(e^\prime\right)|D\left(e^\prime\right)$ for $e^\prime \in E_{k+1}$. We have the following two cases.\\
	\textbf{First case:} $\left|\parents\left(j\left(e^\prime\right)\right)\right| = k$. All parents of $j\left(e^\prime\right)$ occurred in the conditioned set of edges together with $j\left(e^\prime\right)$ in the first $k$ R-vine trees. Hence, $\parents\left(j\left(e^\prime\right)\right) = D\left(e^\prime\right)$ and $\left|D\left(e^\prime\right)\right|=k$.\\
	\textbf{Second case:} $\left|\parents\left(j\left(e^\prime\right)\right)\right| \eqqcolon k_{j\left(e^\prime\right)} < k$. Similarly to the first case, we conclude $\parents\left(j\left(e^\prime\right)\right) \subset D\left(e^\prime\right)$. Let $D\left(e^\prime\right) \setminus \parents\left(j\left(e^\prime\right)\right) = D_1$ with $\left|D_1\left(e^\prime\right)\right| = k-k_{j\left(e^\prime\right)} > 0$. To obtain the elements of $D\left(e^\prime\right)$, recall \ref{ass:2} and consider the column of the R-vine matrix $M$ in which $j\left(e^\prime\right)$ is in the diagonal, say column $p$. The entries $\left\{M_{d-k,p},\ldots,M_{d,p}\right\}$ describe the elements which occurred in conditioned sets together with $j\left(e^\prime\right)$ in the first $k$ trees. As these entries may only be taken from the right of $M_{p,p}=j\left(e^\prime\right)$, these must be non-descendants of $j\left(e^\prime\right)$.
	To conclude the statement for the R-vine trees $T_{k+2},\ldots,T_{d-1}$, we use an inductive argument. Let $e^{\prime\prime} \in E_{k+2}$ and $j\left(e^{\prime\prime}\right)$ is in the diagonal of the R-vine matrix $M$ in column $p$. Then, for the conditioning set of $e^{\prime\prime}$ we have $D\left(e^{\prime\prime}\right) = M_{d-k-1,p} \dotcup \left\{M_{d-k,p},\ldots,M_{d,p}\right\}$. For the set $\left\{M_{d-k,p},\ldots,M_{d,p}\right\}$ we have shown that it can only consist of parents and non-descendants of $j\left(e^{\prime\prime}\right)$. As $M_{d-k-1,p}$ can only have a value occurring in the main diagonal of the R-vine matrix to the right of column $p$, it must be a non-descendant of $j\left(e^{\prime\prime}\right)$. The same argument holds inductively for the trees $T_{k+3},\ldots,T_{d-1}$. Thus, we have shown that for each edge $j\left(e\right),\ell\left(e\right)|D\left(e\right)$ with $e \in E_{k+1},\ldots,E_{d-1}$ we have $D\left(e\right) \subseteq \left\{\parents\left(j\left(e\right)\right) \cup \nondescendants\left(j\left(e\right)\right)\right\}$.
\end{proof}

\begin{example}[Example \ref{ex:exdag1contproof1} cont.]\label{ex:exdag1contproof2}
	Consider the first column of $M$ with $M_{1,1}=4$. Since $\parents\left(4\right)=\left\{2,5\right\}$, $\left\{4,5\right\} \in E_1$ and $4,2|5 \in E_2$, independently of the values in $M_{2,1},\ldots,M_{4,1}$, $\parents\left(4\right)=\left\{2,5\right\}$ is in the conditioning set for each of these edges. There will be more nodes in the conditioning set but $\left\{2,5\right\}$ in higher trees, yet, these are non-descendants of $4$ by \ref{ass:2}.
\end{example}

We will now conclude with the proof of Theorem \ref{thm:transformation} using the Lemmas \ref{lemma:prooflemma1} and \ref{lemma:prooflemma2}.
\begin{proof}
	Abbreviate $j_e \equiv j\left(e\right),\ k_e \equiv \ell\left(e\right),\ D_e \equiv D\left(e\right)$ and set $j_e,k_e|D_e \equiv j\left(e\right),\ell\left(e\right)|D\left(e\right)$ with $e \in E_{k+1},\ldots,E_{d-1}$ arbitrary but fixed. For the node $j_e$ in the DAG $\GGG$ we have by the directed local Markov property \eqref{eq:directedlocalmarkov} that $\condindep{j_e}{\nondescendants\left(j_e\right)}{\parents\left(j_e\right)}$ and thus with Lemma \ref{lemma:prooflemma1},
	\begin{equation}\label{eq:proof1}
		\condindep{j_e}{k_e \dotcup \left(\nondescendants\left(j_e\right) \setminus k_e\right)}{\parents\left(j_e\right)}.
	\end{equation}
	Set $\widehat{\nondescendants} \left(j_e\right) \coloneqq D_e \setminus \parents\left(j_e\right) \mbox{ with } \widehat{\nondescendants} \left(j_e\right) \subseteq \nondescendants\left(j_e\right) \mbox{ by Lemma } \ref{lemma:prooflemma2}$,
	plug it into \eqref{eq:proof1} obtaining
	\begin{equation}\label{eq:proof2}
		\condindep{j_e}{\left(k_e \dotcup \left(\big(\nondescendants\left(j_e\right) \setminus k_e\big) \setminus \widehat{\nondescendants} \left(j_e\right)\right) \dotcup \widehat{\nondescendants}\left(j_e\right)\right)}{\parents\left(j_e\right)},
	\end{equation}
	exploiting $k_e \cap \widehat{\nondescendants}\left(j_e\right) = \emptyset$, i.\,e.\ a node can not be part of the conditioning and the conditioned set of the same edge. Applying Proposition \ref{prop:CI_lemmas} on \eqref{eq:proof2} yields $\condindep{j_e}{k_e \dotcup \widehat{\nondescendants} \left(j_e\right)}{\parents\left(j_e\right)}$ by dropping $\left(\left(\nondescendants\left(j_e\right) \setminus k_e\right) \setminus \widehat{\nondescendants} \left(j_e\right)\right)$ in \eqref{eq:proof2}. $k_e \dotcup \widehat{\nondescendants} \left(j_e\right)$ is a disjoint union on which Proposition \ref{prop:CI_lemmas} can be applied to conclude $\condindep{j_e}{k_e}{\parents\left(j_e\right) \dotcup \widehat{\nondescendants} \left(j_e\right)}$.
	By definition of $\widehat{\nondescendants}\left(j_e\right)$, we have $D_e = \parents\left(j_e\right) \dotcup \widehat{\nondescendants}\left(j_e\right)$ and obtain the final result $\condindep{j_e}{k_e}{D_e} \mbox{ for } e \in E_{k+1},\ldots,E_{d-1}.$ Since each edge is assigned a pair copula density, we can now choose the independence copula density $c^{\perp}$ for these edges in $E_{k+1},\ldots,E_{d-1}$ backed by the conditional independence properties of the DAG. The resulting R-vine is thus a $k$-truncated R-vine.
\end{proof}
\begin{example}[Example \ref{ex:exdag1contproof2} cont.]\label{ex:exdag3contproof}
	We illustrate Theorem \ref{thm:transformation} using the previous Examples \ref{ex:exdag1contproof1} and \ref{ex:exdag1contproof2}. Consider column $1$ of $M$ and edge $4,3|256 \in E_4$. From the conditional independence $\condindep{4}{1,3,6}{2,5}$ obtained from the DAG $\GGG$, we select the non-descendants of $4$ to neglect, i.\ e.\ $1$, to yield $\condindep{4}{3,6}{2,5}$ by application of Proposition \ref{prop:CI_lemmas} and finally $\condindep{4}{3}{2,5,6}$ by second application of Proposition \ref{prop:CI_lemmas}.
\end{example}
Computing an R-vine representation $\VVV\left(\GGG\right)$ of an arbitrary $k$-DAG $\GGG$ is a complex combinatorial problem and the existence of an incomplete R-vine tree sequence satisfying \ref{ass:1} and \ref{ass:2} is not clear. We first show classes of $k$-DAGs where we can prove the existence of their R-vine representations. Afterwards, we introduce necessary conditions for the existence of an $k$-truncated R-vine representation. 

\begin{Corollary}[$k$-DAGs with R-vine representation]\label{cor:alwaysrepresentabledags}
	Let $\GGG= \left(V,E\right)$ be a $k$-DAG such that $V=\left\{v_1,\dots,v_d\right\}$ is an increasing topological ordering of $\GGG$. If, for all $v_i \in V$, $i=1,\dots,d$, we have $\parents\left(v_i\right) \subseteq \left\{v_{i-k},\dots,v_{i-1}\right\}$ or $\parents\left(v_i\right) \subseteq \left\{v_1,\dots,v_k\right\}$, an R-vine representation $\VVV\left(\GGG\right)$ of $\GGG$ exists.
\end{Corollary}
\begin{proof}
	Let $\parents\left(v_i\right) \subseteq \left\{v_{i-k},\dots,v_{i-1}\right\}$. The R-vine representation $\VVV\left(\GGG\right)$ is given by $T_1$ being path from $v_1$ to $v_d$ according to the topological ordering of $\GGG$, i.\,e.\ a \textit{D-vine}. Because of the proximity condition, $T_2,\dots,T_{d-1}$ are uniquely determined by $T_1$. In tree $T_j$, the edges have the form $v_i,v_{i-j}|v_{i-j+1},\dots,v_{i-1}$ and each conditioned set in the first $k$ R-vine trees represents an arrow of $\GGG$, satisfying \ref{ass:1}. \ref{ass:2} also holds since in a D-vine, the main diagonal of the R-vine matrix can be written as ordering of the path $T_1$. If $\parents\left(v_i\right) \subseteq \left\{v_{1},\dots,v_{k}\right\}$, $T_1$ is given a \textit{star} with central node $v_1$. $T_2$ is a star with central node $\left\{v_1{,}v_2\right\}$ and so on, giving rise to a \textit{C-vine}. In tree $T_{k-j}$, the edges have the form $v_i,v_{k-j}|v_{1},\dots,v_{k-j-1}$ for $i \ge k$, satisfying \ref{ass:1}. The main diagonal of the R-vine matrix of a C-vine is ordered according to the central nodes in the C-vine, satisfying \ref{ass:2}. To both, Theorem \ref{thm:transformation} applies. Examples of a $2$-DAG with D-vine and a $k$-DAG with C-vine representation are shown in Figure \ref{fig:exdagdvine}.
\end{proof}
\vspace{-0.25cm}
\begin{figure}[h]
	\centering
	\includegraphics[width=0.3\textwidth, trim={0.2cm 0.2cm 0.2cm 0.2cm},clip]{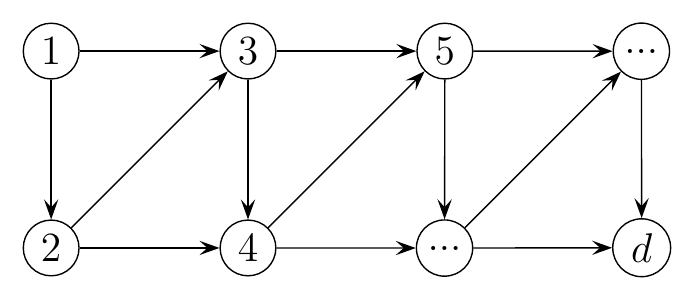}\hspace{1cm}
	\includegraphics[width=0.35\textwidth, trim={0.1cm 0.1cm 0.1cm 0.1cm},clip]{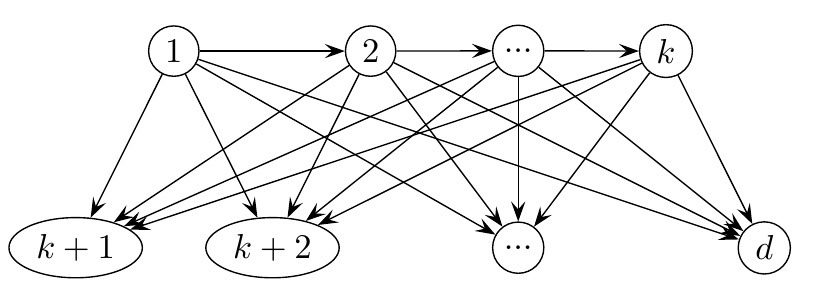}
	\caption{Examples of DAGs with D- and C-vine representation.}
	\label{fig:exdagdvine}
\end{figure}
\vspace{-0.25cm}
We now present necessary conditions for the first tree of an R-vine representation. It is of particular importance as it influences all higher order trees by the proximity condition.
\subsection{Necessary conditions for Theorem \ref{thm:transformation}}\label{subsec:necessaryconditions}
\begin{Proposition}[Necessary conditions]\label{prop:necessaryconditiondag}
	Consider a $k$-DAG $\GGG = \left(V,E\right)$ and the sets
	\begin{equation*}
		V^v \coloneqq \left\{v,\parents\left(v\right)\right\}=\left\{v,w_1^v,\ldots,w_{k_v}^v\right\}, v \in V.
	\end{equation*}
	Assume there exists an R-vine representation $\VVV\left(\GGG\right)= \left(T_1,\dots,T_k\right)$ such that \ref{ass:1} and \ref{ass:2} hold. For $v \in V$, denote the induced subgraphs $T_1^v \coloneqq \left(V^v,E^v\right) \subseteq T_1$ of $T_1=\left(V,E_1\right)$ on $V^v$. Thus, in $E^v$ are all edges in $T_1$ between nodes of $V^v$. Then, $T_1=\left(V,E_1\right)$ must be such that
	\begin{enumerate}[label={(\roman*})]
		\item for all $v \in V$ with $\parents\left(v\right) = k$, $T_1^v$ contains a path involving all nodes of $V^v$, \label{eq:neccond:1}
		\item the union of the induced subgraphs $\bigcup_{i \in \III}~T_1^{v_i} \coloneqq \left(\bigcup_{i \in \III}~V^{v_i},\bigcup_{i \in \III}~E^{v_i}\right) \subseteq T_1$ is acyclic for $\III\coloneqq \left\{i \in V:\left|\parents\left(v_i\right)\right| = k_{v_i} = k\right\}$. \label{eq:neccond2}
	\end{enumerate}
\end{Proposition}

\begin{proof}
	To show $\left(i\right)$ assume $\VVV\left(\GGG\right)= \left(T_1,\dots,T_k\right)$ satisfies \ref{ass:1} and \ref{ass:2}. Choose $v \in V$ with $k_v=k$ arbitrary but fixed. Order the set $\parents\left(v\right)= \left\{w_1^v,\ldots,w_{k}^v\right\}$ such that $v,w_i^v$ is the conditioned set of an edge $e \in E_i,\ i=1,\ldots,k$, ensured by \ref{ass:2}. Then, by the proof of Theorem \ref{thm:transformation}, each edge $e \in E_i,\ i=1,\ldots,k$ corresponding to $v \in V$ must have the form $v{,}w_i^v|w_1^v,\dots,w_{i-1}^v$. By the set formalism, see Example \ref{ex:exvine1}, and the proximity condition, see  \ref{eq:proximitycondition} on page \pageref{eq:proximitycondition}, we have  $\left\{\left\{v{,}w_1^v\right\},\left\{w_1^v{,}w_2^v\right\}\right\} \in E_2$ requiring $\left\{v{,}w_1^v\right\} \in E_1$ and $\left\{w_1^v{,}w_2^v\right\} \in E_1$. For $\left\{\left\{\left\{v,w_1^v\right\},\left\{w_1^v,w_2^v\right\}\right\},\left\{\left\{w_1^v,w_2^v\right\},\left\{w_2^v,w_3^v\right\}\right\}\right\} \in E_3$ we can conclude in a first step $\left\{\left\{w_1^v{,}w_2^v\right\},\left\{w_2^v{,}w_3^v\right\}\right\} \in E_2$ and in a second step $\left\{w_2^v,w_3^v\right\} \in E_1$. This can be extended to $E_k$ and yields $\left\{w_i^v,w_{i+1}^v\right\} \in E_1$ for $i = 1,\ldots,k-1$. Thus, $v$ and its parents, i.\,e.\ $V^v$ represent a path in $T_1^v$. Showing $\left(ii\right)$, for each $i \in \III$ the graph $T_1^{v_i}$ is a subgraph of $T_1$ by $\left(i\right)$. Thus, the union of $T_1^{v_i}$ over all $i \in \III$ must be a subgraph of $T_1$. Since $T_1$ is a tree, it is acyclic, hence, each of its subgraphs must be, and so the graph in $\left(ii\right)$.
\end{proof}
Whereas the proof of $\left(i\right)$ is a direct consequence of the \textit{proximity condition}, the proof of $\left(ii\right)$ is less intuitive. We illustrate this property.

\begin{example}[DAG in $6$ dimensions]\label{ex:exdagcycles}
	Consider the DAG $\GGG_2$ in Figure \ref{fig:exdagcycles:ex2dag1}. By Proposition \ref{prop:necessaryconditiondag}, we need to find an R-vine tree $T_1=\left(V,E_1\right)$ such that the induced subgraphs $T_1^v = \left(V^v,E^v\right) \subseteq T_1$ contain a paths involving all nodes of $V^v$ for $V^v \in \left\{\left\{4,1,2\right\},\left\{5,1,3\right\},\left\{6,2,3\right\}\right\}$. 
	\begin{figure}[h]
		\centering
		\includegraphics[width=0.2\textwidth, trim={0.2cm 0.2cm 0.2cm 0.2cm},clip]{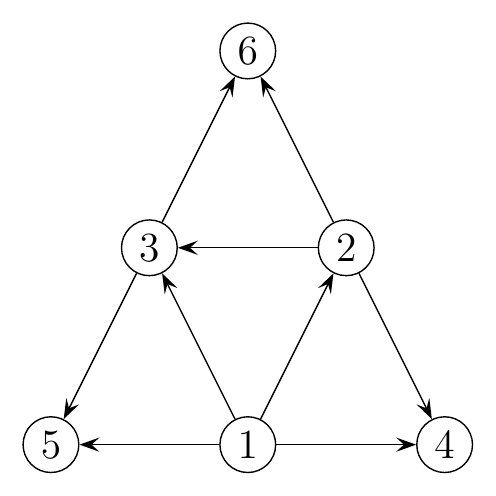}
		\caption{Example \ref{ex:exdagcycles}: DAG $\GGG_2$}
		\label{fig:exdagcycles:ex2dag1}
	\end{figure}
	This is not possible. If it would be, use the path $T_1^4$ from $1$ to $2$, $T_1^6$ from $2$ to $3$ and finally $T_1^5$ from $3$ to $1$. However, this creates a cycle and $T_1$ as such can not be a tree. Yet, removing any edge which closes the cycle yields an induced subgraph which is no longer connected, i.\,e.\ a path. Thus, the DAG $\GGG_2$ can not be represented by a $2$-truncated R-vine. $3$-truncated R-vines are possible which are shown in Appendix \ref{sec:appendix_examples}, Example \ref{ex:exdagcycles:cont.}.
\end{example}

Based on Proposition \ref{prop:necessaryconditiondag}, we are given an intuition how to construct an admissible first R-vine tree $T_1$ of $\VVV\left(\GGG\right)$ for a DAG $\GGG$. Moreover, it also yields a best possible truncation level $k^\prime > k$ for which a $k^\prime$-truncated R-vine representation exists.

\begin{Corollary}[Best possible truncation level $k^\prime$]\label{cor:lowerboundkprime}
	Consider a $k$-DAG $\GGG = \left(V,E\right)$. Let $T_1=\left(V,E_1\right)$ be a tree and for each $v,w \in V$ let $\delta_v^w$ be the length of the unique path from $v$ to $w$ in $T_1$. If $T_1$ is extended by successive R-vine trees $T_i, i \in \left\{2,\dots,d-1\right\}$, then the truncation level $k^\prime$ can be bound from below by
	\begin{equation*}
		k^\prime \ge \max_{v \in V}\max_{w \in \parents\left(v\right)}\delta_v^w.
	\end{equation*}
\end{Corollary}
An example and the proof, using the proximity condition, d-separation and the graphical structure of $T_1$ is given in Appendix \ref{subsec:appendix_proof_corollarylowerbound}. \ref{ass:1} and \ref{ass:2} are strong assumptions and hence only rarely satisfied for arbitrary DAGs. This gives rise to a heuristic approach for arbitrary $k$-DAGs to find a sparse R-vine representation exploiting their conditional independences.

\subsection{Representing $\mathbf{k}$-DAGs as sparse R-vines}\label{subsec:generaldag2}
Our goal is to find an R-vine representation $\VVV\left(\GGG\right)$ of an arbitrary $k$-DAG $\GGG$ for $k \ge 2$. For the first R-vine tree $T_1$, we have $d^{d-2}$ candidates. Considering all these and checking Proposition \ref{prop:necessaryconditiondag} is not feasible. Additionally, \ref{ass:2} is hard to check upfront since it is not fully understood how a certain R-vine matrix diagonal relates to specific R-vines. Fixing the main diagonal may thus result in suboptimal models. Hence, Theorem \ref{thm:transformation} can not be applied directly. Denote $\GGG_k$ a $k$-DAG. By \ref{ass:1}, arrows in $\GGG_k$ shall be modelled as conditioned sets in R-vine trees $T_i$ for $i \in \left\{1,\dots,k\right\}$. Yet, for $k \ge 2$, there may be up to $kd - \left(k\left(k+1\right)\right)/2$ candidate edges for $T_1$ which is limited to $d-1$ edges. Hence, it is crucial to find the most important arrows of $\GGG_k$ for $T_1$. An heuristic measure for the importance of an arrow $v \rightarrow w$ in $\GGG_k$, fitted on data, is how often the arrow $v \rightarrow w$ exists in $1$,\dots,$k-1$-DAGs $\GGG_1,\dots,\GGG_{k-1}$, also fitted on data. However, also an arrow $w \rightarrow v$ is possible. Since R-vines are undirected graphical models, we neglect the orientation of arrows in the DAGs by considering their skeletons. Thus, for each edge $\left(v,w\right)$ in the skeleton $\GGG_k^s$ of the DAG $\GGG_k$ we estimate DAGs $\GGG_1,\dots,\GGG_{k-1}$ on the data, obtain their skeletons $\GGG_1^s,\dots,\GGG_{k-1}^s$ and count how often the edge $\left(v,w\right)$ exists in these graphs. An edge $\left(v,w\right) \in \GGG_i^s$ might be more important than an edge $\left(v,w\right) \in \GGG_j^s$ with $i < j$, which we describe by a non-increasing function of the maximal number of parents $g(i)$. Formally, consider $i$-DAGs $\GGG_i$ for $i=1,\dots,k$ estimated on data. Denote $\GGG_i^s=\left(V,E_i^s\right)$ the skeleton of $\GGG_i$ for $i=1,\dots,k$ and define an undirected graph $\HHH=\left(V,E_1^{\HHH}\right) \coloneqq \bigcup_{i=1}^k~\GGG_i^s$ with edge weights $\mu_1$ for $\left(v,w\right) \in E^{\HHH}_1$ given by
\begin{equation}\label{eq:edgeweight}
	\mu_1\left(v,w\right) \coloneqq \sum_{i=1}^{k}~g\left(i\right)\mathds{1}_{\left\{\left(v,w\right) \in E_i^s\right\}}\left(v,w\right),
\end{equation}
with $g\left(i\right) > 0$ non-increasing for $i=1,\dots,k$. In the remainder, $g\left(i\right) \equiv 1$. To our knowledge, this approach has not been used before. On $\HHH$, find a maximum spanning tree $T_1$ by, e.\,g.\ \citet{Prim1957}, maximizing the sum of weights $\mu_1$. The higher order trees are built iteratively. First, define a full graph $T_2=\left(V_2,E_2\right)$ on $V_2=E_1$ and delete each edge in $E_2$ not allowed by the proximity condition. Denote the edges with conditioned and conditioning set  $j\left(e\right),\ell\left(e\right)|D\left(e\right)$ for $e \in E_2$. Set weights for $e \in E_2$ according to
\begin{equation}\label{eq:edgeweight2}
	\mu_2\left(e\right)= \mu_1\left(j\left(e\right),\ell\left(e\right)\right) > 0 \mbox{ if } \mu_1\left(j\left(e\right),\ell\left(e\right)\right) \neq 0.
\end{equation}
Thus, $e \in E_2$ has positive weight if its conditioned set is an edge in at least one of the skeletons $\GGG_1^s,\dots,\GGG_k^s$. We can not ensure \ref{ass:2}, and thus not use the directed local Markov property as in Theorem \ref{thm:transformation}. We overcome this using \textit{d-separation}. More precisely, for $j\left(e\right),\ell\left(e\right)|D\left(e\right)$, $e \in E_i,\ 1\le i \le d-1$, we check if $j\left(e\right)$ is d-separated from $\ell\left(e\right)$ given $D\left(e\right)$ in $\GGG_k$. To facilitate conditional independence, i.\,e.\ sparsity, for $e\in E_2$ assign $\mu_0 \in \left(0,g\left(k\right)\right)$
\begin{equation*}
	\mu_2\left(e\right)= \mu_0 \mbox{ if }  j\left(e\right) \mbox{ is d-separated from } \ell\left(e\right) \mbox{ by } D\left(e\right) \mbox{ in } \GGG_k.
\end{equation*}
In the remainder, $\mu_0 \coloneqq g\left(1\right)/2 = 1/2$, i.\,e.\ it will not exceed the weight of an edge $j\left(e\right),\ell\left(e\right)|D\left(e\right)$ with $j\left(e\right) \leftrightarrow \ell\left(e\right)$ in any of the DAGs $\GGG_1,\dots,\GGG_k$ as we want to model relationships in the DAGs prioritized. All other weights are zero and a maximum spanning tree algorithm is applied on $E_2$. If an edge with weight $\mu_0$ is chosen, we can directly set the independence copula. We repeat this for $T_3,\ldots,T_{d-1}$. Since each pair of variables occurs exactly once as conditioned set in an R-vine, each weight $\mu_1$ in $\HHH$ is used exactly once. The actual truncation level $k^\prime$ is such that the R-vine trees $T_{k^\prime+1},\ldots,T_{d-1}$ contain only the independence copula. The corresponding algorithm is given in Appendix \ref{sec:appendix_algorithms}, for a toy example, see Appendix \ref{sec:appendix_toyexample}. We test it in the following simulation study and application.

\section{Simulation Study}\label{sec:simstudy}
For the next two sections, let $\mathbf{X}=\left(X_1,\dots,X_d\right) \in \mathbb{R}^d$ and define
\begin{enumerate}
	\itemsep-0.25em
	\item \textit{x-scale}: the original scale of $X_i$ with density $f_i(x_i),\ i=1,\dots,d$,
	\item \textit{u-scale} or \textit{copula-scale}: $U_i = F_i\left(X_i\right)$, $F_i$ the cdf of $X_i$ and $U_i \sim \UUU\left[0,1\right]$, $i=1,\dots,d$,
	\item \textit{z-scale}: $Z_i = \Phi^{-1}\left(U_i\right)$, $\Phi$ the cdf of $\NNN\left(0,1\right)$ thus $Z_i \sim \NNN\left(0,1\right)$, $i=1,\dots,d$.
\end{enumerate}
We show that our approach of Section \ref{subsec:generaldag2} calculates useful R-vine models in terms of goodness-of-fit in very short time. We collected data from January 1, 2000 to December 31, 2014 of the S\&P100 constituents. At the end of the observation period, still $82$ of the original $100$ stocks were in the index. For these $82$ stocks and the index we calculated daily log-returns, obtaining data in $83$ dimensions with $3772$ observations. We remove trend and seasonality off the data using ARMA-GARCH time series models with Student-t distributed residuals to obtain data on the $\textit{x-scale}$. Afterwards, we transform the residuals to the \textit{copula-scale} using their, non-parametrically estimated, empirical cumulative distribution function. For this dataset, we fitted five R-vine models using the R-package \textsf{VineCopula}, see \citet{VineCopula} using the algorithm of \citet{dissmann-etal}, introduced on p. \pageref{page:dissmann}. 
\setlength{\extrarowheight}{-.3em}
\begin{table}[h]
	\centering
	\begin{tabular}{rrrr}
		Scenario & pair copula families & truncation level & level $\alpha$  \\ 
		\hline \hline
		1 & all & - & $0.05$ \\ 
		2 & independence, t & - & $0.05$ \\
		3 & all & 4 & $0.05$ \\
		4 & all & - & $0.2$ \\
		5 & all & 4 & $0.2$ \\
		\hline
	\end{tabular}
	\caption{Parameter settings for scenarios in the simulation study.}
	\label{table:simstudy_parameters}
\end{table}
\vspace{-0.25cm}
The models were fitted with settings as shown in Table \ref{table:simstudy_parameters} such that Scenarios 3, 4 and 5 exhibit more sparsity. This is done by either imposing truncation levels or performing independence tests at level $\alpha$ while fitting the R-vines, see columns $3$ and $4$ of Table \ref{table:simstudy_parameters}. $\alpha = 0.2$ leads to many more independence copulas than $\alpha = 0.05$. From these models, $100$ replications with $1000$ data points each were simulated. For each of the simulated datasets, Dissmann's and our algorithm, see Section \ref{subsec:generaldag2}, were applied using $k$-DAGs $\GGG_k$ with $k=2,3,4$.
\begin{figure}[ht]
	\centering
	\includegraphics[width=0.24\textwidth, trim={0.1cm 0.1cm 0.1cm 0.1cm},clip]{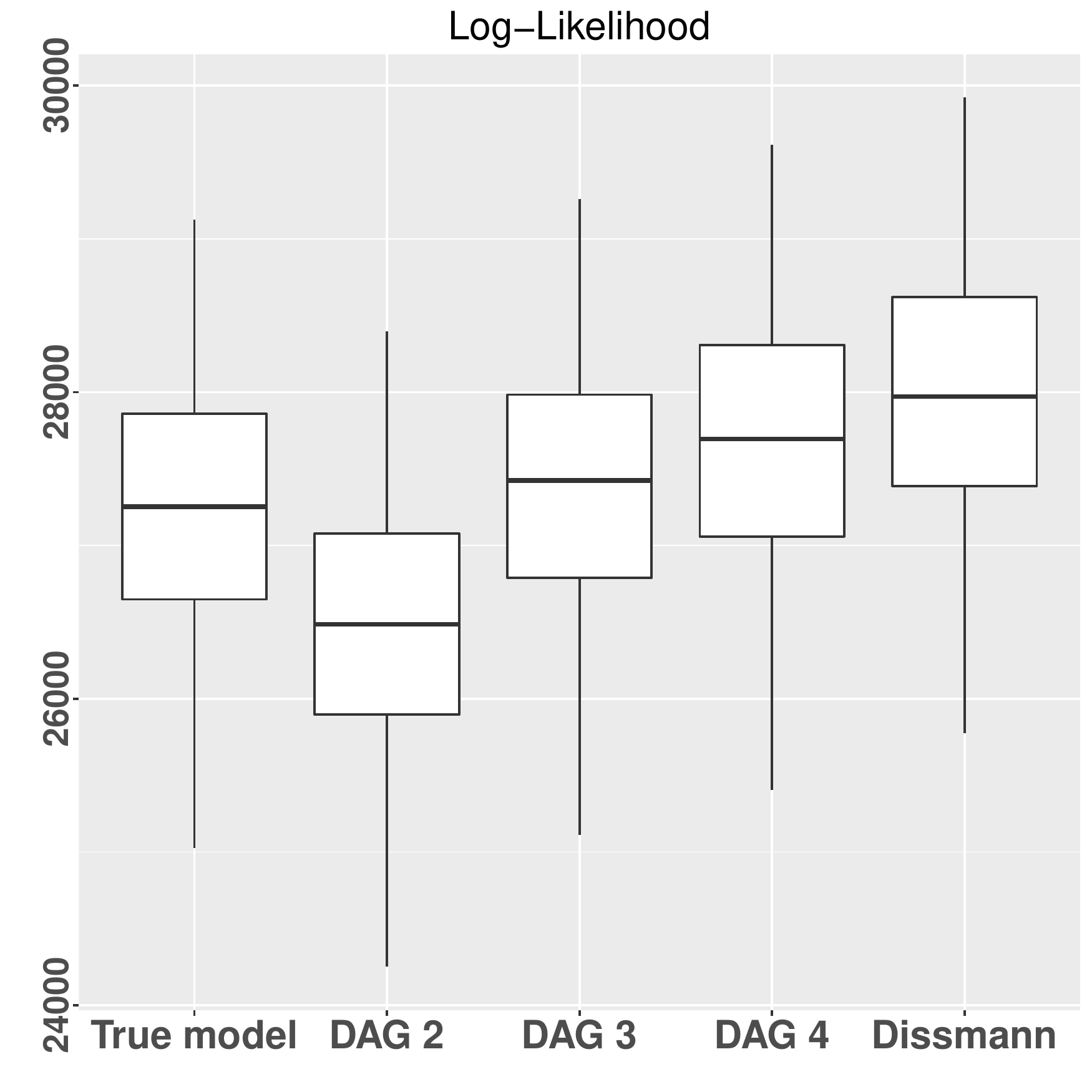}
	\includegraphics[width=0.24\textwidth, trim={0.1cm 0.1cm 0.1cm 0.1cm},clip]{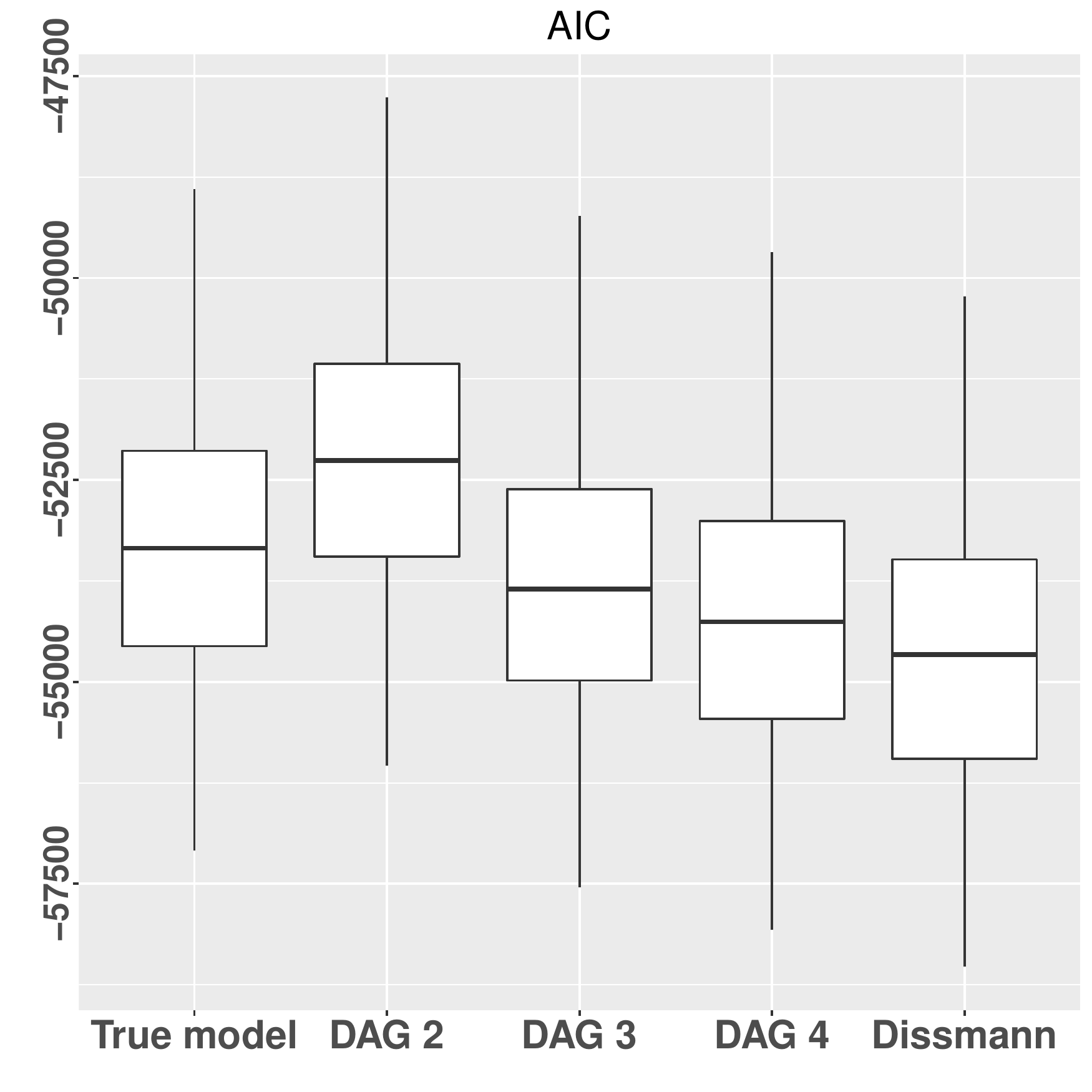}
	\includegraphics[width=0.24\textwidth, trim={0.1cm 0.1cm 0.1cm 0.1cm},clip]{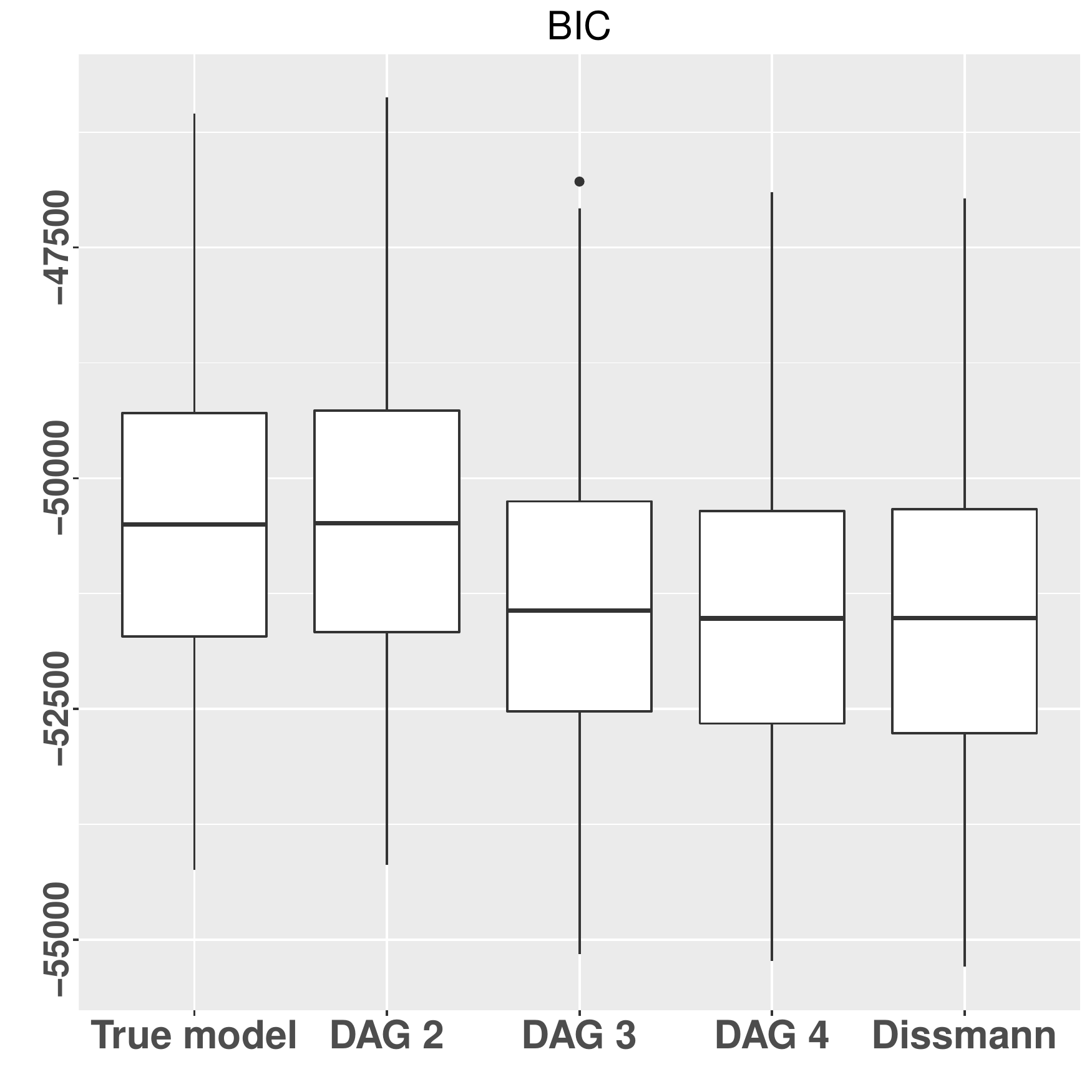}
	\includegraphics[width=0.24\textwidth, trim={0.1cm 0.1cm 0.1cm 0.1cm},clip]{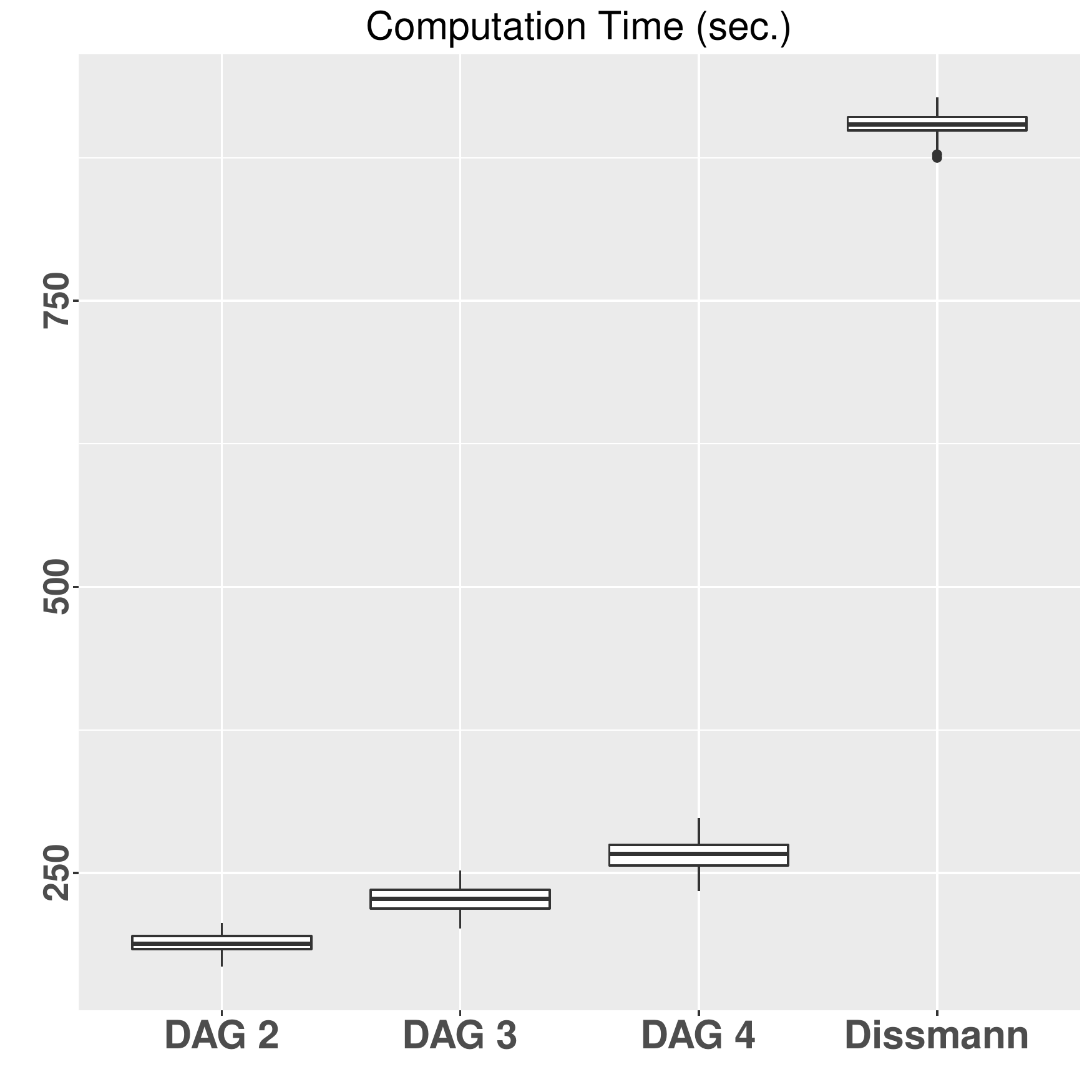}
	\caption{Scenario 5: Comparison of $k$-DAG representations for $k=2,3,4$ and Dissmann's algorithm: log-Likelihood, AIC, BIC, computation time on $100$ replications (left to right).}
	\label{fig:simstudy:results5}
\end{figure}
\vspace{-0.25cm}
We consider the results for the sparsest Scenario 5 in Figure \ref{fig:simstudy:results5}. Dissmann's algorithm achieves better results in terms of log-Likelihood and AIC, however, tends to overfit the data. For BIC, our approach using $k$-DAGs with $k=3,4$ achieves similar results as Dissmann. However the computation times are significantly shorter for our approach. The results are very similar for the other scenarios, henceforth we deferred their results in Appendix \ref{sec:appendix_simstudy}.
A second aspect is the distance of associated correlation matrices. We consider the data on the \textit{z-scale} and calculate the Kullback-Leibler divergence, see \citet{KullbackLeibler1951}. First, between the sample correlation matrix $\widehat{\Sigma}$ and the correlation matrix of $\GGG_k$, $\Sigma_{\GGG_k}$ (a). Next, we compare $\Sigma_{\GGG_k}$ and the correlation matrix of the representing R-vine model $\Sigma_{\VVV\left(\GGG_k\right)}$ (b). Finally, we compare $\widehat{\Sigma}$ and $\Sigma_{\VVV\left(\GGG_k\right)}$ and the correlation matrix of the Dissmann model $\Sigma_D$ (c).
\begin{figure}[ht]
	\centering
	\includegraphics[width=0.32\textwidth, trim={0.1cm 0.1cm 0.1cm 0.1cm},clip]{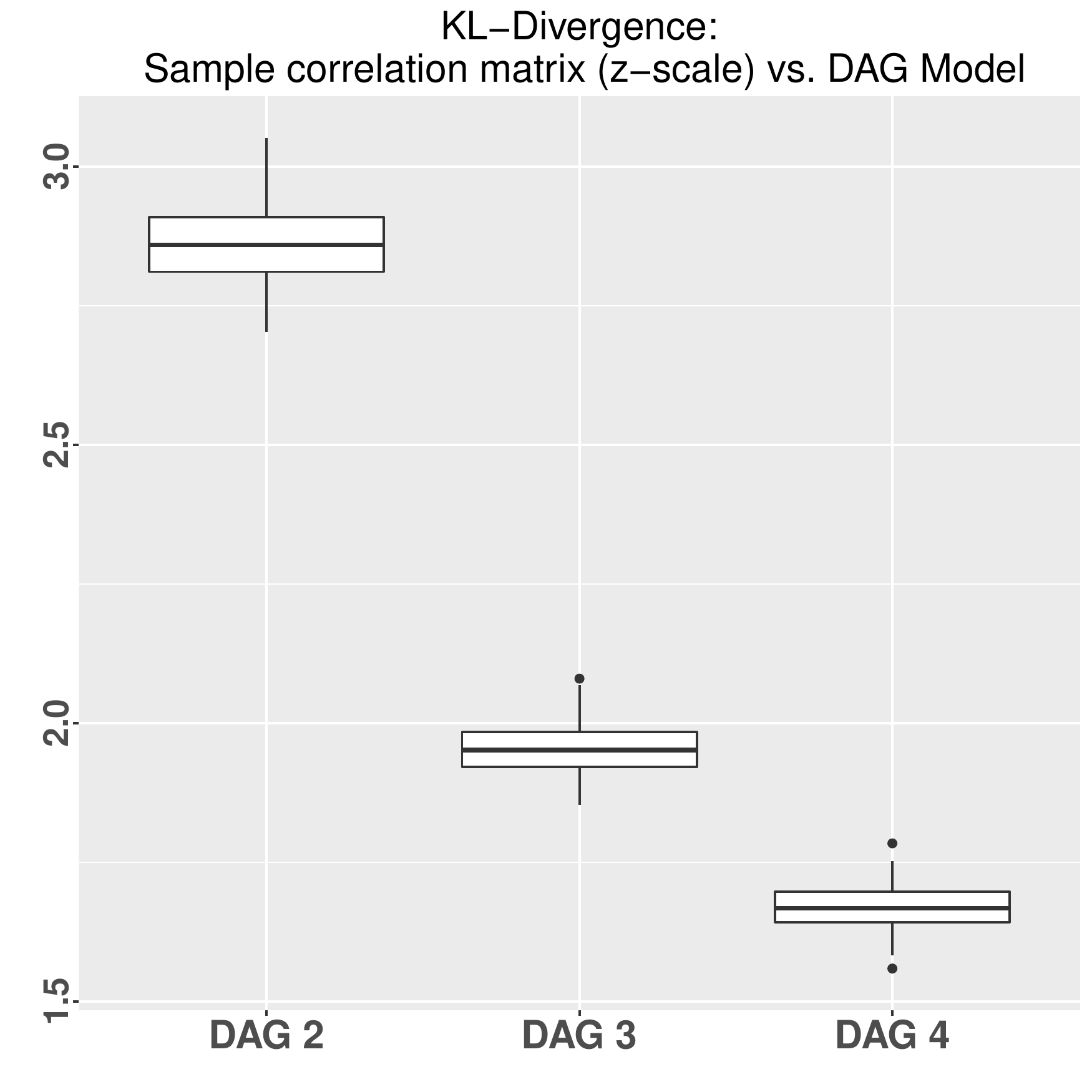}
	\includegraphics[width=0.32\textwidth, trim={0.1cm 0.1cm 0.1cm 0.1cm},clip]{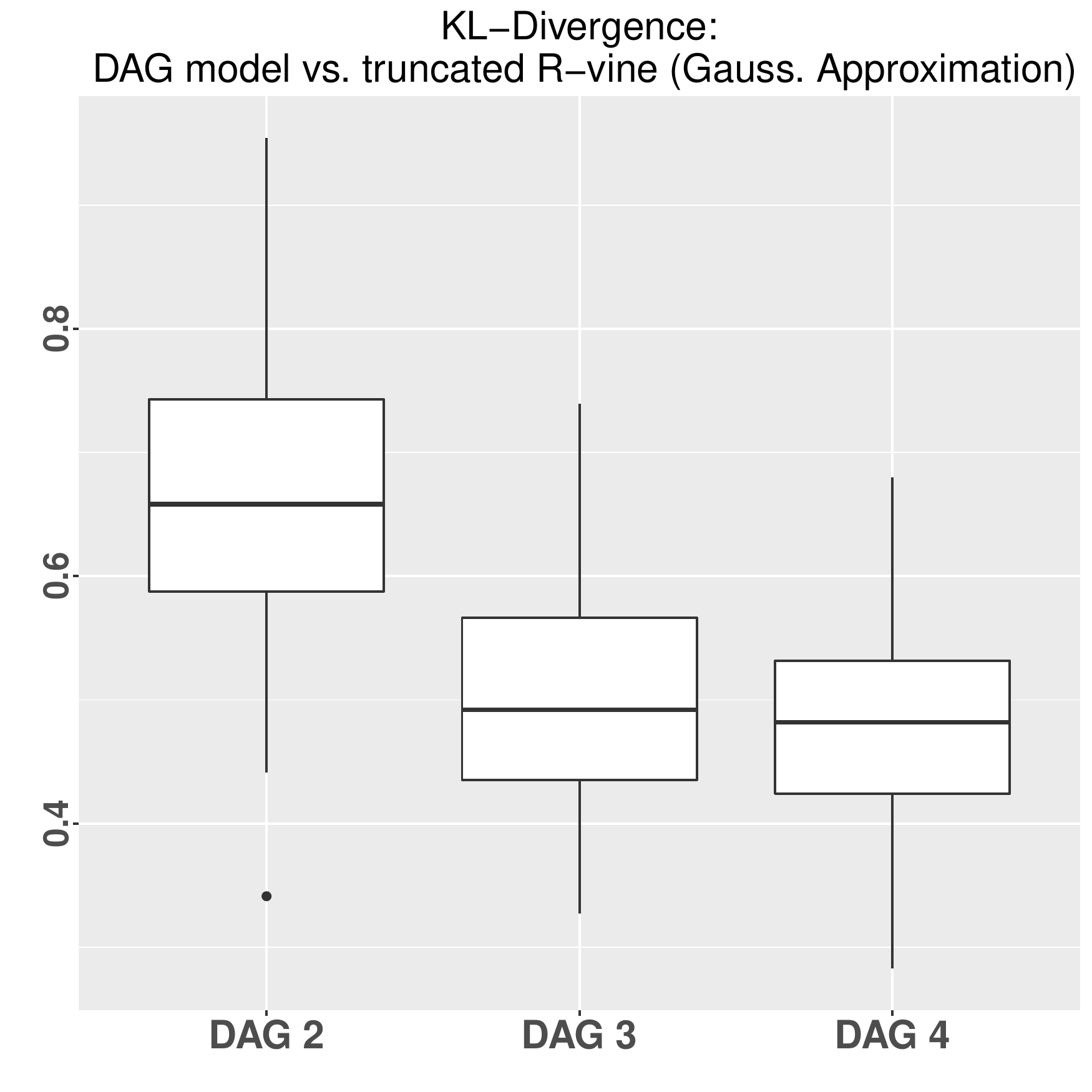}
	\includegraphics[width=0.32\textwidth, trim={0.1cm 0.1cm 0.1cm 0.1cm},clip]{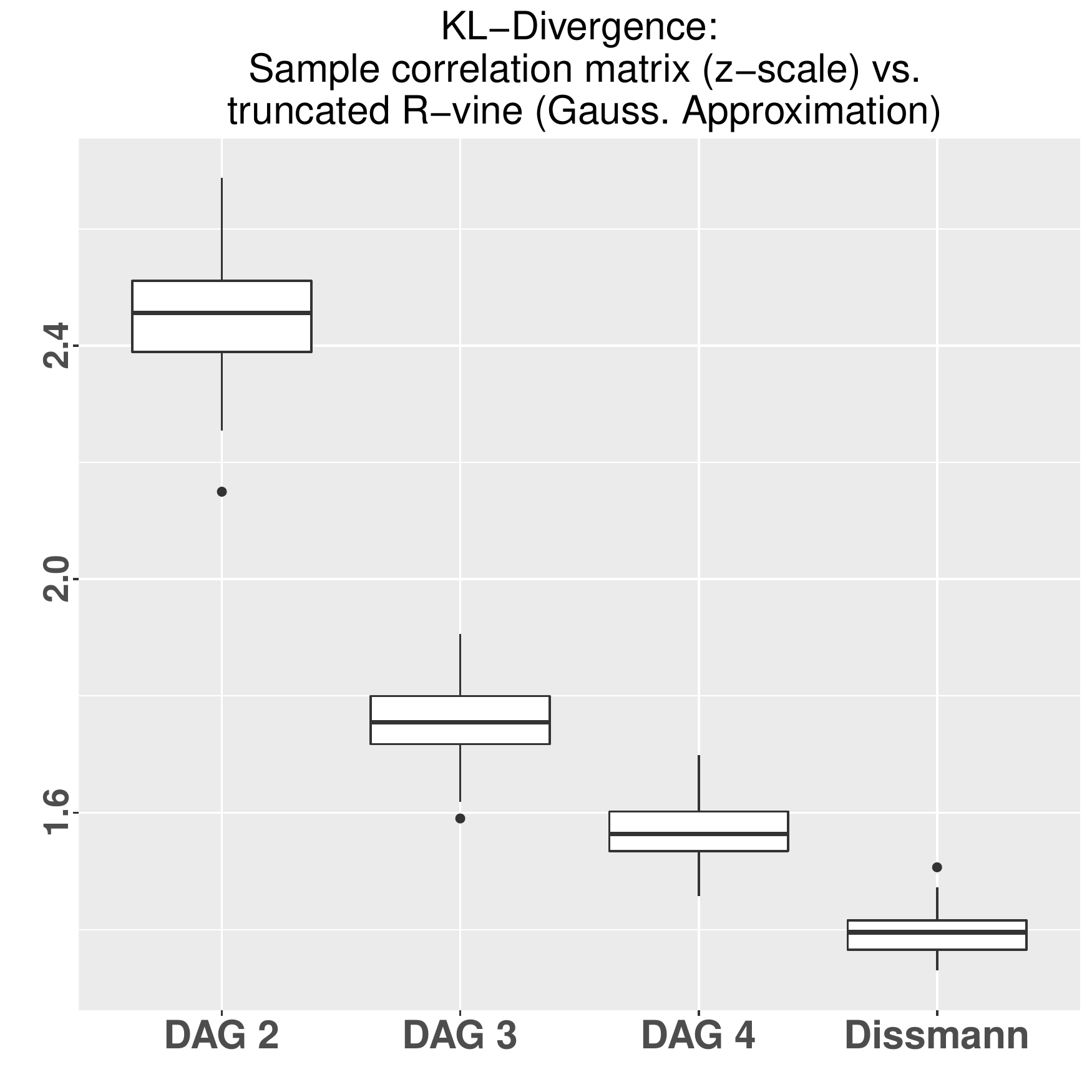}
	\caption{Kullback-Leibler divergence on $100$ replications between (a) $\widehat{\Sigma}$ and $\Sigma_{\GGG_k}$ (left), (b) $\Sigma_{\GGG_k}$ and $\Sigma_{\VVV\left(\GGG_k\right)}$ (centre), (c) $\widehat{\Sigma}$ and $\Sigma_{\VVV\left(\GGG_k\right)}$, $\Sigma_D$ (right) for $k=2,3,4$.}
	\label{fig:simstudy:resultsKL}
	\vspace{-0.25cm}
\end{figure}
We draw the conclusion that a $2$-DAG is not a good approximation of the sample correlation matrix, but we obtain better fit with $3$- and $4$-DAG. The rather low values in the centre plot indicate that our approach maps the structure between DAG and R-vine representation quite well. In the right plot, we see that Dissmann's algorithm obtains a smaller distance to the sample correlation matrix on the \textit{z-scale}. However, the distance between $k$-DAG and sample can still decrease for higher $k$, whereas the Dissmann model is already fully fitted.

\section{Application}\label{sec:application}
In \citet{brech-cc-2012a}, the authors analyzed the Euro Stoxx 50 and collected time series of daily log returns of $d=52$ major stocks and indices from May 22, 2006 to April 29, 2010 with $n=985$ observations. For these log returns, they fitted ARMA-GARCH time series models with Student-t's error distribution to remove trend and seasonality, obtaining standardized residuals. These are said to be on the \textit{x-scale} with a marginal distribution corresponding to a suitably chosen Student-t error distribution $F_i$. Using this parametric estimate for $F_i$, $i=1,\dots,d$, the copula data $U_i = F_i\left(X_i\right)$ is calculated. Since our approach uses Gaussian DAGs, we transform the data to have standard normal marginals, i.\,e.\ to the \textit{z-scale} by calculating $Z_i = \Phi^{-1}\left(U_i\right)$, with $\Phi$ the cdf of a $\NNN\left(0,1\right)$ distribution.\\
To learn $k$-DAGs for $k=1,\dots,10$ from the \textit{z-scale} data, we use the \textit{Hill-Climbing} algorithm of the R-package \textsf{bnlearn}, see \citet{bnlearnpaper} since we can limit the maximal number of parents. These $k$-DAGs are shown in Appendix \ref{subsec:dags}. Then, we apply our algorithm \texttt{RepresentDAGRVine} to calculate R-vine representations of the DAGs. To find pair copulas and parameters on these R-vines with independence copulas at given edges, we adapt functions of the R-package \textsf{VineCopula}, see \citet{VineCopula} and is apply them onto data on the \textit{u-scale}. All pair copula families of the R-package \texttt{VineCopula} were allowed.\\
As laid out initially, the paper has two goals. The first was to find truncated R-vines related to Gaussian DAGs which overcome the restriction of Gaussian distributions. Thus, we compare the goodness-of-fit of the $k$-DAGs $\GGG_k$ to their R-vine representations $\VVV\left(\GGG_k\right)$ from our algorithm. Given that our approach represents the structure of the DAGs well and there is non-Gaussian dependence, the variety of pair copula families of an R-vine should improve the fit notably. Second, we want to check whether our approach can compete with Dissmann's algorithm. Using their algorithm, we calculate a sequence of $t$-truncated R-vines for $t=1,\dots,51$, using an level $\alpha = 0.05$ independence test. Overall, we consider three different models in terms of the number of parameters and the corresponding log-likelihood and BIC values. Comparing the log-likelihood of DAGs and R-vines, we have to bear in mind that the marginals in the DAG are assumed to be standard normal and we also have to assume the same marginals for the R-vines, as done in e.\,g.\ \citet{Haffetal}. Yet, an advantage of vine copulas is that we can model marginals independently of the dependency structure. Thus, there is additional upside potential for the R-vine model. 
\begin{figure}[ht]
	\centering
	\includegraphics[width=0.49\textwidth, trim={0.1cm 0.6cm 0.7cm 0.6cm},clip]{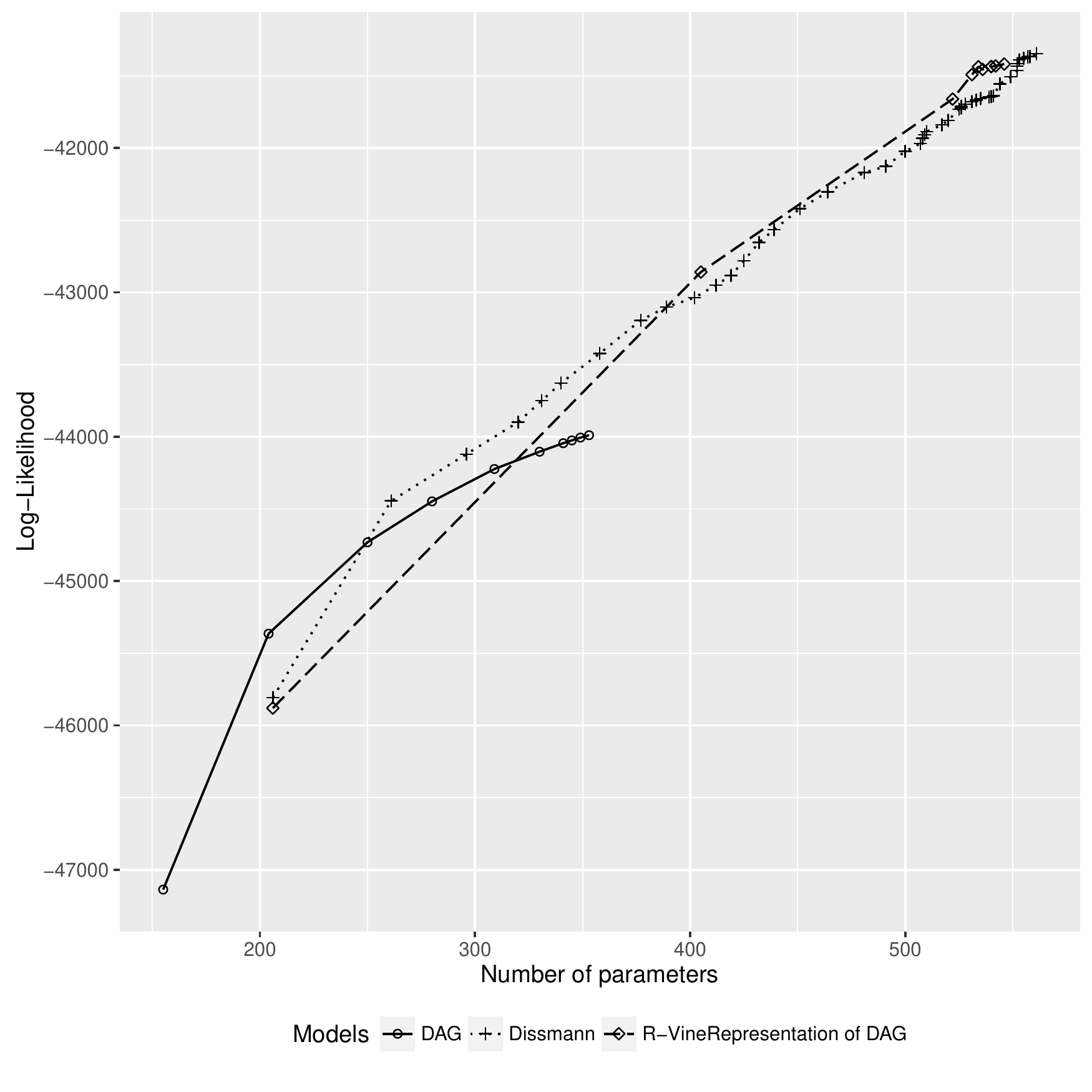}
	\includegraphics[width=0.49\textwidth, trim={0.1cm 0.6cm 0.7cm 0.6cm},clip]{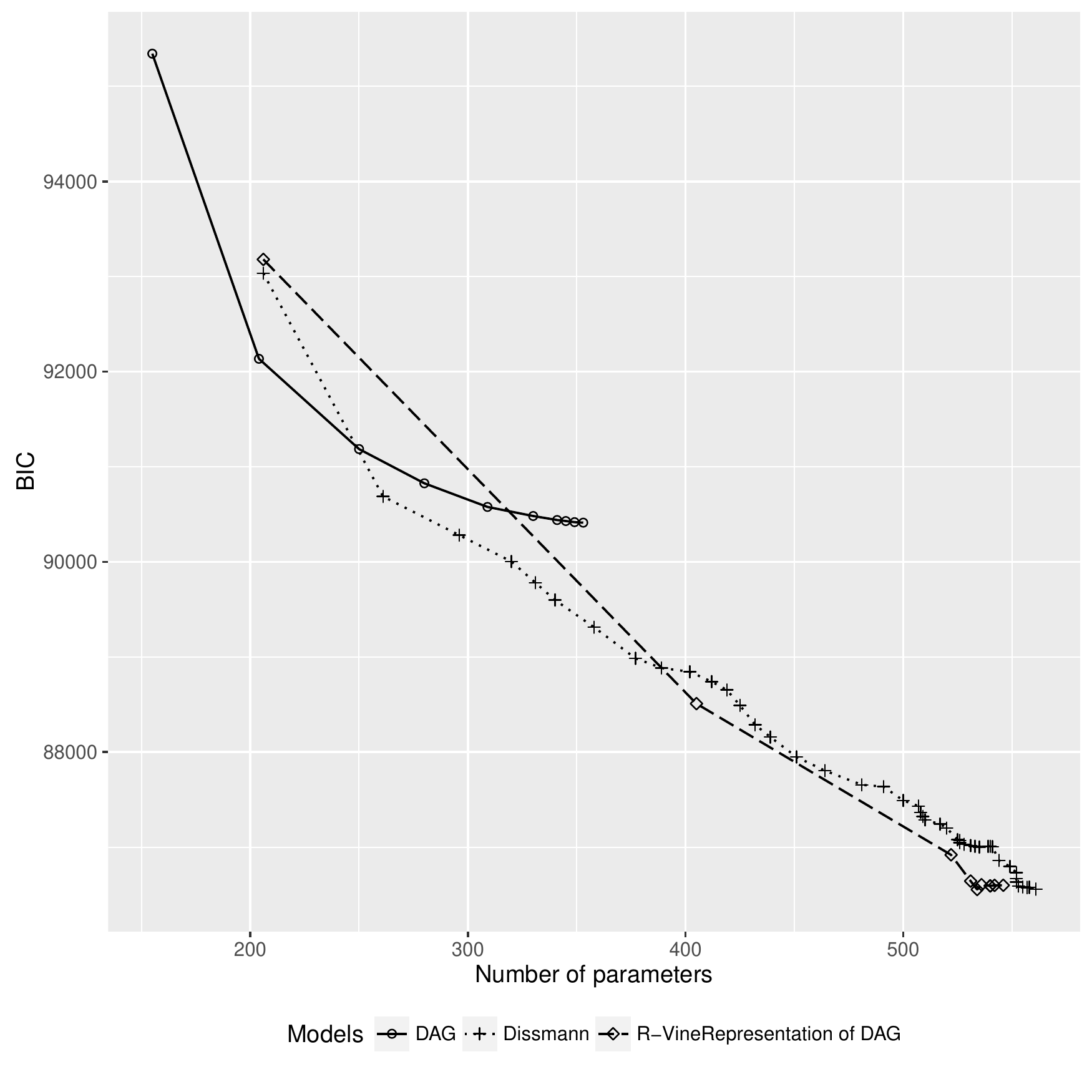}
	\caption{Comparison of $k$-DAGs $\GGG_k$, R-vine representations $\VVV\left(\GGG_k\right)$, $k=1,\dots,10$ and $t$-truncated Dissmann's algorithm, $t=1,\dots,51$ on \textit{z-scale}. Left: number of parameters vs. log-likelihood, right: number of parameters vs. BIC.}
	\label{fig:stoxx:dagvsvine}
\end{figure}
The results are given in Figure \ref{fig:stoxx:dagvsvine} and Tables \ref{table:application:dag}, \ref{table:application:dissmann} in Appendix \ref{subsec:appendix_application_tables}. The DAG models have the least parameters but their goodness-of-fit falls behind the two competitors. The reason is the presence of non-Gaussian dependence, i.\,e.\ $t$-copulas in the data which can not be modelled by the DAG. Comparing Dissmann's approach to our algorithm, we see a very similar behaviour when it comes to log-likelihood and BIC. However, our approach finds more parsimonious models given fixed levels of BIC. The computation time for our algorithm ranges from 125 sec. for a $1$-DAG to 270 sec. for a $10$-DAG. Dissmann's algorithm needs more than 600 sec. for a first R-vine tree and up to 760 sec. for a full estimation. Thus, our approach is about $3$ to $5$ times faster. This is also what we inferred from the simulation study. The computations were performed on a Linux Cluster with $32$ cores. Our approach is significantly faster, since given a specific edge $j\left(e\right),\ell\left(e\right)|D\left(e\right)$, Dissmann's algorithm first carries out an independence test for the pair copula. If the hypothesis is rejected, a maximum likelihood fit of the pair copula is carried out. Our approach checks $\condindep{j\left(e\right)}{\ell\left(e\right)}{D\left(e\right)}$ based on the d-separation in $\GGG_k$ and the corresponding copula is set to the independence.
The actual truncation levels $k^\prime$ of the R-vine representations are given in Table \ref{table:application:dag}. They are relatively high given the number of parents of these DAGs. However, this is because of very few non independence copulas in higher trees. For example, in the R-vine representation of the $2$-DAG, $T_{19},\dots,T_{51}$ contain $45$ non-independence copulas of $561$ edges, i.\,e.\ about $8$ \% are non-independence, see also Figure \ref{fig:stoxx:heatmap} in Appendix \ref{sec:appendix_heatmap}. This sparsity pattern is not negatively influencing the computation times or BIC as our examples demonstrated. It is also not intuitively apparent that a specific truncation level is more sensible to describe the data compared to a generally sparse structure.
%--------------------------------------------------------------------------
\section{Conclusion}\label{sec:conclusion}
%--------------------------------------------------------------------------
This paper aimed to link high dimensional DAG models with R-vines. Thus, the DAGs can be represented by a flexible modeling approach, overcoming the restrictive assumption of multivariate normality. Additionally, we intended to find new ways for non-sequential estimation of R-vine structures, a computationally highly demanding task. We proved a connection under sufficient conditions mapping $k$-DAGs to $k$-truncated R-vines. Afterwards, we gave necessary conditions for the corresponding DAGs to infer whether such R-vine models exist. For most cases more complex than a Markov Tree or special cases, an exact representation of a $k$-DAGs in terms of a truncated R-vine is not possible. However, it motives a general procedure to find more parsimonious R-vine models comparable to the standard algorithm, but multiple times faster. We expect this to leverage the application of R-vines in even higher dimensional settings with up to $1000$ variables.

	\section*{Acknowledgement}\label{sec:acknowledgment}
	The authors gratefully acknowledge the helpful comments of the referees, which further improved the manuscript.
	The first author is thankful for support from Allianz Deutschland AG. The second author is supported by the German Research foundation (DFG grant GZ 86/4-1). Numerical computations were performed on a Linux cluster supported by DFG grant INST 95/919-1 FUGG.

%\textsc{\bigskip
%\begin{center}
%	{\Large\bf SUPPLEMENTARY MATERIAL}
%\end{center}
%
%\begin{description}
%	
%	\item[Appendices] Definitions from Graph Theory (A), Examples \ref{ex:exvine1}, \ref{ex:exdagcycles} cont. (B), Proofs (C), Toy-example for heuristics (D), Supplementary material to simulation study (E),  Supplementary material to application (F), Algorithms (G) (.pdf file).
%	
%	\item[R-package VineCopulaDAG] Code to calculate R-vines based on Gaussian $k$-DAGs, code to check necessary conditions on given $k$-DAGs, both with example code and Euro Stoxx 50 data used in Section 6 (zipped .tar file).
%	
%\end{description}}

\bibliographystyle{chicago}
\bibliography{bibliography}
\addcontentsline{toc}{section}{REFERENCES}

\newpage

\appendix

\setcounter{table}{0}
\renewcommand{\thetable}{A\arabic{table}}
\renewcommand{\thefigure}{A\arabic{figure}}
\setcounter{page}{1}

\section{Definitions from Graph Theory}\label{sec:appendix_definitionsgraphtheory}
\begin{Definition}[Graph]
	Let $V \neq \emptyset$ be a finite set. Let $E \subseteq \left\{\left(v,w\right)\in V \times V: v \neq w\right\}$. Then, $\GGG=\left(V,E\right)$ is a graph with node set $V$ and edge set $E$.
\end{Definition}

\begin{Definition}[Edges]
	Let $\GGG=\left(V,E\right)$ be a graph and let $\left(v,w\right) \in E$ be an edge. We call an edge $\left(v,w\right)$ undirected if $\left(v,w\right) \in E \Rightarrow \left(w,v\right) \in E$ and directed if $\left(v,w\right) \in E \Rightarrow \left(w,v\right) \notin E$. A directed edge $\left(v,w\right)$ is also denoted by an arrow $v \rightarrow w$. Hereby, $w$ is called head of the edge and $v$ is called tail of the edge. By $v \leftrightarrow w$ we denote that $\left(v,w\right) \in E$ and $\left(w,v\right) \notin E$ or $\left(w,v\right) \in E$ and $\left(v,w\right) \notin E$, i.\,e.\ there is either a directed edge $v \rightarrow w$ or $w \rightarrow v$ in $\GGG$. By $v \nleftrightarrow w$ we denote that $\left(v,w\right) \notin E$ and $\left(w,v\right) \notin E$, i.\,e.\ there is no directed edge between $v$ and $w$.
\end{Definition}

\begin{Definition}[Directed and undirected graphs]
	Let $\GGG=\left(V,E\right)$ be a graph. We call $\GGG$ directed if each edge is directed. Similarly, we call $\GGG$ undirected if each edge is undirected.
\end{Definition}

\begin{Definition}[Weighted graph]
	A weighted graph is a graph $\GGG=\left(V,E\right)$ with \textit{weight function} $\mu$ such that $\mu: E \mapsto \mathbb{R}$.
\end{Definition}

\begin{Definition}[Skeleton]
	Let $\GGG=\left(V,E\right)$ be a directed graph. If we remove the edge orientation of each directed edge $v \rightarrow w$, we obtain the skeleton $\GGG^s$ of $\GGG$.
\end{Definition}

\begin{Definition}[Path]
	Let $\GGG = \left(V,E\right)$ be a graph. A path of length $k$ from $\alpha$ to $\beta$ is a sequence of distinct nodes $\alpha = \alpha_0, \ldots ,\alpha_k = \beta$ such that $\left(\alpha_{i-1},\alpha_i\right) \in E$ for $i=1,\ldots,k$. This definition applies for both directed and undirected graphs.
\end{Definition}

\begin{Definition}[Cycle]
	Let $\GGG = \left(V,E\right)$ be a graph and let $v \in V$. A cycle is defined as a path from $v$ to $v$.
\end{Definition}

\begin{Definition}[Acyclic graph]
	Let $\GGG=\left(V,E\right)$ be a graph. We call $\GGG$ acyclic if there exists no cycle within $\GGG$.
\end{Definition}

\begin{Definition}[Chain]
	Let $\GGG=\left(V,E\right)$ be a directed graph. A \textit{chain} of length $k$ from $\alpha$ to $\beta$ is a sequence of distinct nodes $\alpha=\alpha_0, \ldots, \alpha_k=\beta$ with $\alpha_{i-1} \rightarrow \alpha_i$ or $\alpha_{i} \rightarrow \alpha_{i-1}$ for $i=1,\ldots,k$. 
\end{Definition}

\begin{example}[Paths and chains in directed graphs]
	Consider the following two directed graphs $\GGG_1,\GGG_2$.
	\begin{figure}[H]
		\centering
		\includegraphics[width=0.48\textwidth]{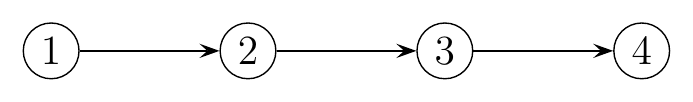}\\
		\includegraphics[width=0.48\textwidth]{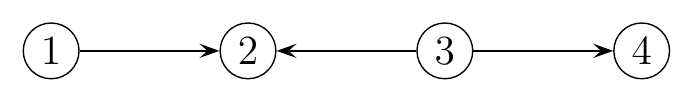}
		\caption{Example graphs $\GGG_1$ (upper) and $\GGG_2$ (lower)}
		\label{fig:ex:pathchain}
	\end{figure}
	For each of the two graphs, we consider the question whether a path or chain from $1$ to $4$ exists. In $\GGG_1$, clearly a path from $1$ to $4$ along $2$ and $3$ exists. Additionally, also a chain from $1$ to $4$ exists, as well as a chain from $4$ to $1$ since for the existence of a chain, the specific edge orientation is not relevant. With the same argument, in $\GGG_2$ there exists a chain between $1$ and $4$. However, no path between $1$ and $4$ exists as there is no edge $2 \rightarrow 3$.
\end{example}

\begin{Definition}[Subgraph]
	Let $\GGG=\left(V,E\right)$ be a graph. A graph $\HHH=\left(W,F\right)$ is a \textit{subgraph} of $\GGG = \left(V,E\right)$ if $W \subseteq V$ and $F \subseteq E$.
\end{Definition}
\begin{Definition}[Induced subgraph]
	Let $\GGG=\left(V,E\right)$ be a graph. A subgraph $\HHH=\left(W,F\right)$ of $\GGG$ is an \textit{induced subgraph} of $\GGG = \left(V,E\right)$ if $F = \left\{\left(v,w\right)|\left(v,w\right) \in W \times W: v \neq w \right\} \cap E$. 
\end{Definition}

\begin{example}[Subgraphs and induced subgraphs]
	Consider the following three graphs where $\GGG_2$ and $\GGG_3$ are subgraphs of $\GGG_1$. $\GGG_3$ is an induced subgraph of $\GGG_1$, whereas $\GGG_2$ is not since the edges $\left(2,3\right)$ and $\left(1,5\right)$ are present in $\GGG_1$ on the subset of nodes $\left\{1,2,3,5\right\}$ but are missing in $\GGG_2$.
	\begin{figure}[H]
		\centering
		\includegraphics[width=0.32\textwidth]{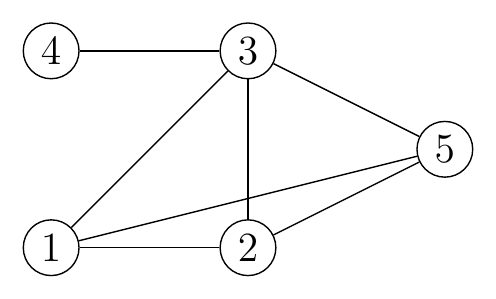}
		\includegraphics[width=0.32\textwidth]{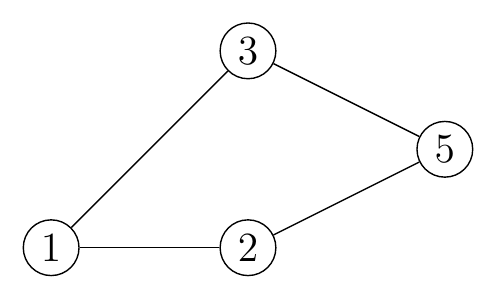}
		\includegraphics[width=0.32\textwidth]{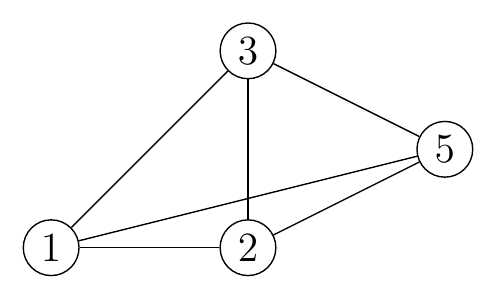}
		\caption{Example graphs $\GGG_1,\GGG_2,\GGG_3$ from left to right}
		\label{fig:ex:subgraph}
	\end{figure}
\end{example}

\begin{Definition}[Complete graph]
	Let $\GGG=\left(V,E\right)$ be a graph. $\GGG$ is called \textit{complete} if $E = \left\{\left(v,w\right)|\left(v,w\right) \in V \times V: v \neq w \right\}$.
\end{Definition}

\begin{Definition}[Connected graph]
	Let $\GGG=\left(V,E\right)$ be an undirected graph. If a path from $v$ to $w$ exists for all $v,w \in V$ in $\GGG$, we say that $\GGG$ is \textit{connected}.
\end{Definition}

\begin{Definition}[Weakly connected graph]
	Let $\GGG=\left(V,E\right)$ be a directed graph. If a path from $v$ to $w$ exists for all $v,w \in V$ in the (undirected) skeleton $\GGG^s$ of $\GGG$, we say that $\GGG$ is weakly connected.
\end{Definition}

\begin{Definition}[Tree]
	Let $\GGG=\left(V,E\right)$ be an undirected graph. $\GGG$ is a tree if it is connected and acyclic.
\end{Definition}

\begin{Definition}[Separator]
	Let $\GGG=\left(V,E\right)$. A subset $C \subseteq V$ is said to be an $\left(\alpha,\beta\right)$ \textit{separator} in $\GGG$ if all paths from $\alpha$ to $\beta$ intersect $C$. The subset $C$ is said to \textit{separate} $A$ from $B$ if it is an $\left(\alpha,\beta\right)$ separator for every $\alpha \in A$, $\beta \in B$. 
\end{Definition}

\begin{Definition}[v-structure]
	Let $\GGG=\left(V,E\right)$ be a directed acyclic graph. We define a \textit{v-structure} by a triple of nodes $\left(u,v,w\right) \in V$ if $u \rightarrow v$ and $w \rightarrow v$ but $u \nleftrightarrow w$.
\end{Definition}

\begin{Definition}[Moral graph]
	Let $\GGG=\left(V,E\right)$. The \textit{moral graph} $\GGG^m$ of a DAG $\GGG$ is defined as the skeleton $\GGG^s$ of $\GGG$ where for each v-structure $\left(u,v,w\right)$ an undirected edge $\left(u,w\right)$ is introduced in $\GGG^s$.
\end{Definition}

\begin{Definition}[d-separation]
	Let $\GGG = \left(V,E\right)$ be an directed acyclic graph. An chain $\pi$ from $a$ to $b$ in $\GGG$ is said to be \textit{blocked} by a set of nodes $S$, if it contains a node $\gamma \in \pi$ such that either
	\begin{enumerate}
		\item $\gamma \in S$ and arrows of $\pi$ do not meet head-to-head at $\gamma$, or
		\item $\gamma \notin S$ nor has $\gamma$ any descendants in $S$, and arrows of $\pi$ do meet head-to-head at $\gamma$.
	\end{enumerate}
	A chain that is not blocked by $S$ is said to be \textit{active}. Two subsets $A$ and $B$ are now said to be \textit{d-separated} by $S$ if all chains from $A$ to $B$ are blocked by $S$.
\end{Definition}

\begin{example}[d-separation]\label{ex:dsep}
	We give an example of a DAG $\GGG$, see Figure \ref{fig:exdsep}.
	\begin{figure}[H]
		\centering
		\includegraphics[width=0.4\textwidth]{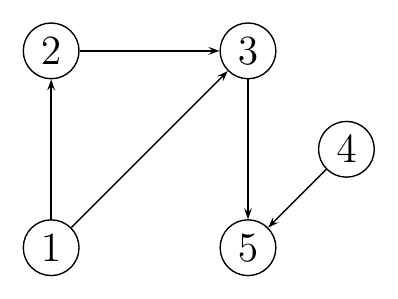}
		\caption{DAG $\GGG$}
		\label{fig:exdsep}
	\end{figure}
	First, we want to consider whether $\condindep{3}{4}{5}$ holds. In this case, $a=3$, $b=4$ and $S=5$. By application of the d-separation, we see that $S=5$ can not block a chain from $a$ to $b$ as arrows do meet head-to-head at $S=5$, hence $\condindep{3}{4}{5}$ does not hold.\\
	Second, we consider whether $\condindep{1}{5}{23}$. Thus, $a=1$, $b=5$ and $S=\left\{2,3\right\}$. The chain from $a$ to $b$ via $3$ is blocked as arrows meet not head-to-head at $3$. Second, also the chain from $a$ to $b$ via $\left\{2,3\right\}$ is blocked as arrows meet not head-to-head. Hence, we conclude $\condindep{1}{5}{23}$. For another example, see \citep[p.\ 50]{Lauritzen1996}.
\end{example}

\newpage

\section{Examples}\label{sec:appendix_examples}
\begin{example}[Example \ref{ex:exvine1} cont.]\label{ex:appendix_example}
	We continue Example \ref{ex:exvine1} showing the remaining m-children and m-descendants.
	$$
	\begin{array}{p{0.05\textwidth}p{0.15\textwidth}p{0.3\textwidth}p{0.3\textwidth}}
	tree  & edge $e$ & m-children of $e$ & m-descendants of $e$ \\
	\hline\hline
	$T_1$ & $2{,}1$ & $1{,}2$ & $1{,}2$ \\
	& $6{,}2$ & $2{,}6$ & $2{,}6$ \\
	& $3{,}6$ & $3{,}6$ & $3{,}6$ \\
	& $5{,}2$ & $5{,}2$ & $5{,}2$ \\
	& $4{,}5$ & $4{,}5$ & $4{,}5$ \\
	\hline
	$T_2$ & $6{,}1|2$ & $\left\{2,1\right\};\left\{6,2\right\}$ & $1{,}2{,}6$ \\
	& $3{,}2|6$ & $\left\{3,6\right\};\left\{6,2\right\}$ & $2{,}3{,}6$ \\
	& $5{,}6|2$ & $\left\{5,2\right\};\left\{6,2\right\}$ & $2{,}5{,}6$ \\
	& $4{,}2|5$ & $\left\{4,5\right\};\left\{5,2\right\}$ & $2{,}4{,}5$ \\
	\hline
	$T_3$ &$ 3{,}1|26 $&$ 6{,}1|2$; $ 3{,}2|6 $& $1{,}2{,}3{,}6$ \\
	& $5{,}3|26 $&$ 3{,}2|6$; $5{,}6|2$ & $2{,}3{,}5{,}6$ \\
	& $4{,}6|25 $&$ 5{,}6|2$; $4{,}2|5$ & $2{,}4{,}5{,}6$ \\
	\hline
	$T_4$ &$ 5{,}1|236 $&$ 3{,}1|26$, $5{,}3|26 $& $1{,}2{,}3{,}5{,}6$ \\
	&$ 4{,}3|256 $&$ 5{,}3|26$, $4{,}6|25 $& $2{,}3{,}4{,}5{,}6$ \\
	\hline
	$T_5$ & $1{,}4|2356$ & $5{,}1|236${,}$ 4{,}3|256 $ & $ 1{,}2{,}3{,}4{,}5{,}6$\\
	\hline
	\end{array}
	$$
	\captionof{table}{Edges, m-children and m-descendants in the R-vine of Example \ref{ex:exvine1}.}
	\label{table:exvine:2}
\end{example}

\begin{example}[Example \ref{ex:exdagcycles} cont.]\label{ex:exdagcycles:cont.}
	We show the corresponding first and second R-vines trees as outlined which lead to $3$-truncated R-vines.
	\begin{figure}[H]
		\centering
		\includegraphics[width=0.32\textwidth]{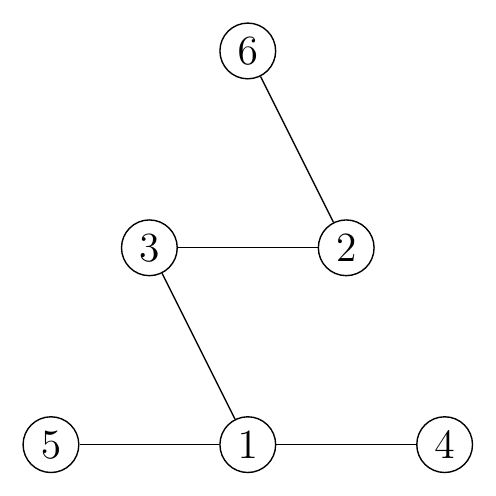}
		\includegraphics[width=0.32\textwidth]{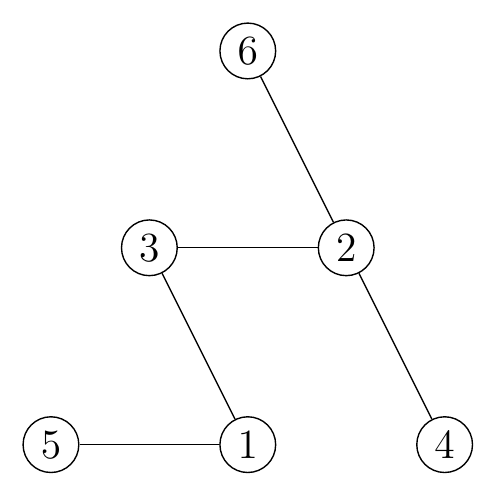}
		\includegraphics[width=0.32\textwidth]{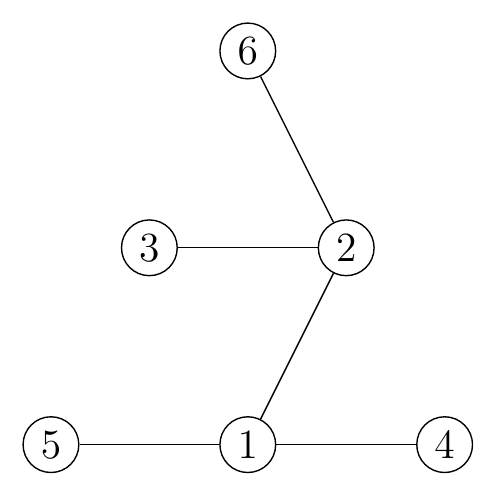}\\
		\includegraphics[width=0.32\textwidth]{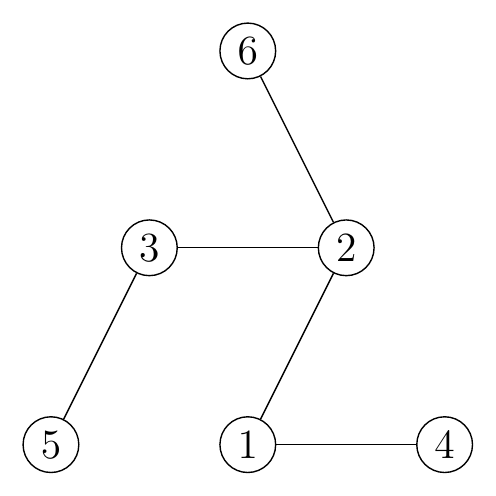}
		\includegraphics[width=0.32\textwidth]{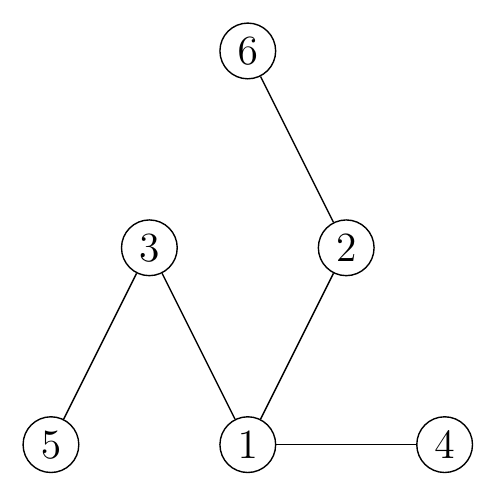}
		\includegraphics[width=0.32\textwidth]{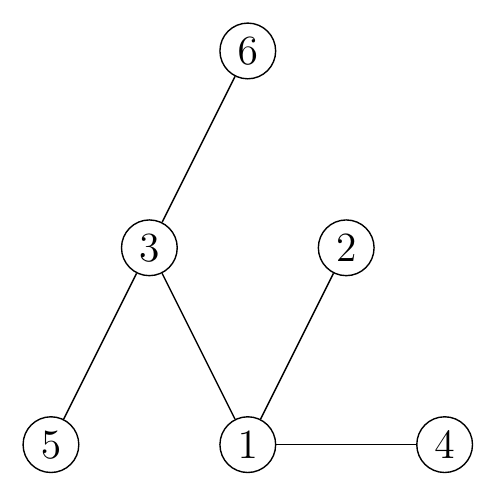}
		\caption{Example \ref{ex:exdagcycles}: $6$ admissible choices of first R-vine trees $T_1$ leading to a $3$-truncated R-vine.}
		\label{fig:exdagcycles:ex2vine1}
	\end{figure}
	\begin{figure}[H]
		\centering
		\includegraphics[width=0.24\textwidth]{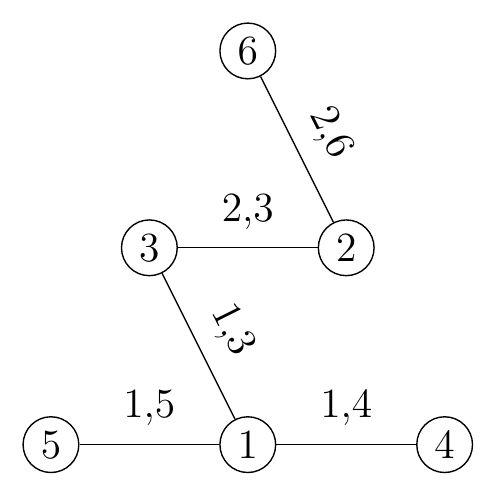}
		\includegraphics[width=0.24\textwidth]{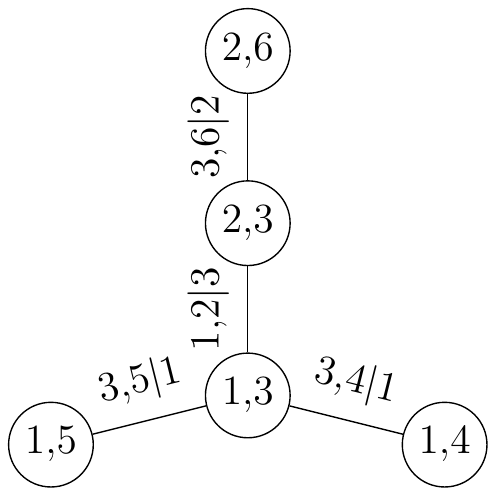}
		\includegraphics[width=0.24\textwidth]{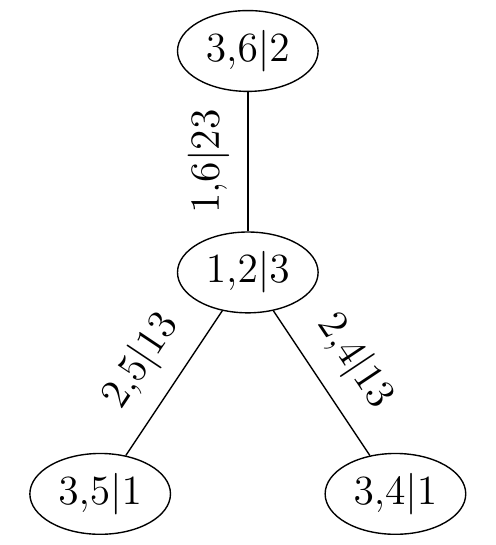}
		\includegraphics[width=0.24\textwidth]{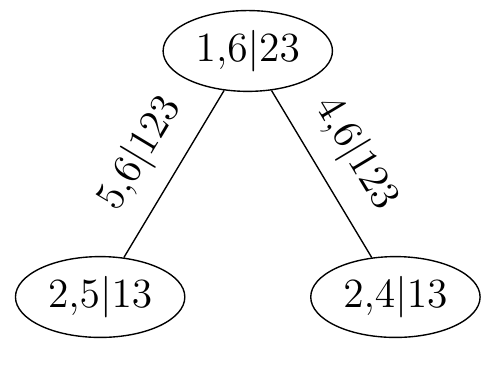}
		\caption{Example \ref{ex:exdagcycles}: First $4$ R-vine trees $T_1,T_2,T_3,T_4$ (from left to right), showing the $3$-truncation as given by the conditional independence properties in the DAG $\GGG_2$. Note that $24|13$ can not be set to the independence copula since $2$ is a parent of $4$.}
		\label{fig:exdagcycles:ex2vine2}
	\end{figure}
\end{example}

\newpage

\section{Proofs}\label{sec:appendix_proof}
\subsection{Proof of Proposition \ref{prop:dagk1}}\label{subsec:appendix_proof_prop_markovtree}
\begin{proof}
	We assume $\left|E\right| = d-1$. If not, the argument can be applied to each weakly connected subgraph of $\GGG$. Since $k=1$, there are no \textit{v-structures}, hence, the moral graph $\GGG^m$ is the skeleton $\GGG^s$ of $\GGG$ and $\GGG^m$ is connected. Since there are $d-1$ arrows in $\GGG$, there are $d-1$ undirected edges in $\GGG^m$. Since each connected graph on $d$ nodes with $d-1$ edges is a tree, $\GGG^m$ is a tree. Additionally, each edge in $\GGG^m$ corresponds to an arrow $w \leftrightarrow v$ in the DAG $\GGG$, satisfying Assumption \ref{ass:1}. The main diagonal of the R-vine matrix can be chosen to be a decreasing topological ordering of $\GGG$ by starting with a node which has no descendants but one parent, say $v_d$ and let the corresponding R-vine matrix $M$ be such that $M_{1,1}=v_d$. Thus, its parent and all other nodes must occur on the diagonal to the right of it. Next, take a node which has either one descendant, i.\,e.\ $v_d$ or no descendant, denote $v_{d-1}$ and set $M_{2,2}=v_{d-1}$. This can be repeated until $v_1=M_{d,d}$ and determines the R-vine matrix main diagonal which is a decreasing topological ordering of $\GGG$, satisfying Assumption \ref{ass:2} onto which Theorem \ref{thm:transformation} applies.
\end{proof}

\subsection{Proof of Corollary \ref{cor:lowerboundkprime}}\label{subsec:appendix_proof_corollarylowerbound}
\begin{proof}
	Since $T_1$ is a tree, all paths are unique. If not, there exist two distinct paths between $v$ and $w$ and both paths together are a cycle from $v$ to $v$. Consider an arbitrary node $v \in V$ with parents $\parents\left(v\right)=\left\{w_1^v,\dots,w_{k_v}^v\right\}$ in $\GGG$ such that $w_{k_v}^v \coloneqq \argmax_{w \in \parents\left(v\right)}~\delta_v^w$, then there exists a unique path from $v$ to $w_{k_v}^v$, $v=\alpha_0,\dots,\alpha_{\delta_v^w}=w_{k_v}^v$. From Theorem \ref{thm:transformation}, our goal is to obtain edges with conditioned sets $v,w$ with $w \in \parents\left(v\right)$ in an R-vine tree $T_i$ with lowest possible order $i$. Similar to the proof of Proposition \ref{prop:necessaryconditiondag}, we try to obtain an edge
	\begin{equation}\label{eq:proofcorollaryedge}
		v,w_{k_v}^v|\alpha_1,\ldots,\alpha_{\delta_v^w-1} \in T_{\delta_v^w},
	\end{equation}
	with $\delta_v^w-1$ entries in the conditioning set. The conditioned set of \ref{eq:proofcorollaryedge} can not occur in a tree $T_i$ with $i < \delta_v^w$ because of the proximity condition and since the path between $v$ and $w_{k_v}^v$ is unique. By the d-separation, page \pageref{page:dsep}, two nodes in a DAG connected by an arrow, i.\,e.\ $v$ and its parent $w_{k_v}^v$, can not be d-separated by any set $S$. Thus, the pair copula density associated to the edge \eqref{eq:proofcorollaryedge} in $T_{\delta_v^w}$ is not the independence copula density $c^\perp$. This tree $T_{\delta_v^w}$ is characterized by a path distance in $T_1$ and the maximum path distance over all parents of $v \in V$ yields the highest lower bound. As it has to hold for all $v \in V$, we obtain a lower bound for the truncation level $k^\prime$ by the maximum over all $v \in V$.
\end{proof}

We present a brief example for the Corollary.

\begin{example}[Example for Corollary \ref{cor:lowerboundkprime}]\label{ex:corollary}
	Consider the R-vine tree $T_1$ in Figure \ref{fig:excorollary}.
	\begin{figure}[h]
		\centering
		\includegraphics[width=0.3\textwidth, trim={0.2cm 0.2cm 0.2cm 0.2cm},clip]{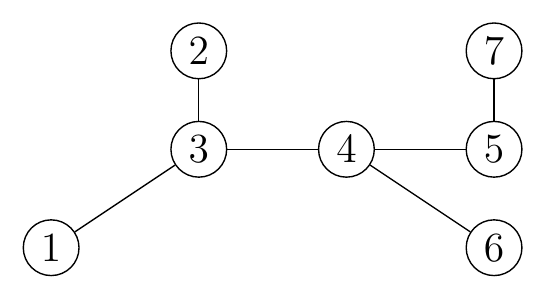}
		\caption{Example \ref{ex:corollary}: R-vine tree $T_1$}
		\label{fig:excorollary}
	\end{figure}
	Assume an underlying DAG $\GGG$ with $1\in \parents\left(7\right)$. We have a lower bound for the truncation level $k^\prime \ge 4$ since the path in $T_1$ from $7$ to $1$ is $7-5-4-3-1$ with a path length $\ell_7^1=4$. Not earlier as in tree $T_4$, i.\,e.\ not in the trees $T_1,T_2,T_3$ an edge with conditioned set $7,1$ can be obtained which can not be represented by the independence copula.
	%	The corresponding edges in the trees $T_1,\ldots,T_4$ are as follows.
	%	\begin{equation*}
	%	\begin{array}{cccccccc}
	%	T_1: &1,3 & &3,4 & &4,5 & &5,7\\
	%	T_2: &&1,4|3 && 3,5|4 && 4,7|5 &\\
	%	T_3: &&&1,5|3,4 &&3,7|4,5 &&  \\
	%	T_4: &&&&1,7|3,4,5& & &
	%	\end{array}
	%	\end{equation*}
\end{example}
\newpage

\section{Toy-example for heuristics}\label{sec:appendix_toyexample}
\begin{example}[Heuristics for transformation]\label{ex:heuristics}
	Consider the DAGs $\GGG_k$, $k=1,2,3$, with at most $k$ parents, see Figure \ref{fig:ExDAGHeuristics:DAG}, from left to right. 
	\begin{figure}[H]
		\centering
		\includegraphics[width=0.32\textwidth]{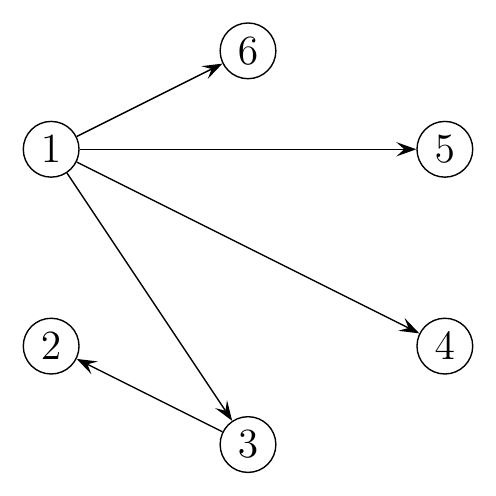}
		\includegraphics[width=0.32\textwidth]{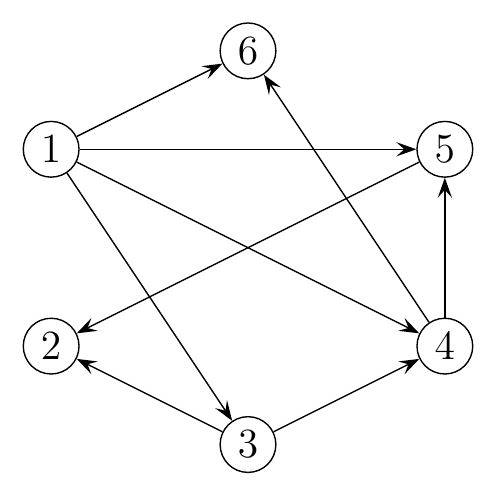}
		\includegraphics[width=0.32\textwidth]{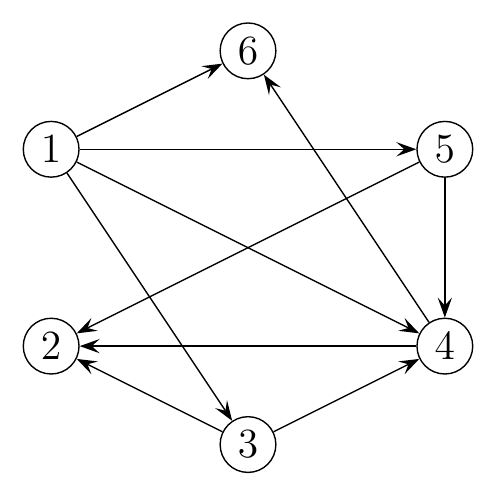}
		\caption{Example DAGs $\GGG_k$ for $k=1,2,3$ with at most $k$ parents (from left to right).}
		\label{fig:ExDAGHeuristics:DAG}
	\end{figure}
	Applying a maximum spanning tree algorithm on $\HHH$ to find the first R-vine tree $T_1$, we obtain the skeleton $\GGG_1^s$ (see Figure \ref{fig:ExDAGHeuristics:Vine1}, first figure). This is however not in general the case. We sketch the intermediate step of building $T_2$, where we already removed edges not allowed by the proximity condition and assigned weights according to Equation \eqref{eq:edgeweight2} (see Figure \ref{fig:ExDAGHeuristics:Vine1}, second to fourth figures).
	\begin{figure}[H]
		\centering
		\includegraphics[width=0.24\textwidth]{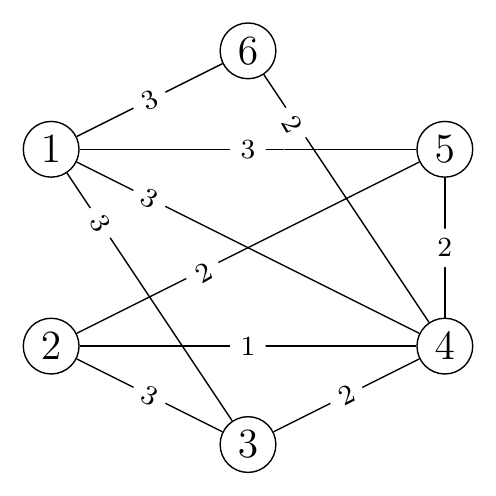}
		\includegraphics[width=0.24\textwidth]{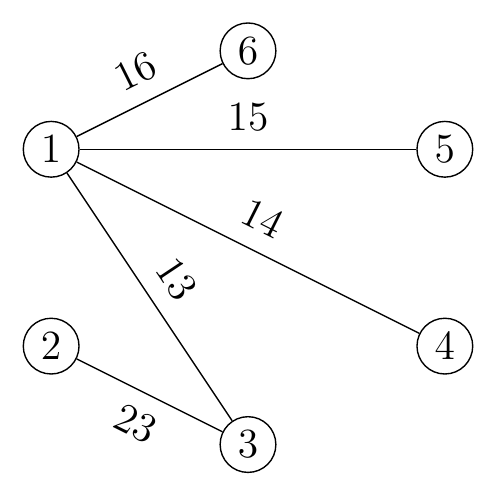}
		\includegraphics[width=0.24\textwidth]{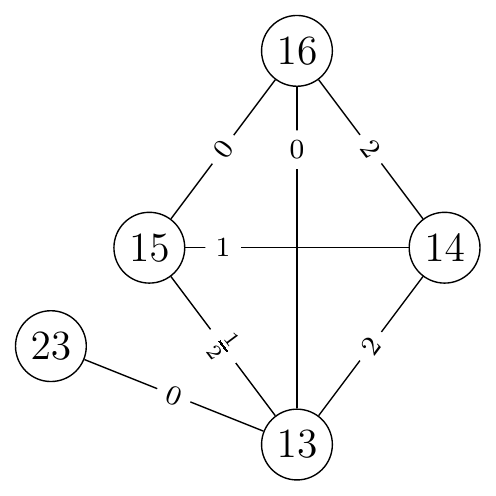}
		\includegraphics[width=0.24\textwidth]{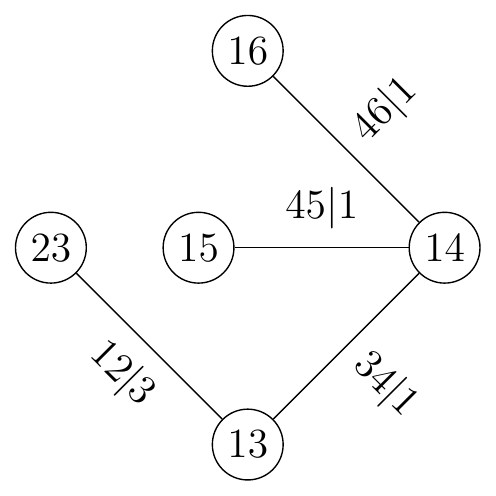}
		\caption{Weighted graph $\HHH$ with weight function  $g\left(i\right) \equiv 1$ and $\mu_0 = \frac{g\left(1\right)}{2} = \frac{1}{2}$ (left) and R-vine representation of the DAG $\GGG_3$ with R-vine trees $T_1$, intermediate step for building tree $T_2$ and final $T_2$ (from left to right). Note that $\condindep{3}{5}{1}$ by the d-separation in $\GGG_3$ and hence the weight of the corresponding edge assigned is $\mu_0 = \frac{1}{2}$. However, this edge is not chosen by the maximum spanning tree algorithm.}
		\label{fig:ExDAGHeuristics:Vine1}
	\end{figure}
	We see that $T_3$ has the form of a so called D-vine, i.\,e.\ the R-vine tree is a path. Thus, the structure of higher order trees $T_4$ and $T_5$ is already determined, see Figure \ref{fig:ExDAGHeuristics:Vine2}.
	\begin{figure}[H]
		\centering
		\includegraphics[width=0.32\textwidth]{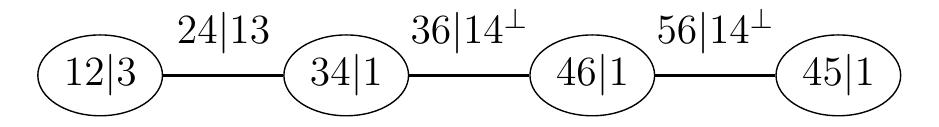}
		\includegraphics[width=0.32\textwidth]{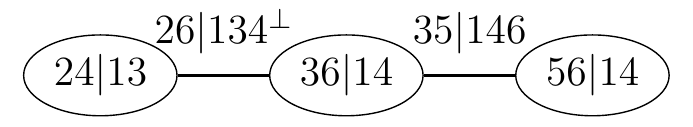}
		\includegraphics[width=0.24\textwidth]{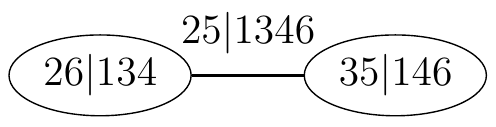}
		\caption{R-vine representation of the DAG $\GGG_3$. Trees $T_3, T_4, T_5$ (from left to right). Edges with superscript $\perp$ are associated with the independence copula by the d-separation in $\GGG_3$.}
		\label{fig:ExDAGHeuristics:Vine2}
	\end{figure}
	Based on the first R-vine tree $T_1$ and Corollary \ref{cor:lowerboundkprime} we infer the lower bound for the truncation level. We consider the sets $V^v = \left\{v,\parents\left(v\right)\right\}$ for $v \in V$ based on $\GGG_3$. For example, the node $2$ has the parents $\parents\left(2\right) = \left\{3,4,5\right\}$ in $\GGG_3$. Based on the first R-vine tree $T_1$ we check the lengths of shortest paths between $2$ and its parents and obtain $\ell_2^3=1$, $\ell_2^4=3$ and $\ell_2^5=3$. By application of Corollary \ref{cor:lowerboundkprime}, this gives a lower bound for the truncation level $k^\prime \ge 3$. The lengths of the shortest paths in $T_1$ for all nodes $v \in V$ can be found in Table \ref{table:ExDAGHeuristics:distances}.
	\begin{table}[H]
		\centering
		\begin{tabular}{crrr}
			$v$ & $\parents\left(v\right)=\left\{w_1^v,w_2^v,w_3^v\right\}$ & $\ell_v^{w_1} ,\ell_v^{w_2},\ell_v^{w_3}$ & $\max_{w \in \parents\left(v\right) }~\ell_v^w$ \\
			\hline \hline
			1 & - & - & -\\
			2 & 3,4,5 & 1,3,3 & 3\\
			3 & 1 & 1 & 1\\
			4 & 1,3,5 & 1,2,2 & 2\\
			5 & 1 & 1 & 1\\
			6 & 1,4 & 1,2 & 2\\
			\hline
		\end{tabular}
		\caption{Shortest path distances in $T_1$ between nodes $v$ and its parents $\parents\left(v\right)$ in DAG $\GGG_3$.}
		\label{table:ExDAGHeuristics:distances}
	\end{table}
	We obtain $k^\prime = \max_{v \in V}~\max_{w \in \parents\left(v\right)}~\ell_v^w = 3$. Note that this lower bound is not attained as we have the conditioned set $\left\{2,5\right\}$ in the R-vine Tree $T_5$ which can not be represented by the independence copula as this conditioned set is associated to an edge in the DAG $\GGG_3$. However, several edges with superscript $\perp$ can be associated with the independence copula by the d-separation. The trees $\GGG_1$ and $\GGG_2$ are only used to obtain the weights for the corresponding trees, but not with respect to check for d-separation.
\end{example}

\newpage

\section{Supplementary material to simulation study}\label{sec:appendix_simstudy}
We restate the simulation setup for the remaining scenarios.
\begin{table}[h]
	\centering
	\begin{tabular}{rrrr}
		Scenario & pair copula families & truncation level & indep. test significance level $\alpha$  \\ 
		\hline \hline
		1 & all & - & $0.05$ \\ 
		2 & independence, t & - & $0.05$ \\
		3 & all & 4 & $0.05$ \\
		4 & all & - & $0.2$ \\
		5 & all & 4 & $0.2$ \\
		\hline
	\end{tabular}
	\caption{Parameter settings for sample models in simulation study.}
	\label{table:simstudy_parameters_appendix}
\end{table}
Now, we present the results for the remaining scenarios $1$ to $4$.
\begin{figure}[H]
	\centering
	\includegraphics[width=0.24\textwidth, trim={0.1cm 0.1cm 0.1cm 0.1cm},clip]{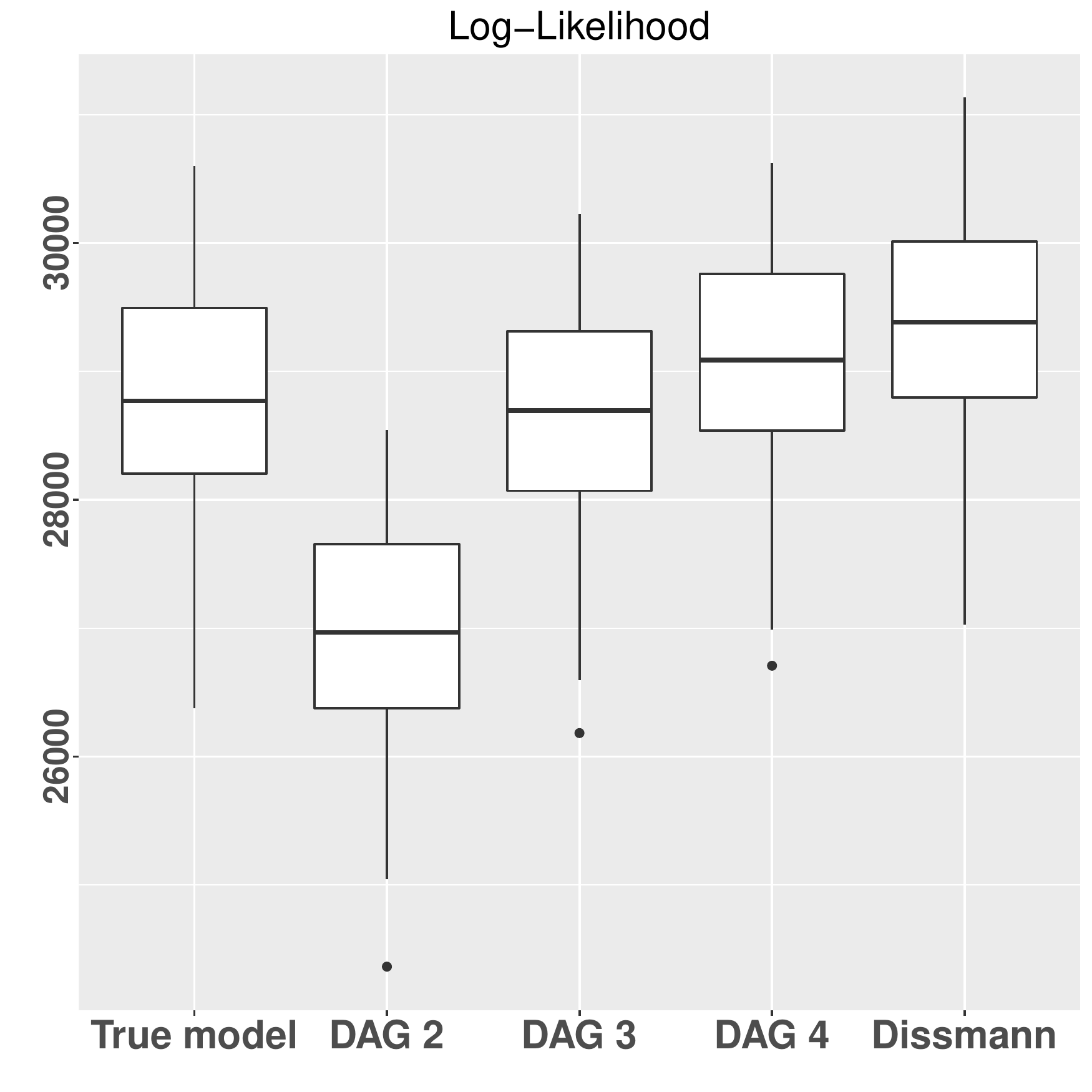}
	\includegraphics[width=0.24\textwidth, trim={0.1cm 0.1cm 0.1cm 0.1cm},clip]{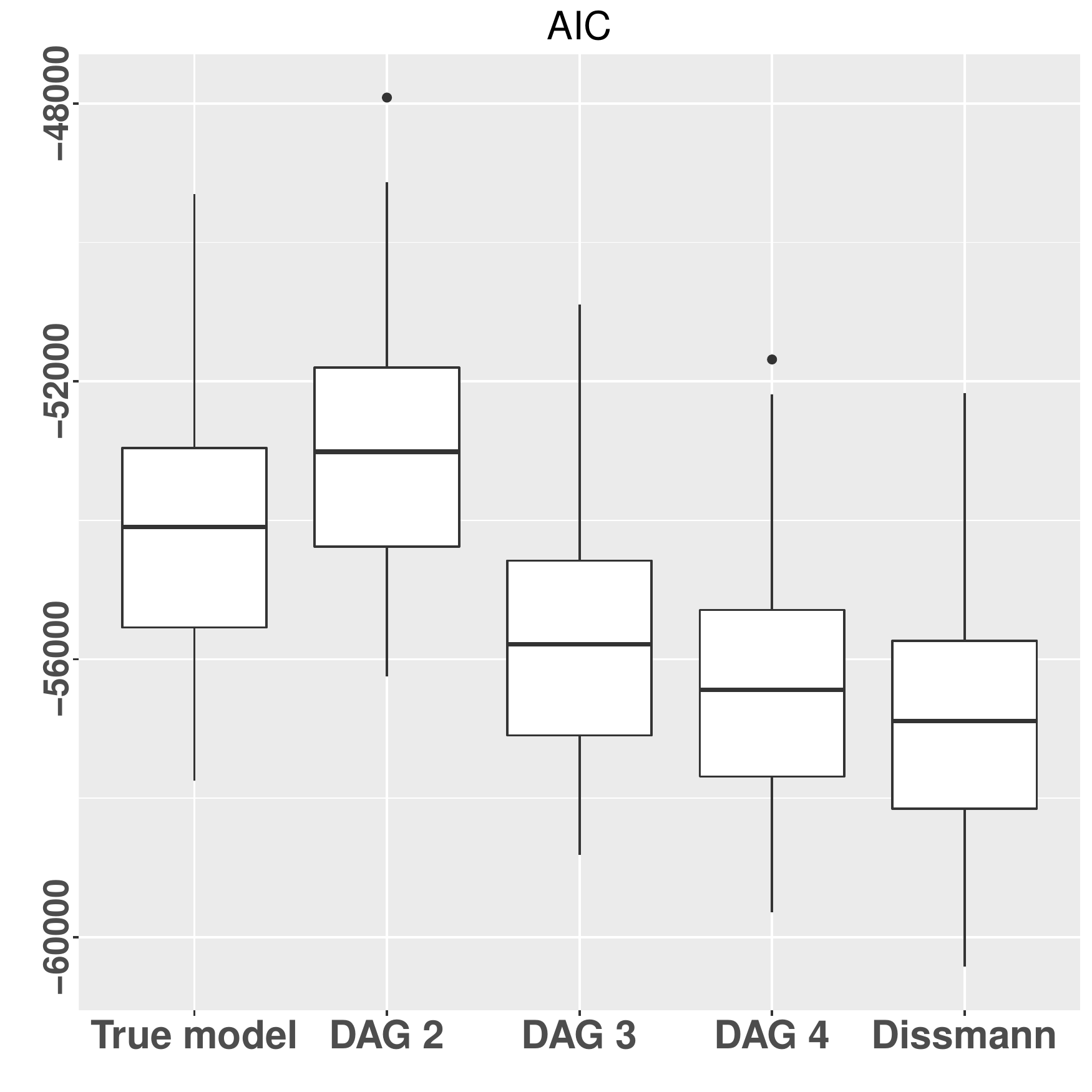}
	\includegraphics[width=0.24\textwidth, trim={0.1cm 0.1cm 0.1cm 0.1cm},clip]{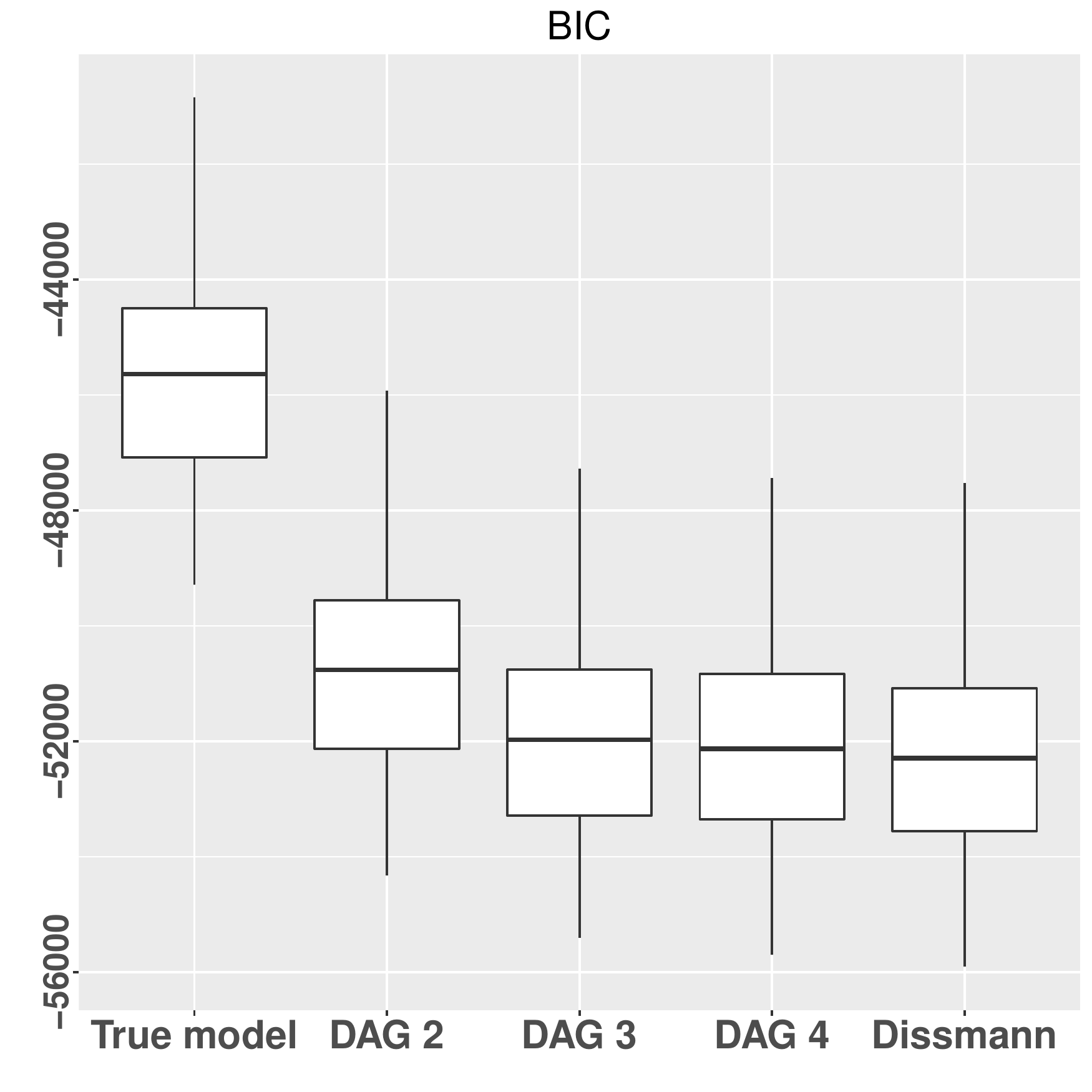}
	\includegraphics[width=0.24\textwidth, trim={0.1cm 0.1cm 0.1cm 0.1cm},clip]{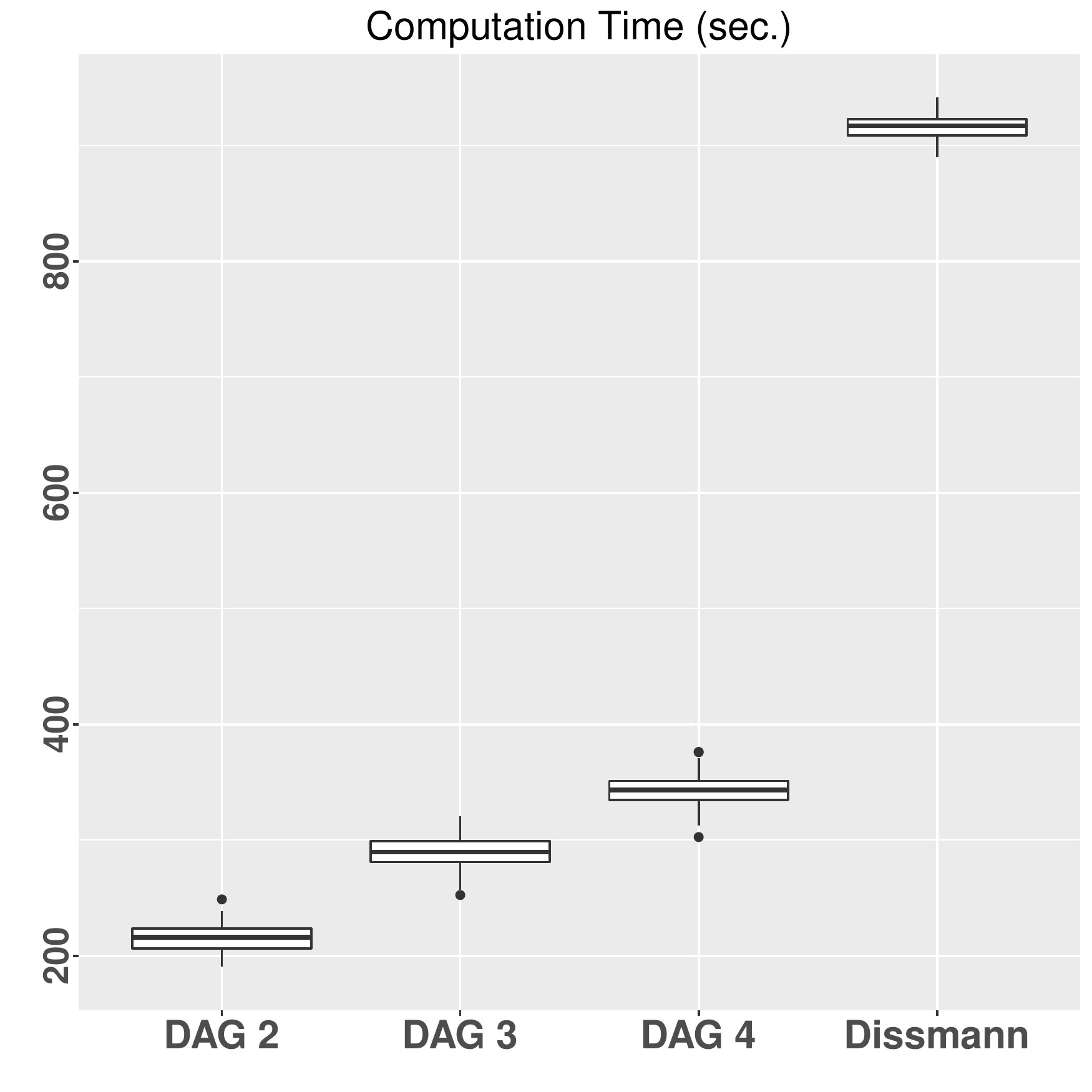}
	\caption{Scenario 1: Comparison of $k$-DAG representations for $k=2,3,4$ with Dissmann algorithm considering log-Likelihood, AIC, BIC and computation time in seconds on $100$ replications (from left to right).}
	\label{fig:simstudy:results1}
\end{figure}
\begin{figure}[H]
	\centering
	\includegraphics[width=0.24\textwidth, trim={0.1cm 0.1cm 0.1cm 0.1cm},clip]{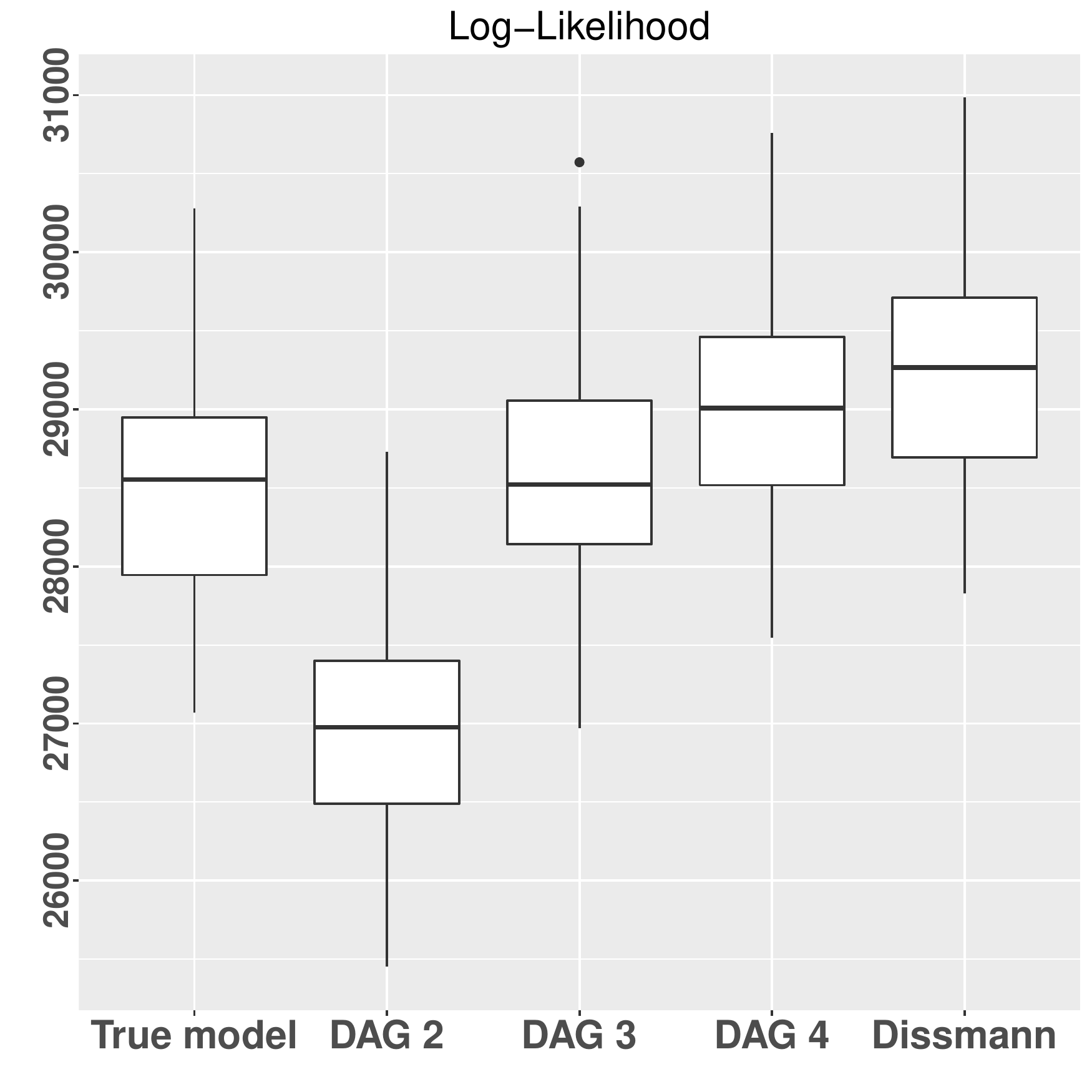}
	\includegraphics[width=0.24\textwidth, trim={0.1cm 0.1cm 0.1cm 0.1cm},clip]{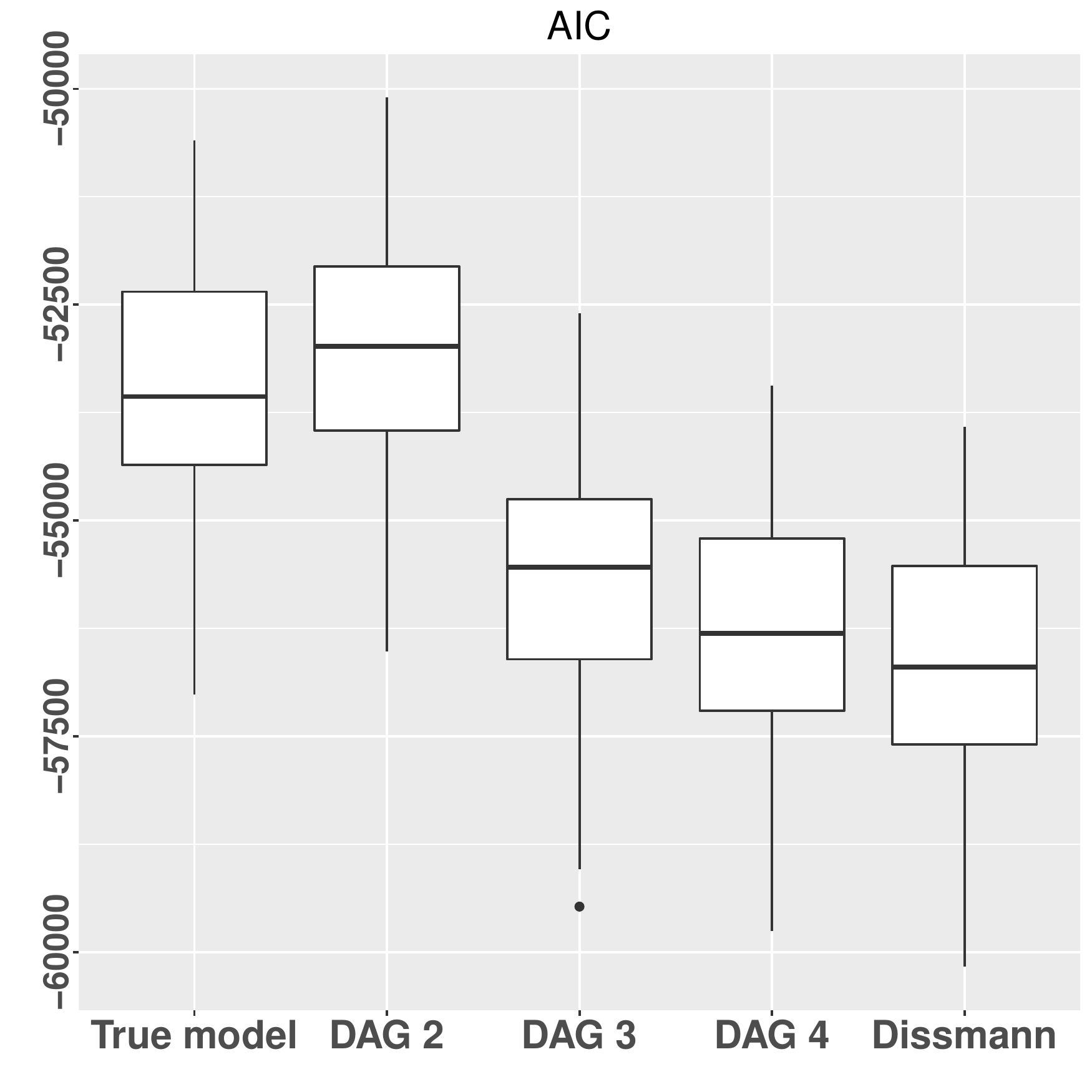}
	\includegraphics[width=0.24\textwidth, trim={0.1cm 0.1cm 0.1cm 0.1cm},clip]{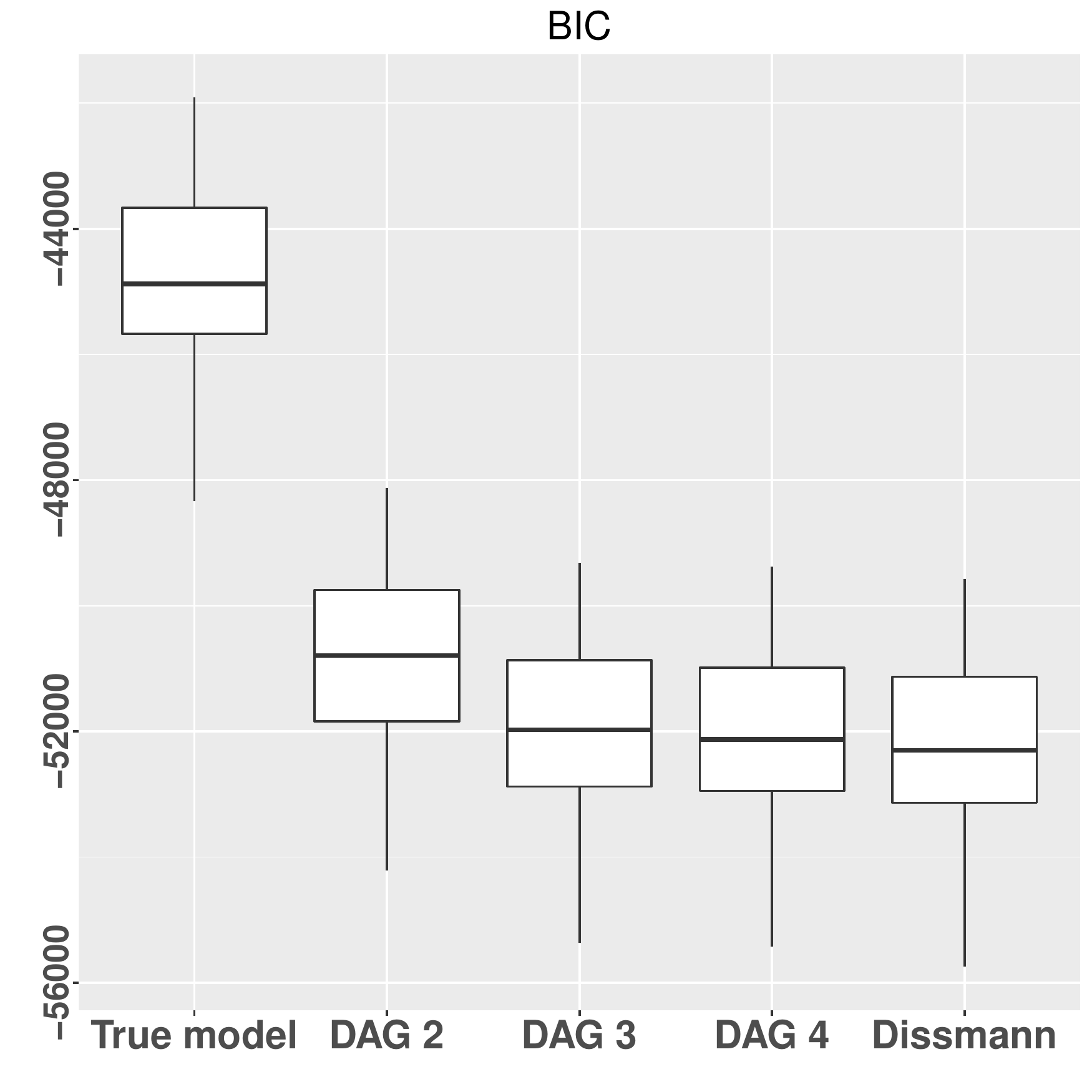}
	\includegraphics[width=0.24\textwidth, trim={0.1cm 0.1cm 0.1cm 0.1cm},clip]{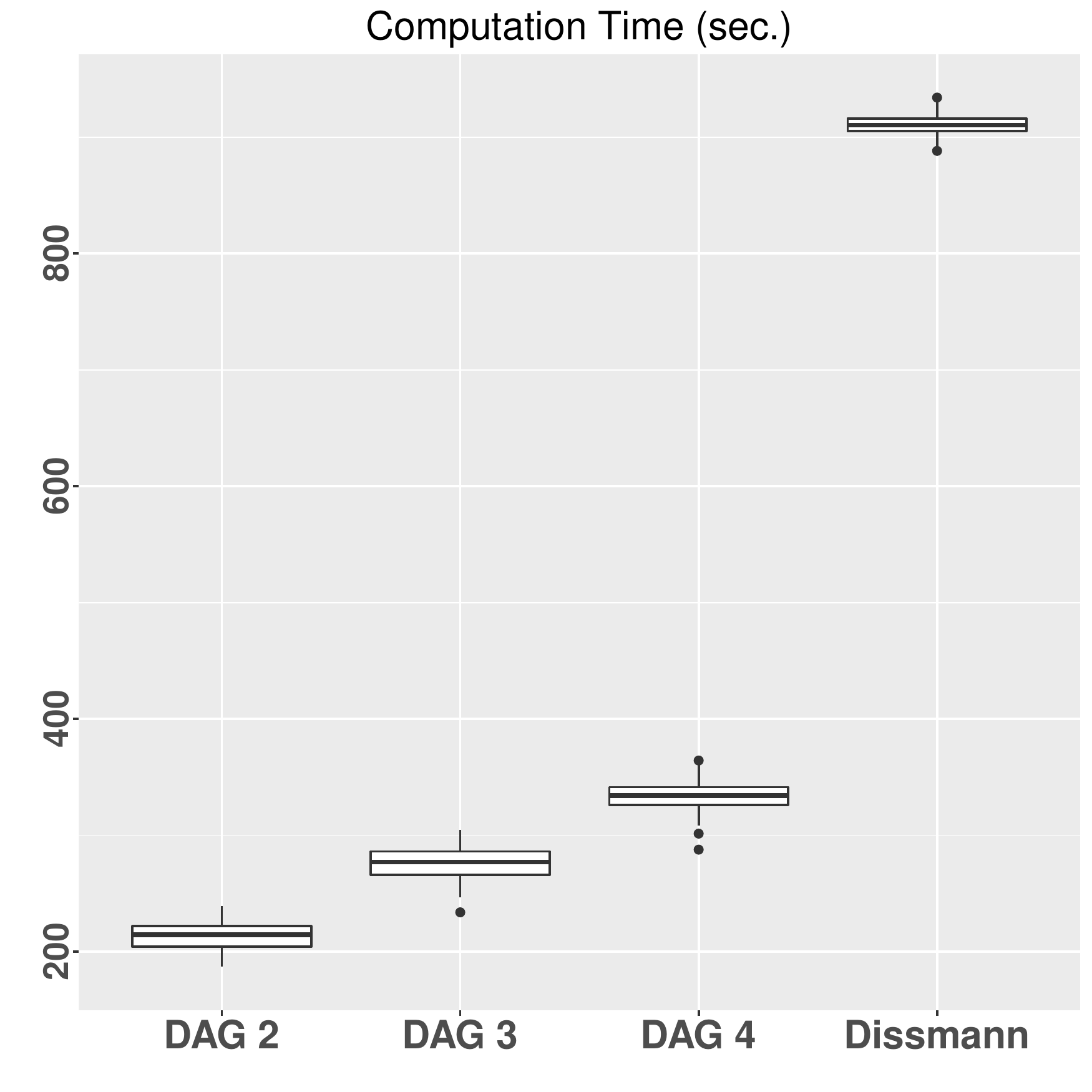}
	\caption{Scenario 2: Comparison of $k$-DAG representations for $k=2,3,4$ with Dissmann algorithm considering log-Likelihood, AIC, BIC and computation time in seconds on $100$ replications (from left to right).}
	\label{fig:simstudy:results2}
\end{figure}
\begin{figure}[H]
	\centering
	\includegraphics[width=0.24\textwidth, trim={0.1cm 0.1cm 0.1cm 0.1cm},clip]{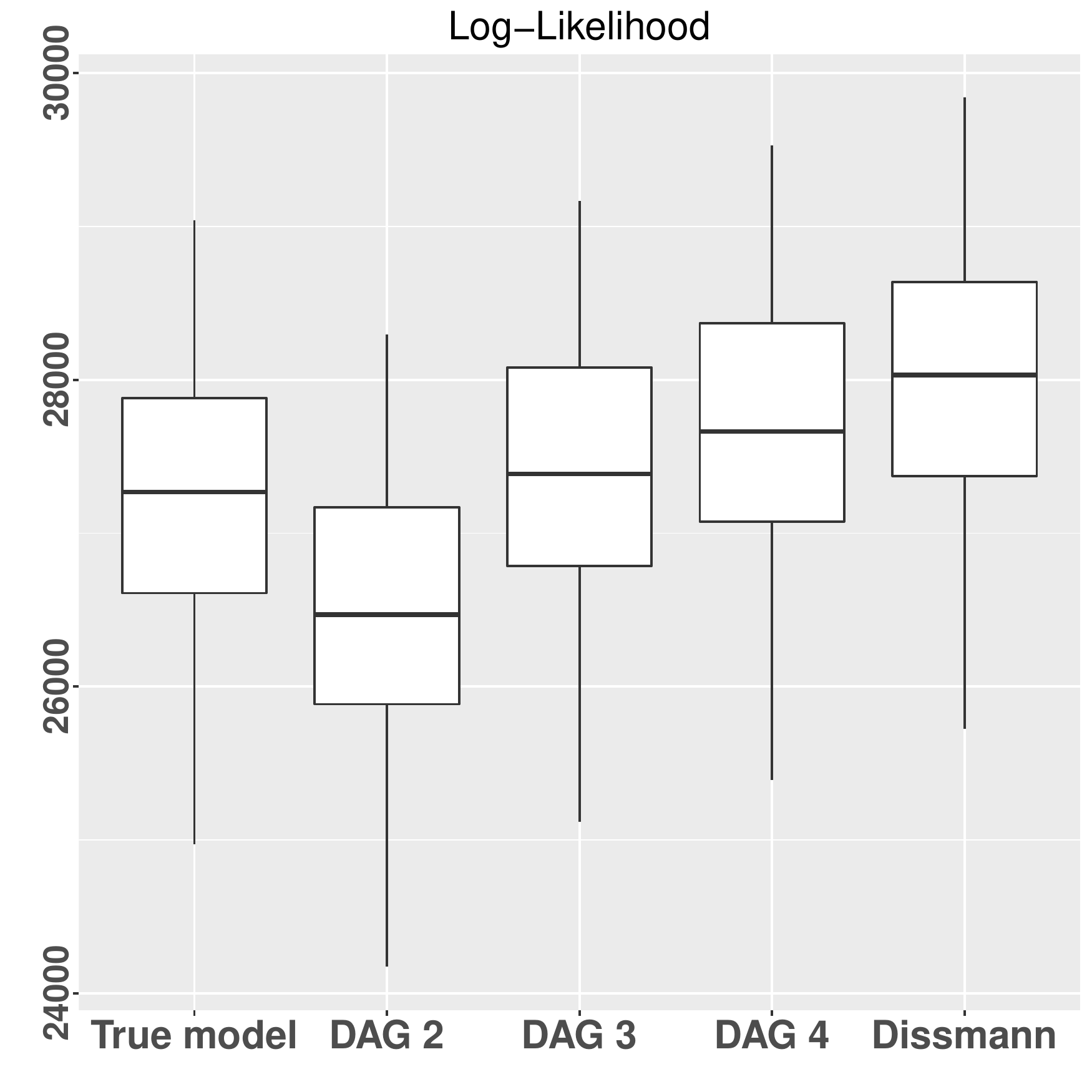}
	\includegraphics[width=0.24\textwidth, trim={0.1cm 0.1cm 0.1cm 0.1cm},clip]{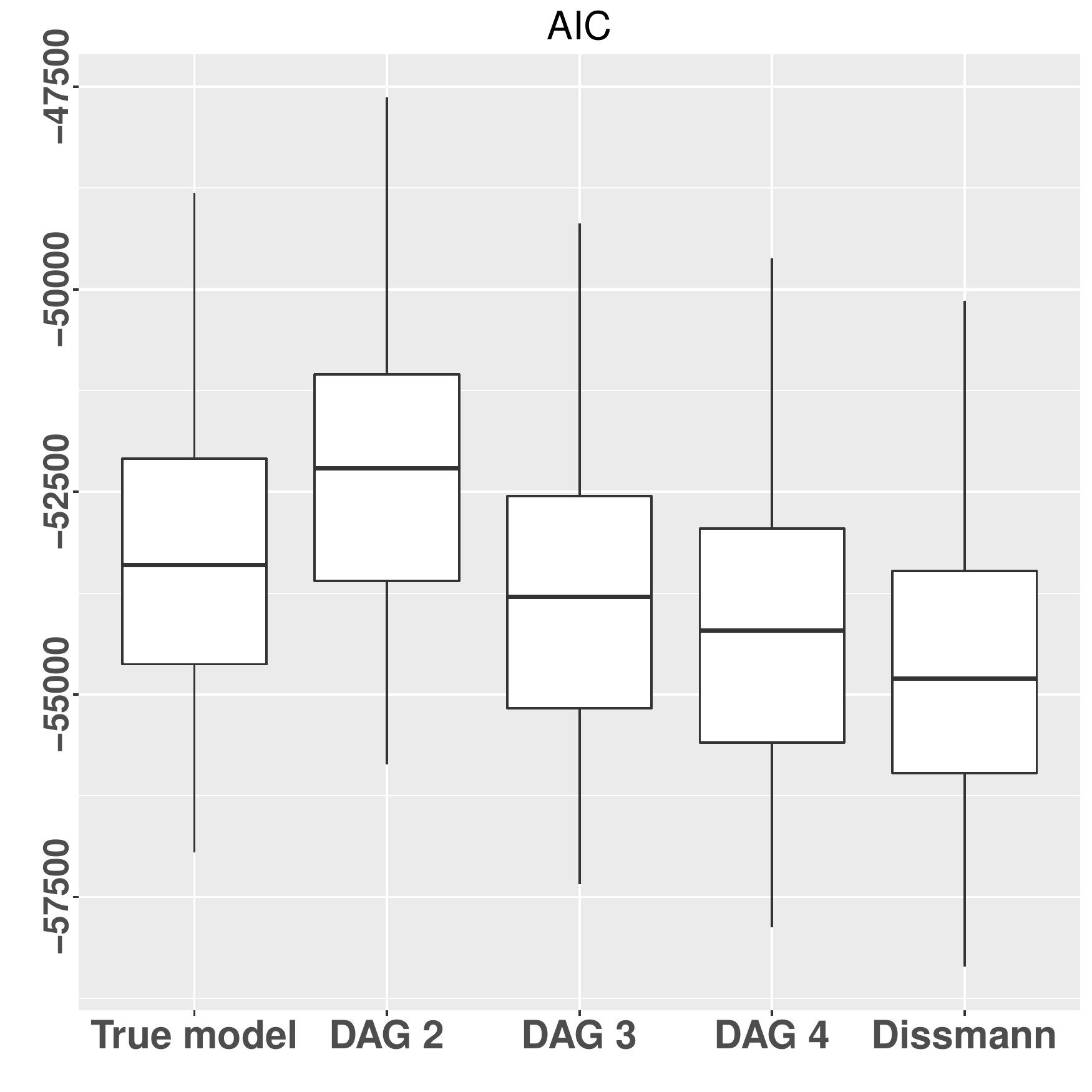}
	\includegraphics[width=0.24\textwidth, trim={0.1cm 0.1cm 0.1cm 0.1cm},clip]{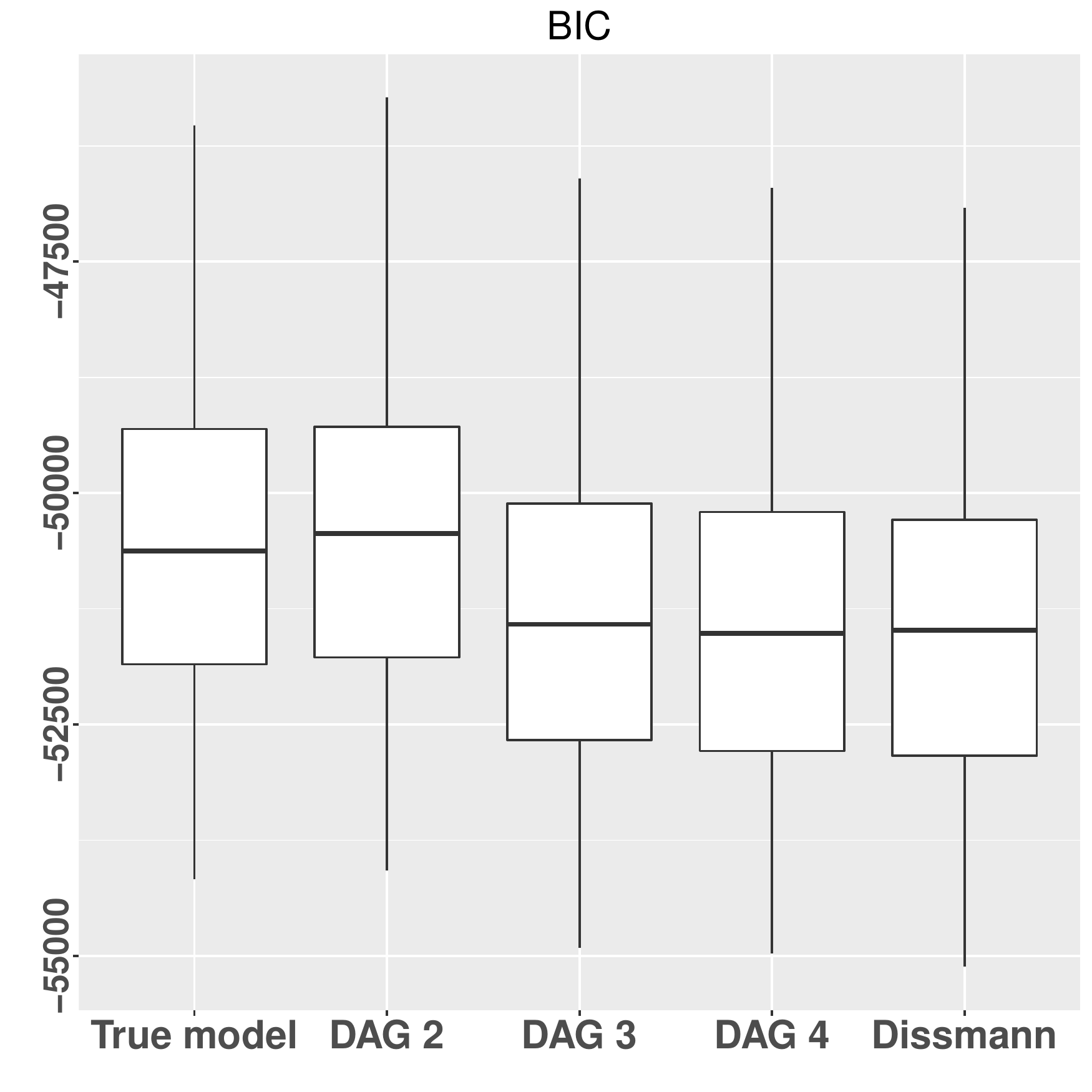}
	\includegraphics[width=0.24\textwidth, trim={0.1cm 0.1cm 0.1cm 0.1cm},clip]{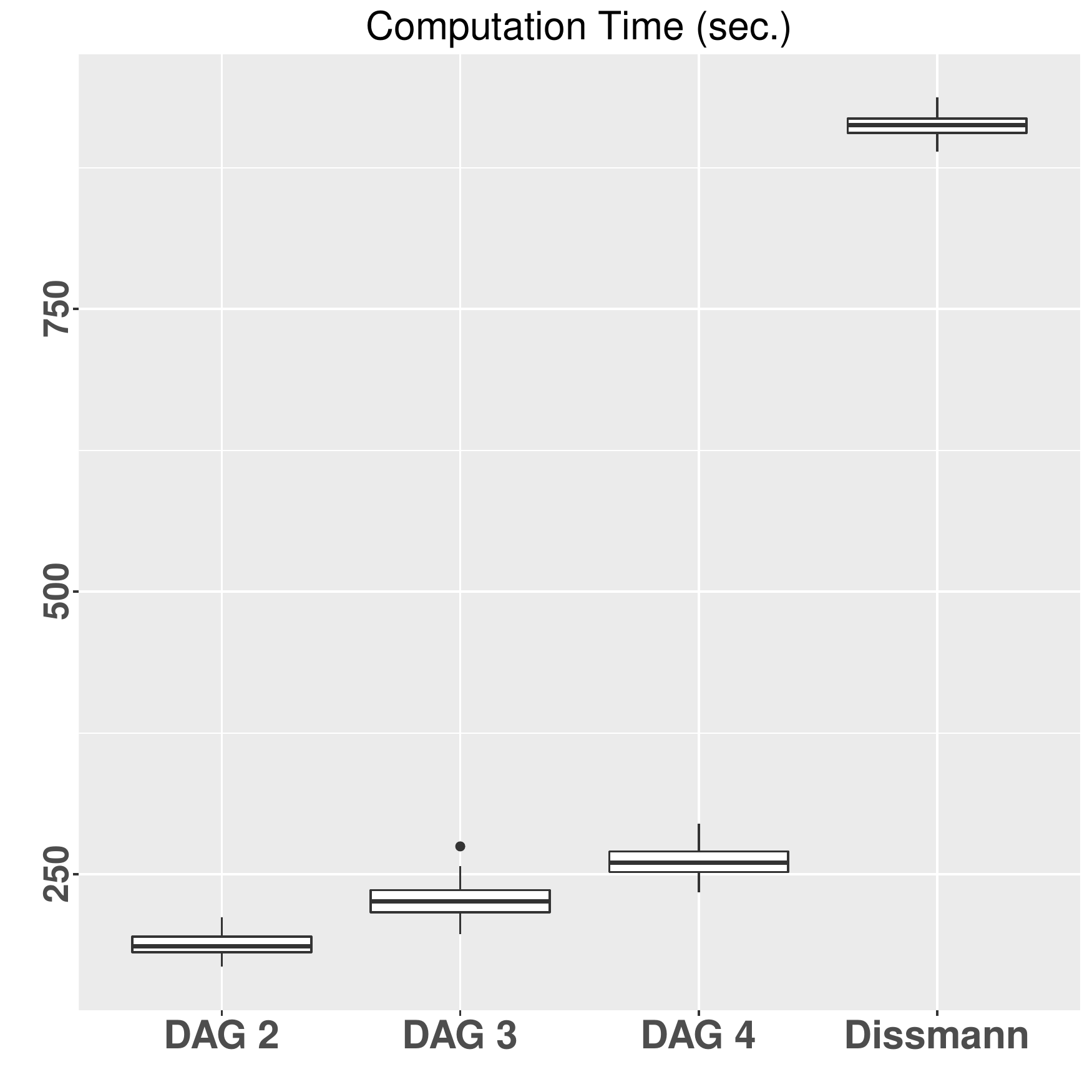}
	\caption{Scenario 3: Comparison of $k$-DAG representations for $k=2,3,4$ with Dissmann algorithm considering log-Likelihood, AIC, BIC and computation time in seconds on $100$ replications (from left to right).}
	\label{fig:simstudy:results3}
\end{figure}
\begin{figure}[H]
	\centering
	\includegraphics[width=0.24\textwidth, trim={0.1cm 0.1cm 0.1cm 0.1cm},clip]{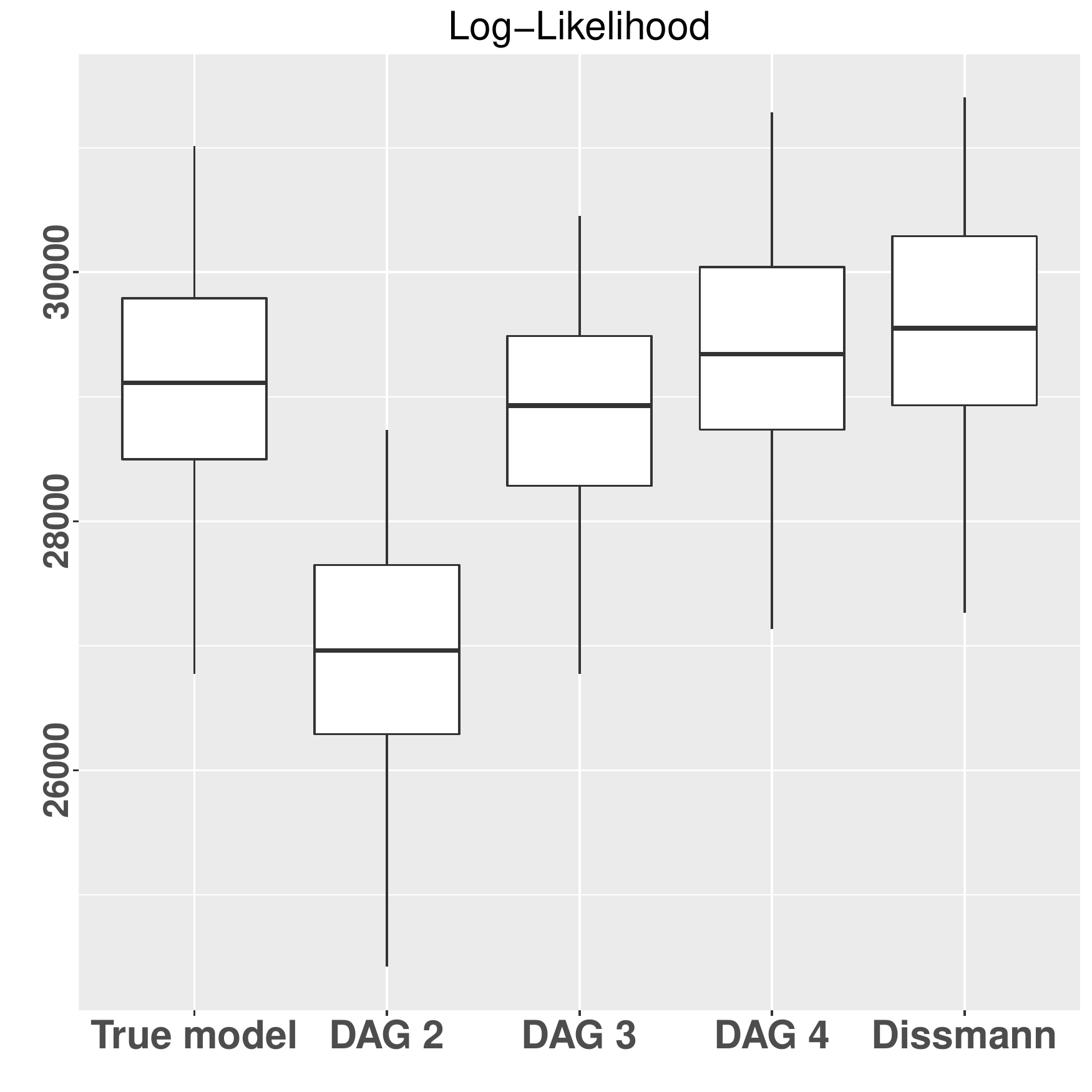}
	\includegraphics[width=0.24\textwidth, trim={0.1cm 0.1cm 0.1cm 0.1cm},clip]{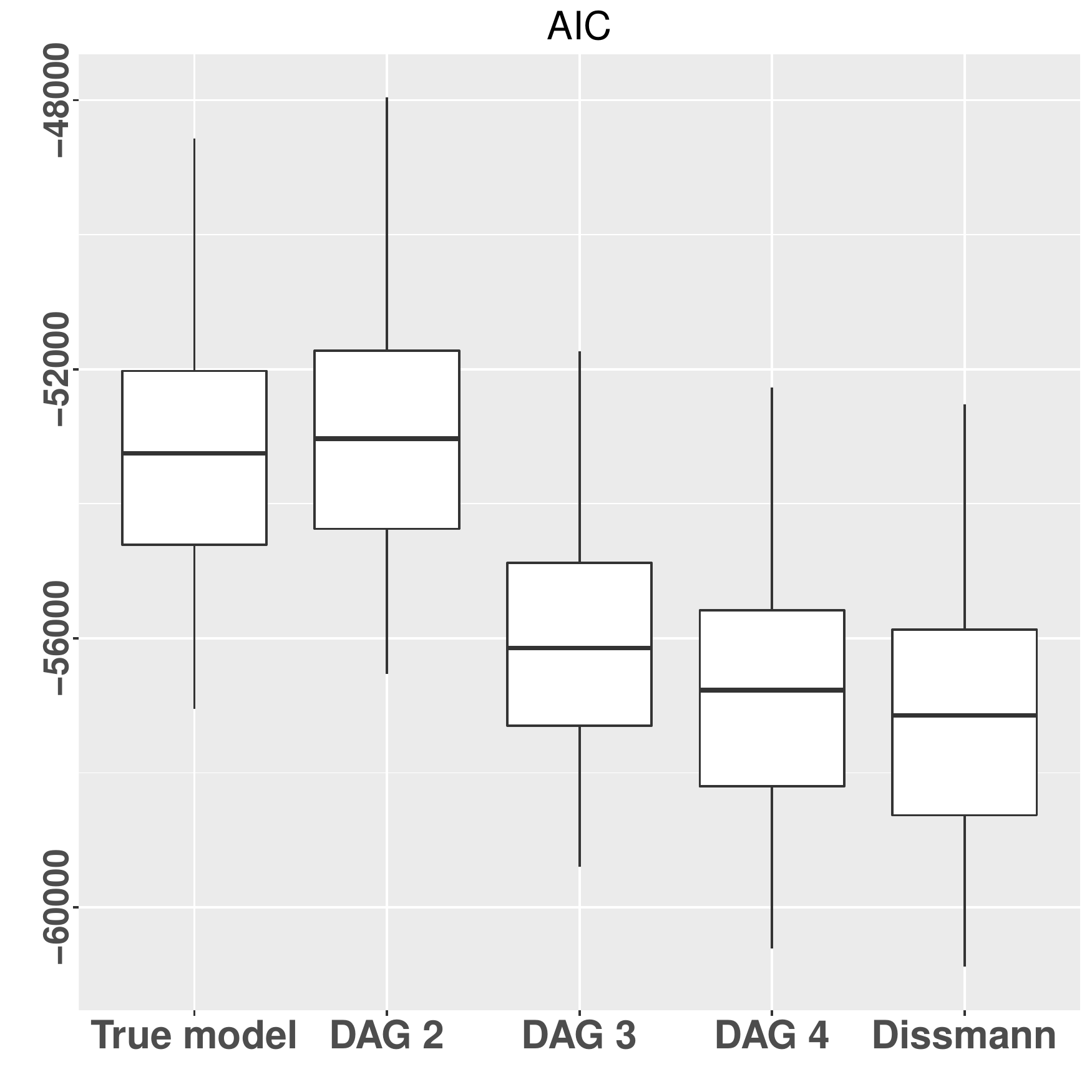}
	\includegraphics[width=0.24\textwidth, trim={0.1cm 0.1cm 0.1cm 0.1cm},clip]{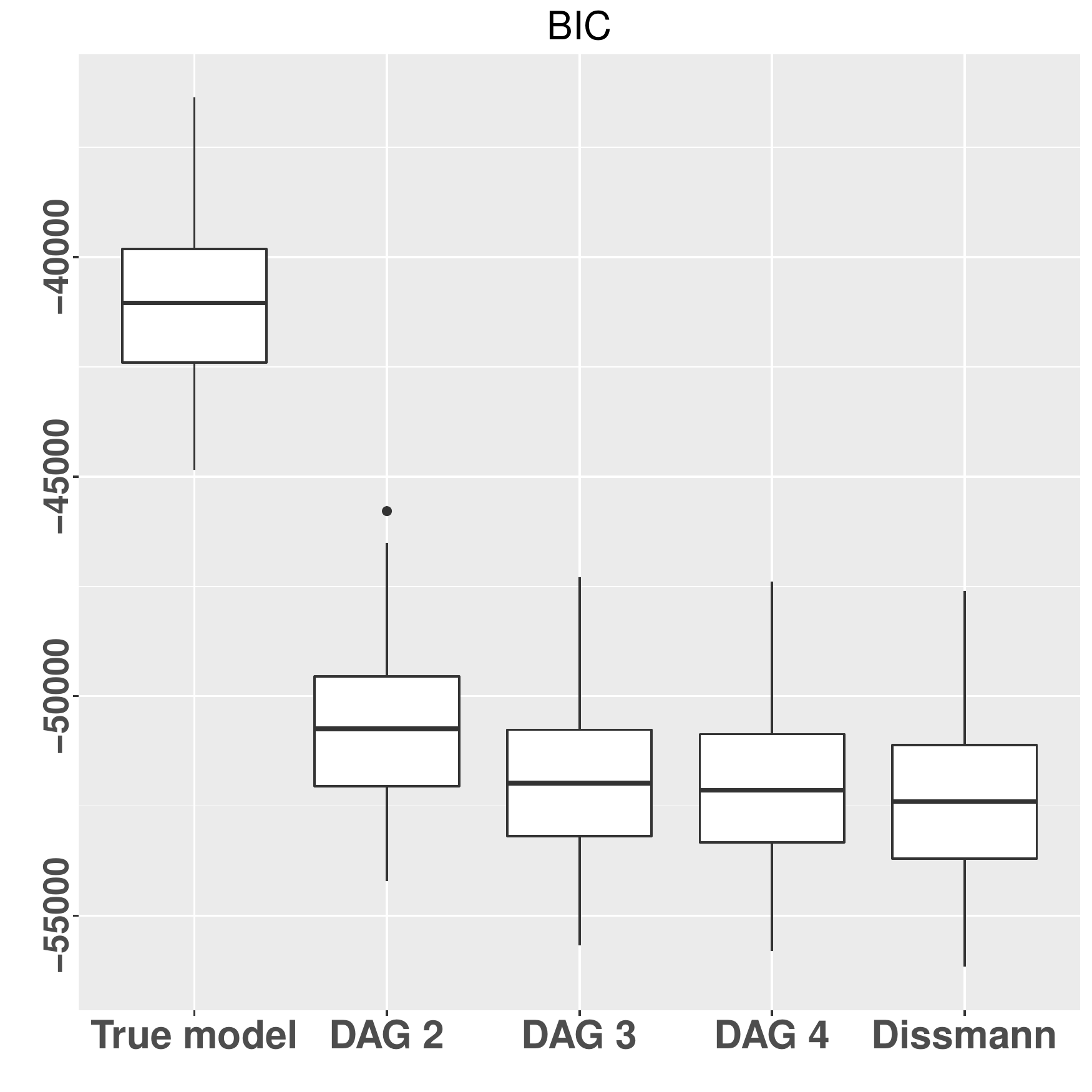}
	\includegraphics[width=0.24\textwidth, trim={0.1cm 0.1cm 0.1cm 0.1cm},clip]{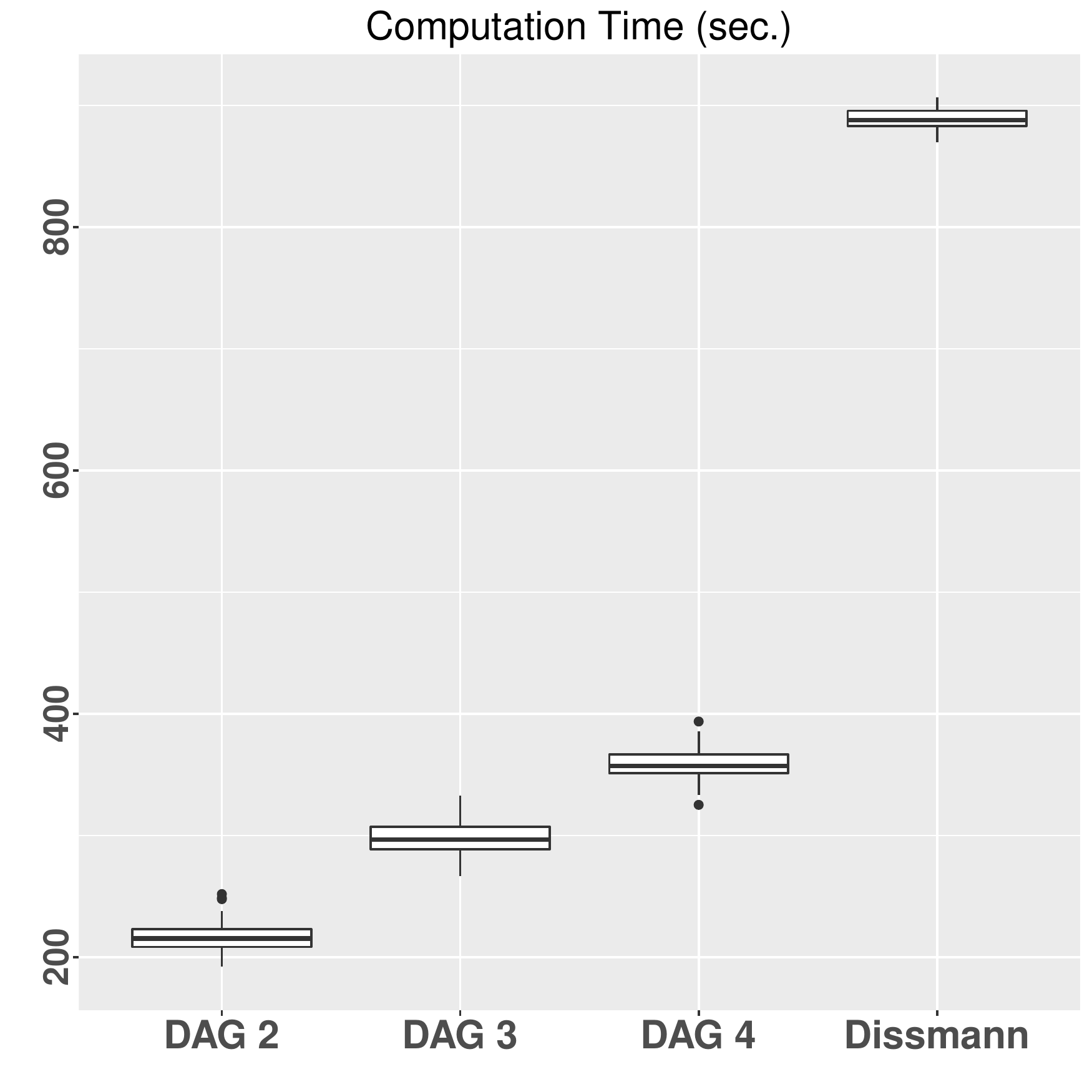}
	\caption{Scenario 4: Comparison of $k$-DAG representations for $k=2,3,4$ with Dissmann algorithm considering log-Likelihood, AIC, BIC and computation time in seconds on $100$ replications (from left to right).}
	\label{fig:simstudy:results4}
\end{figure}

\section{Supplementary material to application}\label{sec:appendix_application}
\subsection{DAGs estimated on Euro Stoxx 50}\label{subsec:dags}
\begin{figure}[H]
	\centering
	\includegraphics[width=0.6\textwidth,trim={0.0cm 1.0cm 0.0cm 1.0cm},clip]{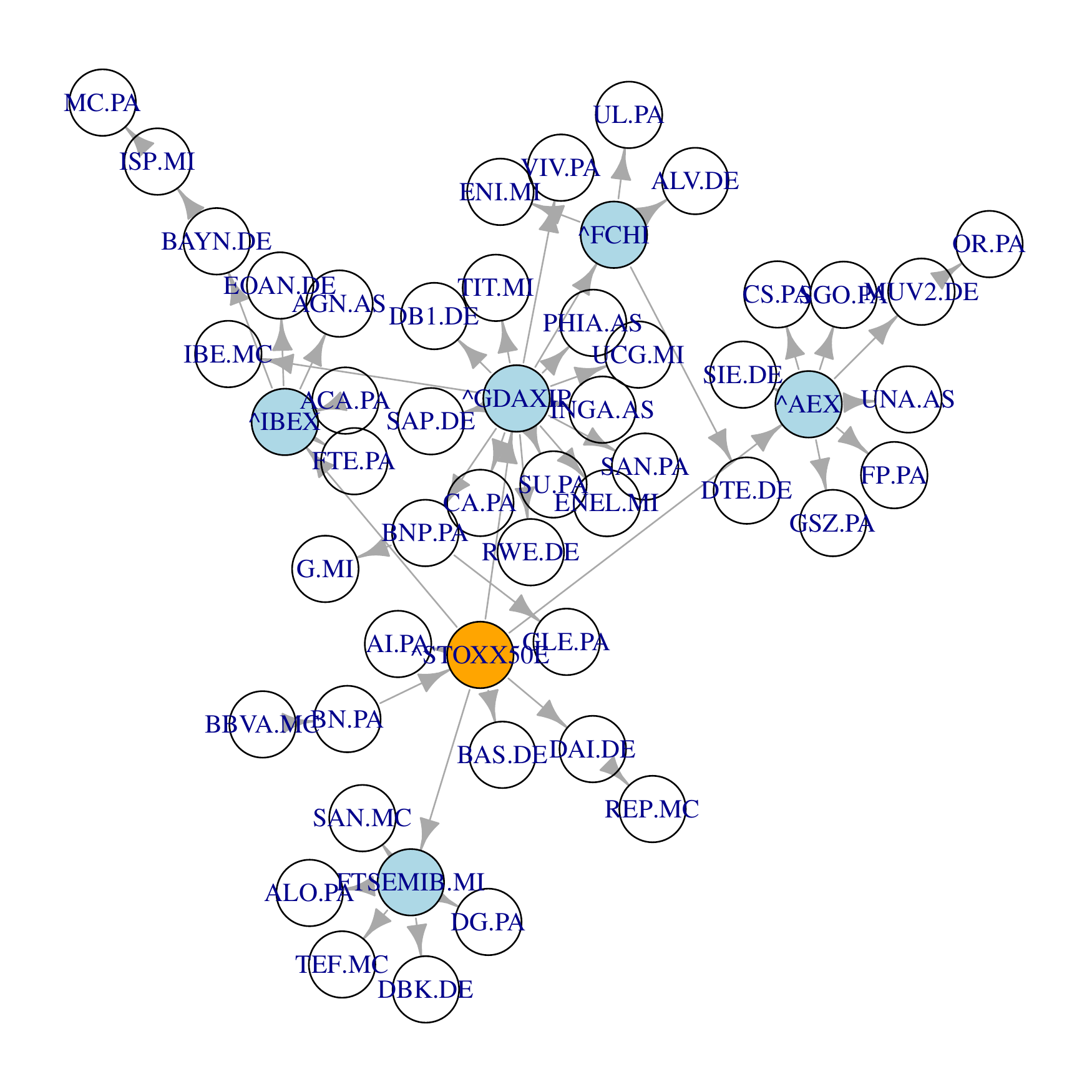}\\
	\includegraphics[width=0.6\textwidth,trim={0.0cm 1.0cm 0.0cm 1.0cm},clip]{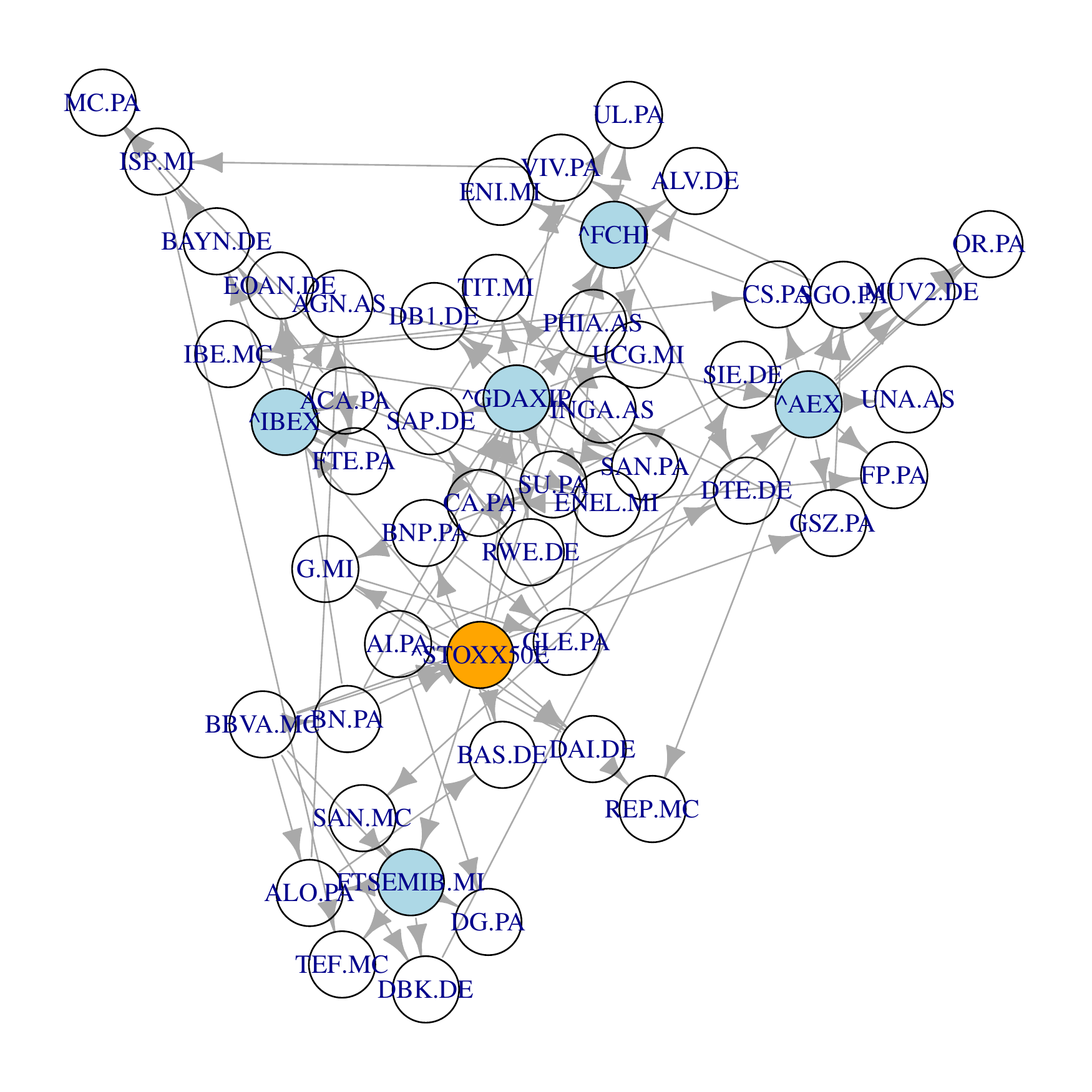}
	\caption{DAGs estimated on Euro Stoxx 50 with at most $k=1,2$ parents (upper, lower).}
	\label{fig:stoxx:dags1}
\end{figure}

\begin{figure}[H]
	\centering
	\includegraphics[width=0.75\textwidth,trim={0.0cm 1.0cm 0.0cm 1.0cm},clip]{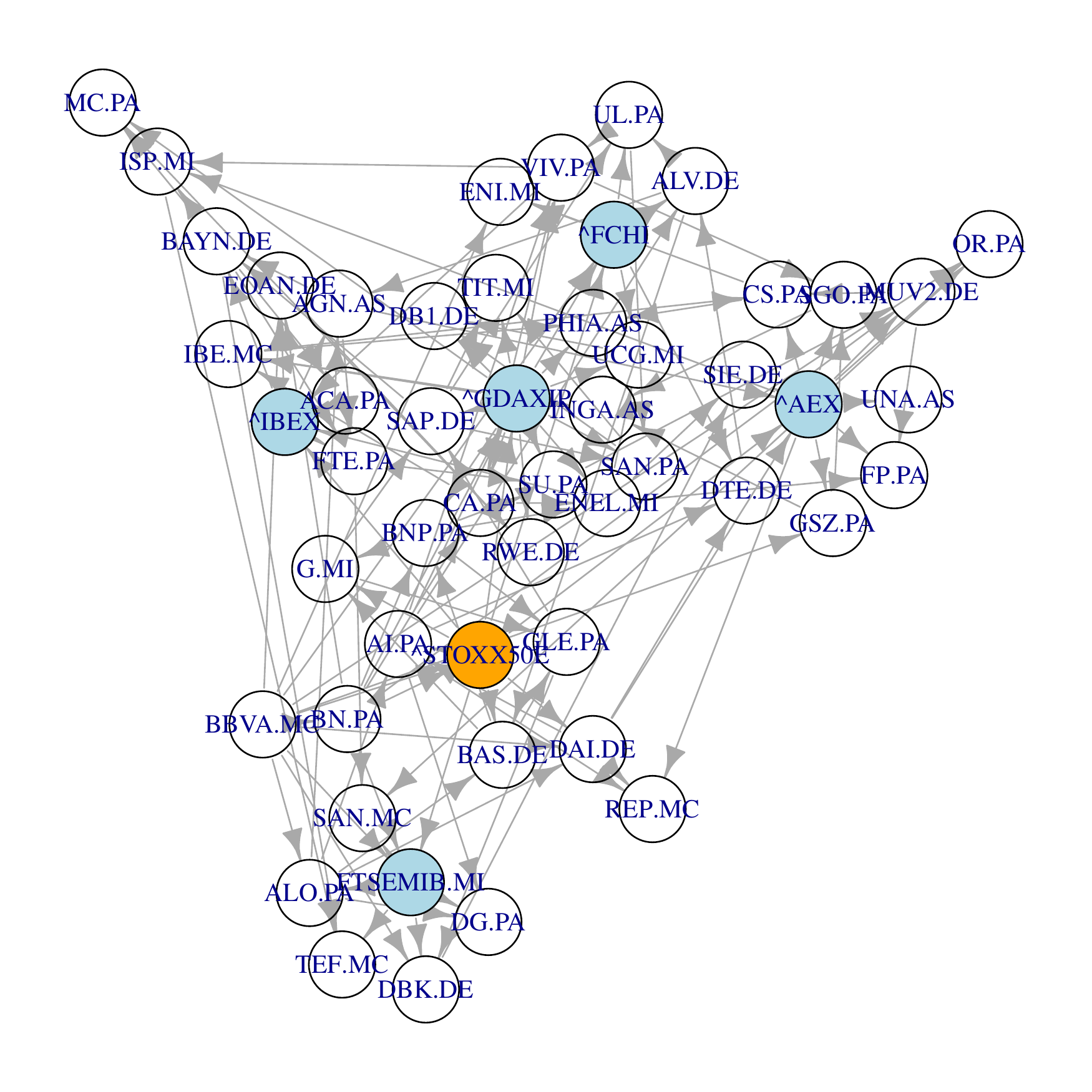}\\
	\includegraphics[width=0.75\textwidth,trim={0.0cm 1.0cm 0.0cm 1.0cm},clip]{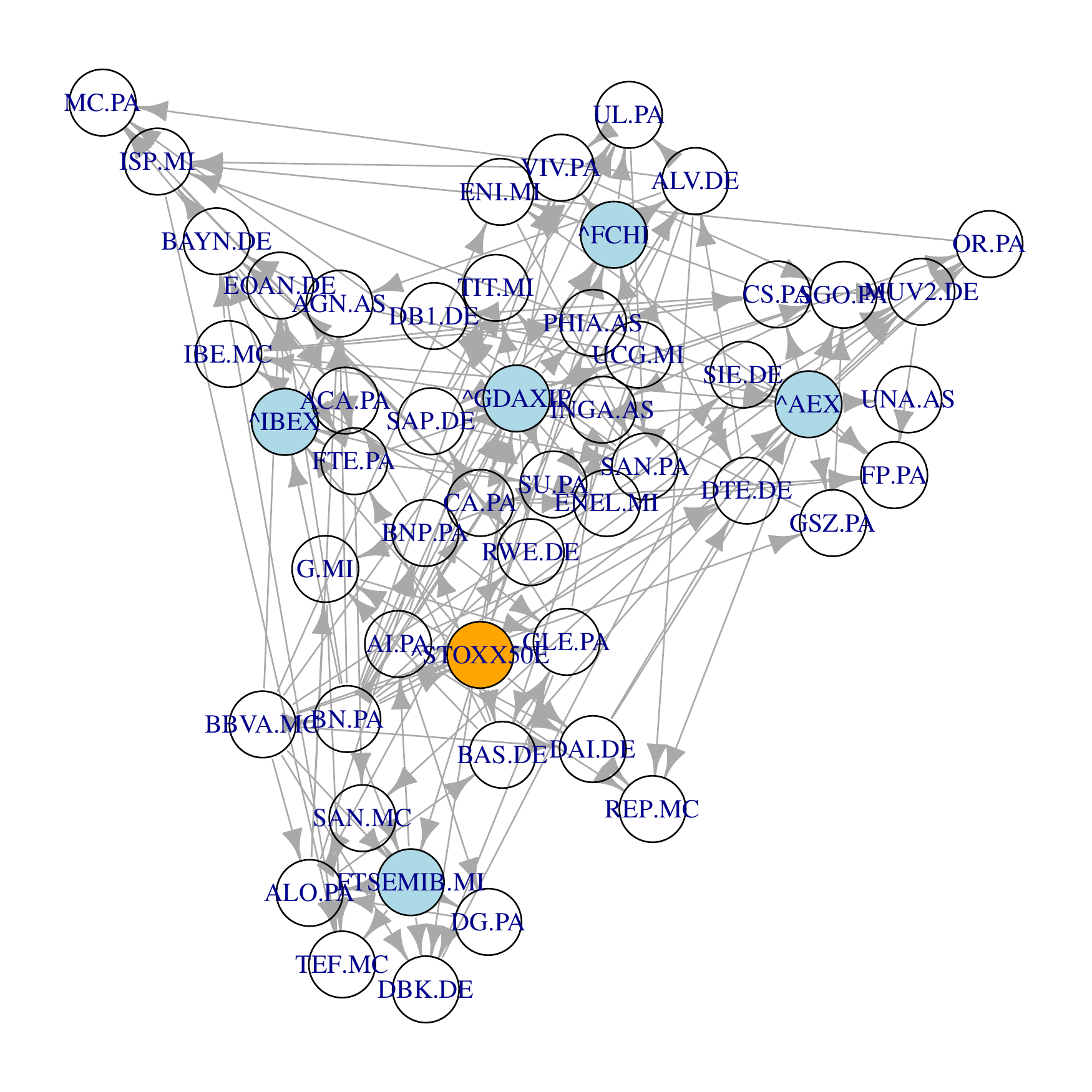}
	\caption{DAGs estimated on Euro Stoxx 50 with at most $k=3,4$ parents (upper, lower).}
	\label{fig:stoxx:dags2}
\end{figure}

\begin{figure}[H]
	\centering
	\includegraphics[width=0.75\textwidth,trim={0.0cm 1.0cm 0.0cm 1.0cm},clip]{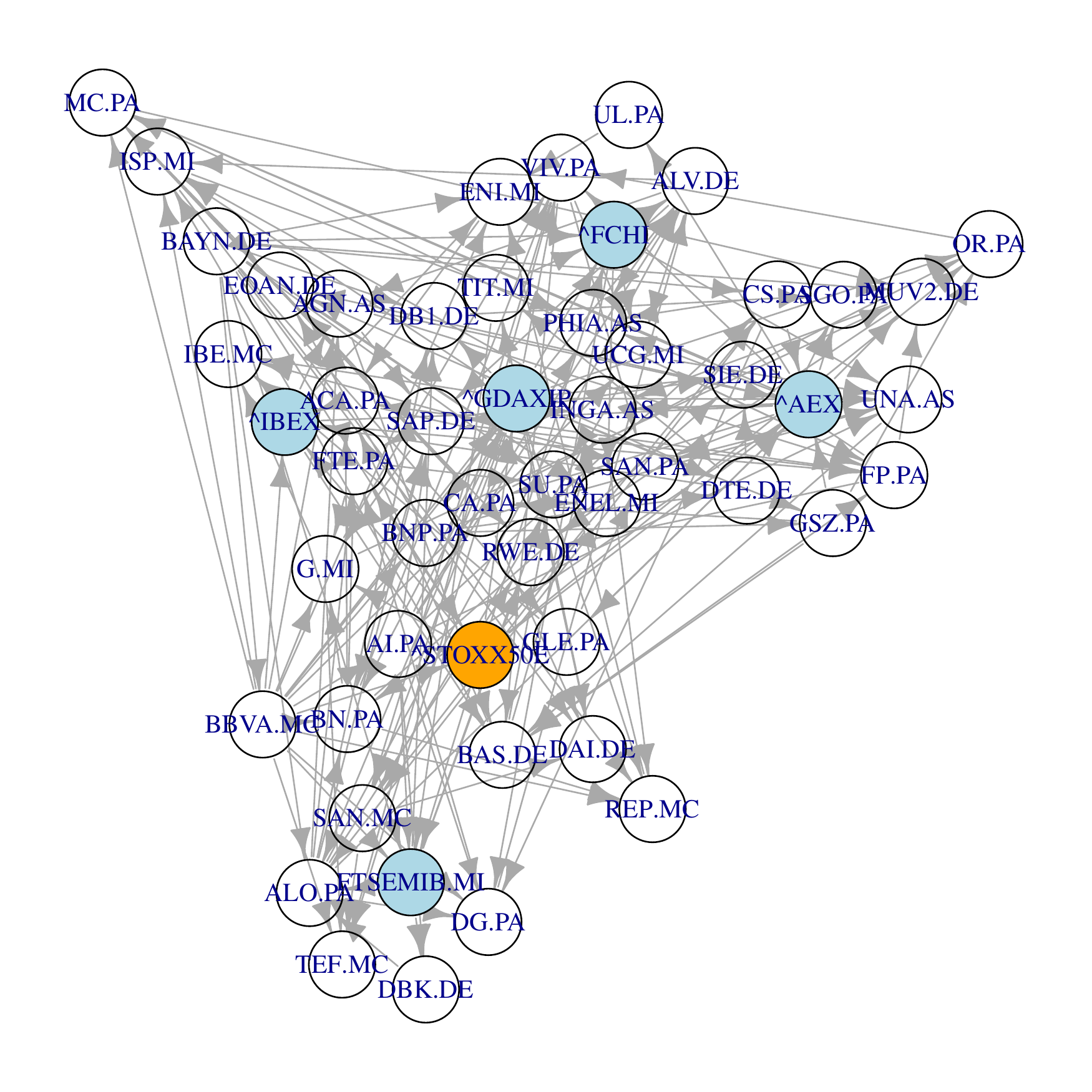}\\
	\includegraphics[width=0.75\textwidth,trim={0.0cm 1.0cm 0.0cm 1.0cm},clip]{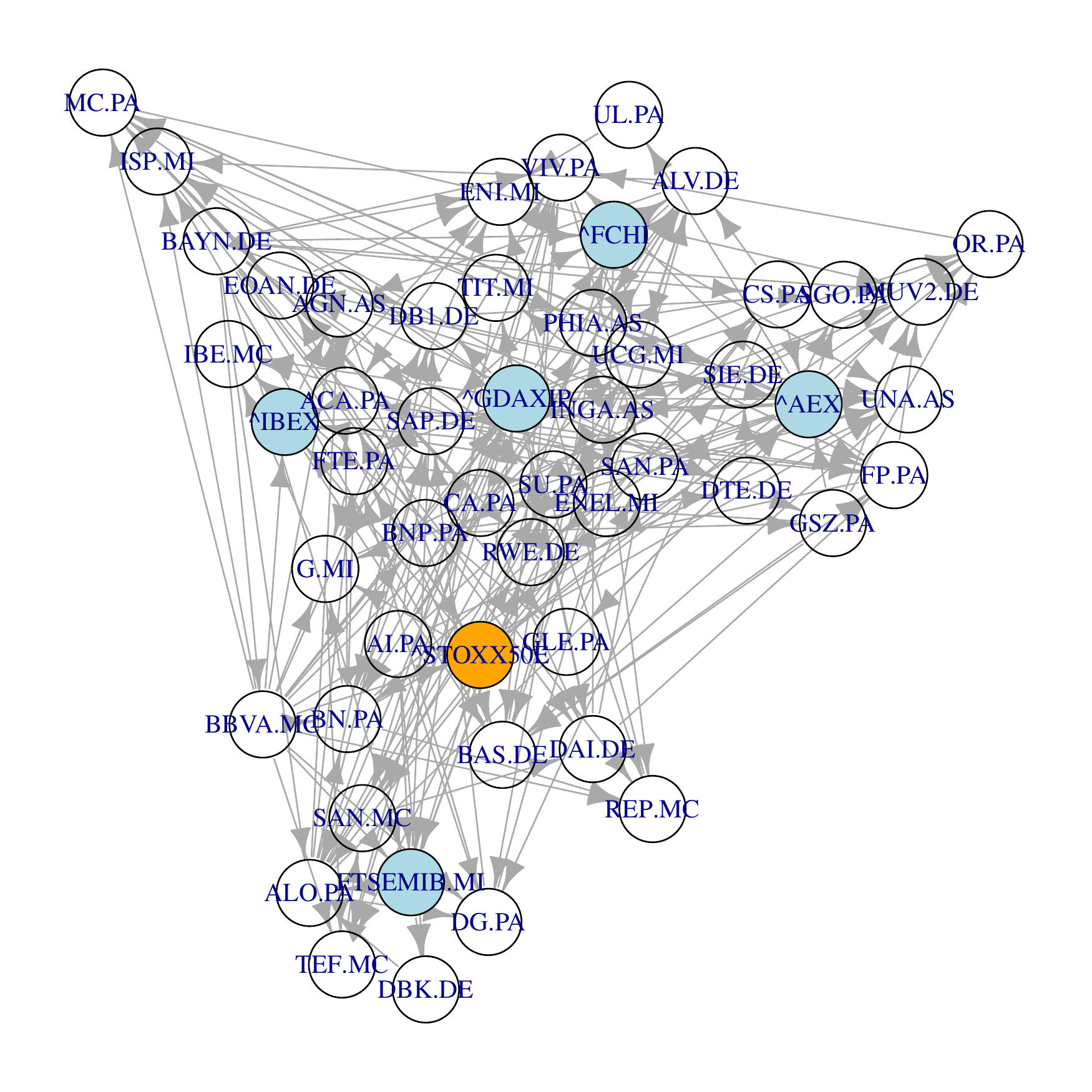}
	\caption{DAGs estimated on Euro Stoxx 50 with at most $k=5,6$ parents (upper, lower).}
	\label{fig:stoxx:dags3}
\end{figure}

\begin{figure}[H]
	\centering
	\includegraphics[width=0.75\textwidth,trim={0.0cm 1.0cm 0.0cm 1.0cm},clip]{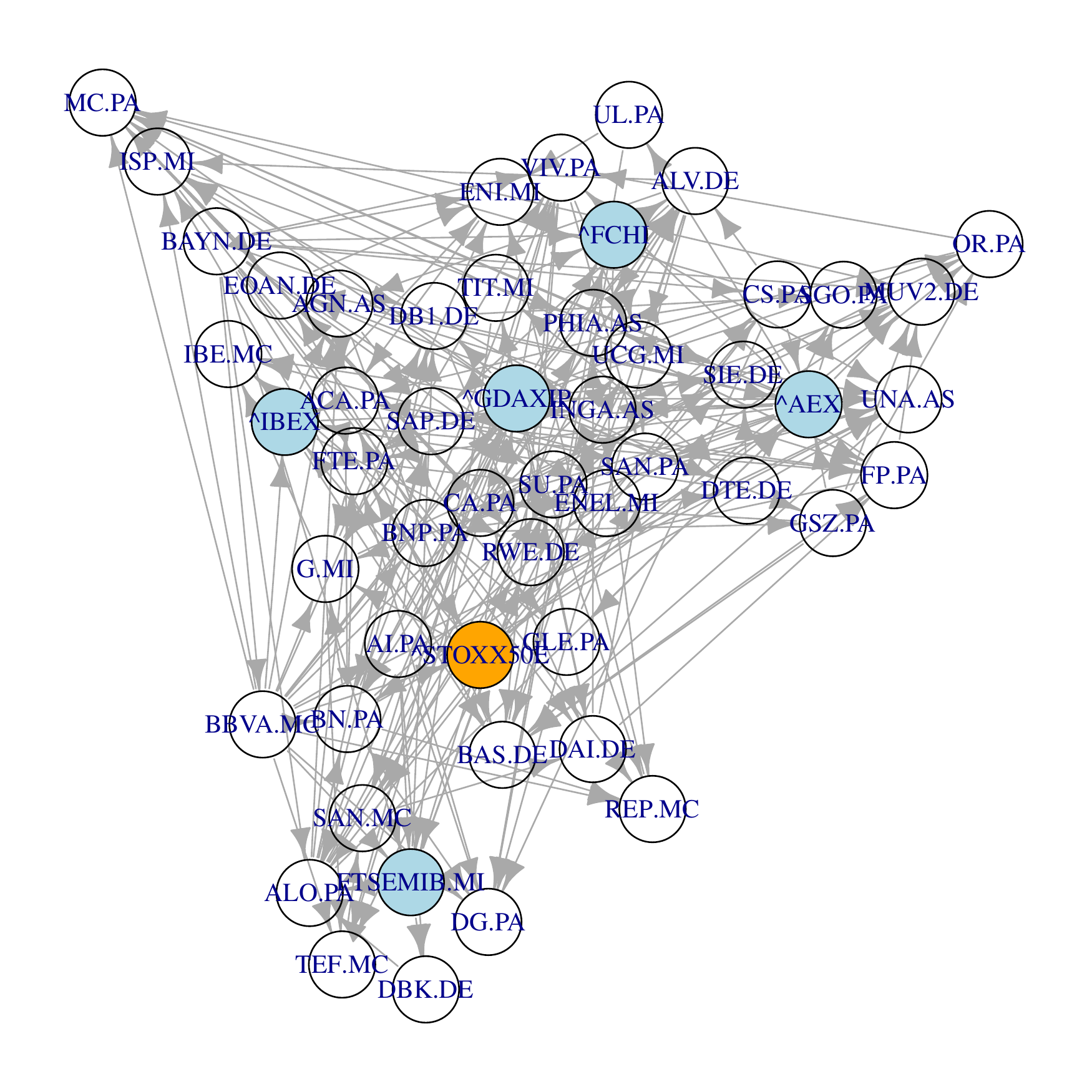}\\
	\includegraphics[width=0.75\textwidth,trim={0.0cm 1.0cm 0.0cm 1.0cm},clip]{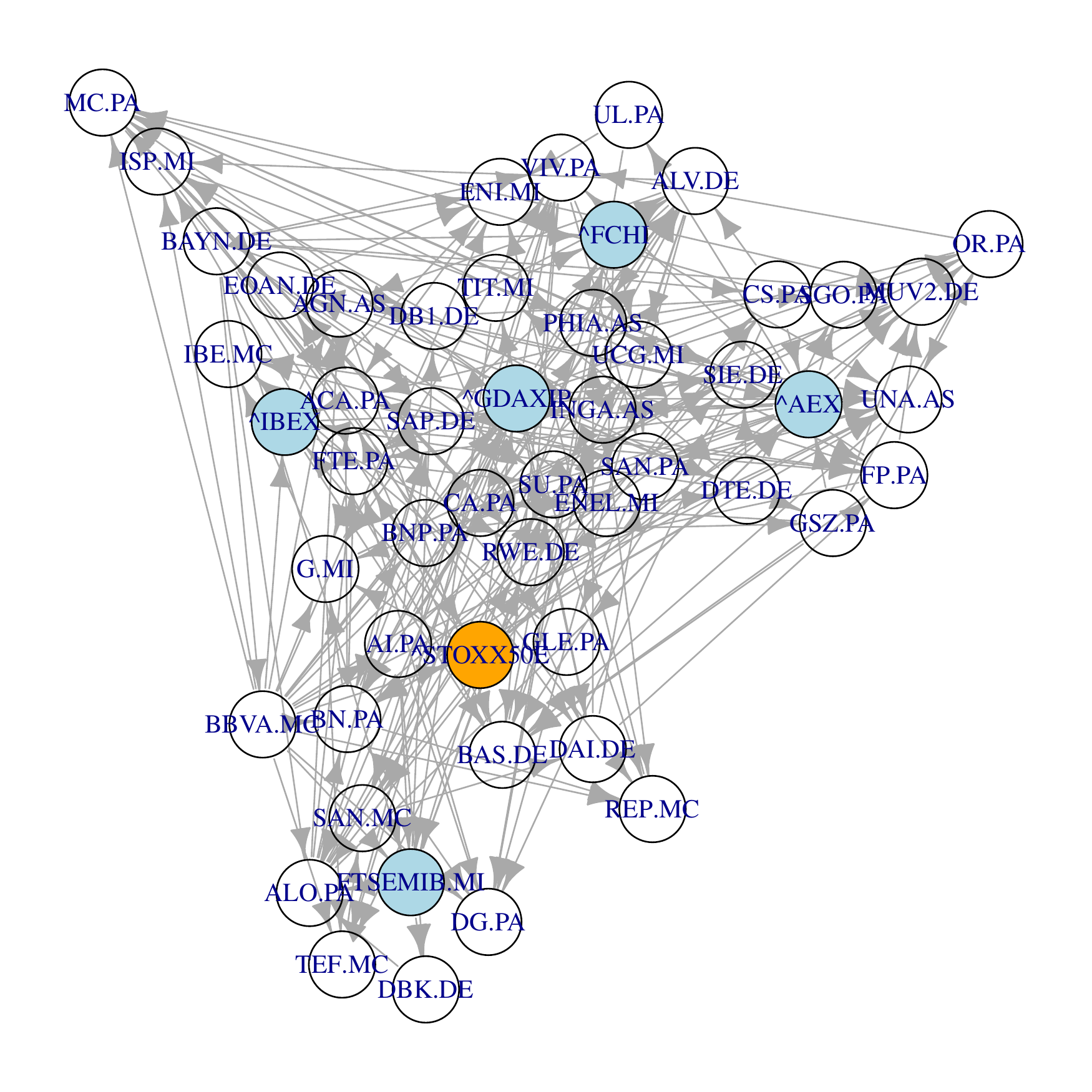}
	\caption{DAGs estimated on Euro Stoxx 50 with at most $k=7,8$ parents.}
	\label{fig:stoxx:dags4}
\end{figure}

\begin{figure}[H]
	\centering
	\includegraphics[width=0.75\textwidth,trim={0.0cm 1.0cm 0.0cm 1.0cm},clip]{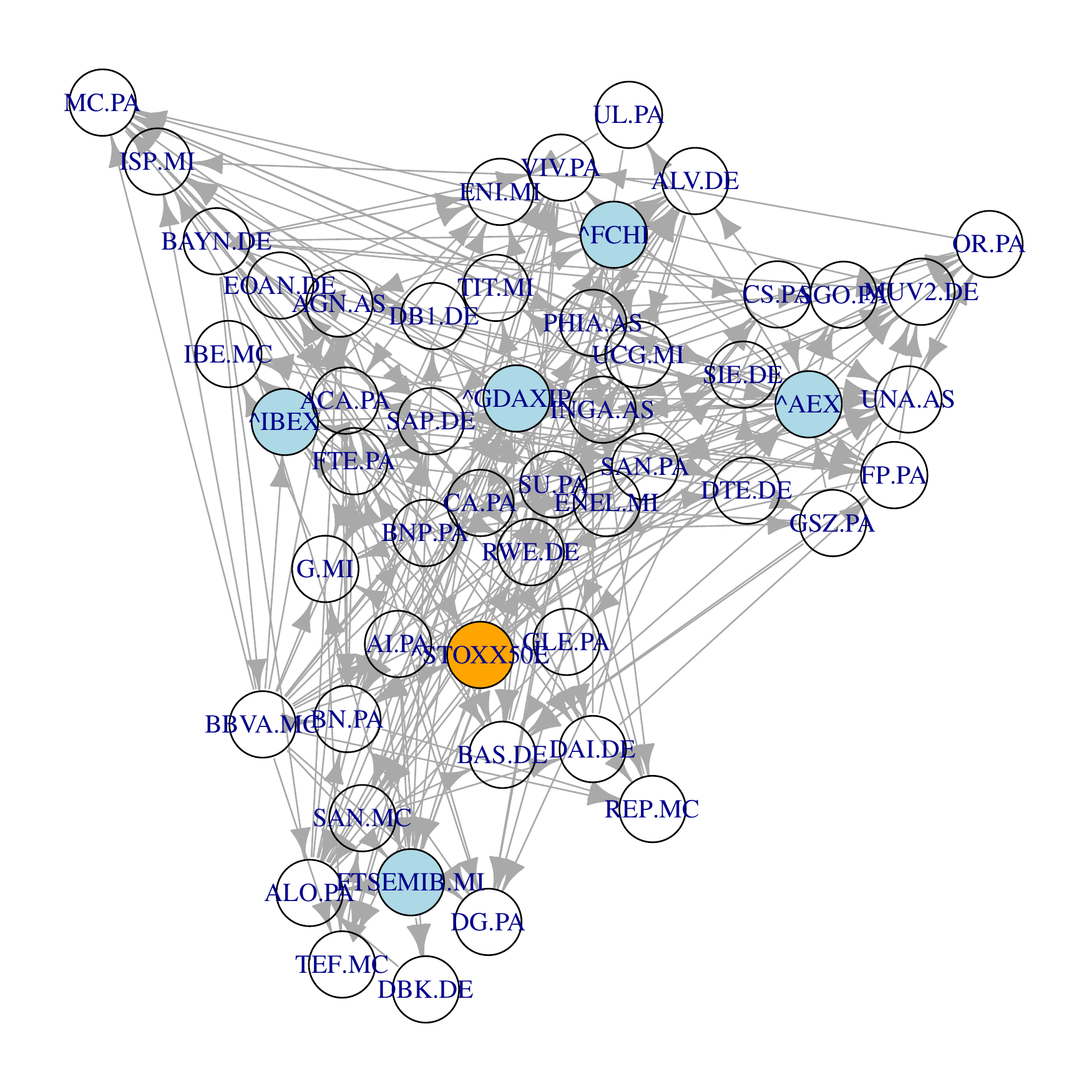}\\
	\includegraphics[width=0.75\textwidth,trim={0.0cm 1.0cm 0.0cm 1.0cm},clip]{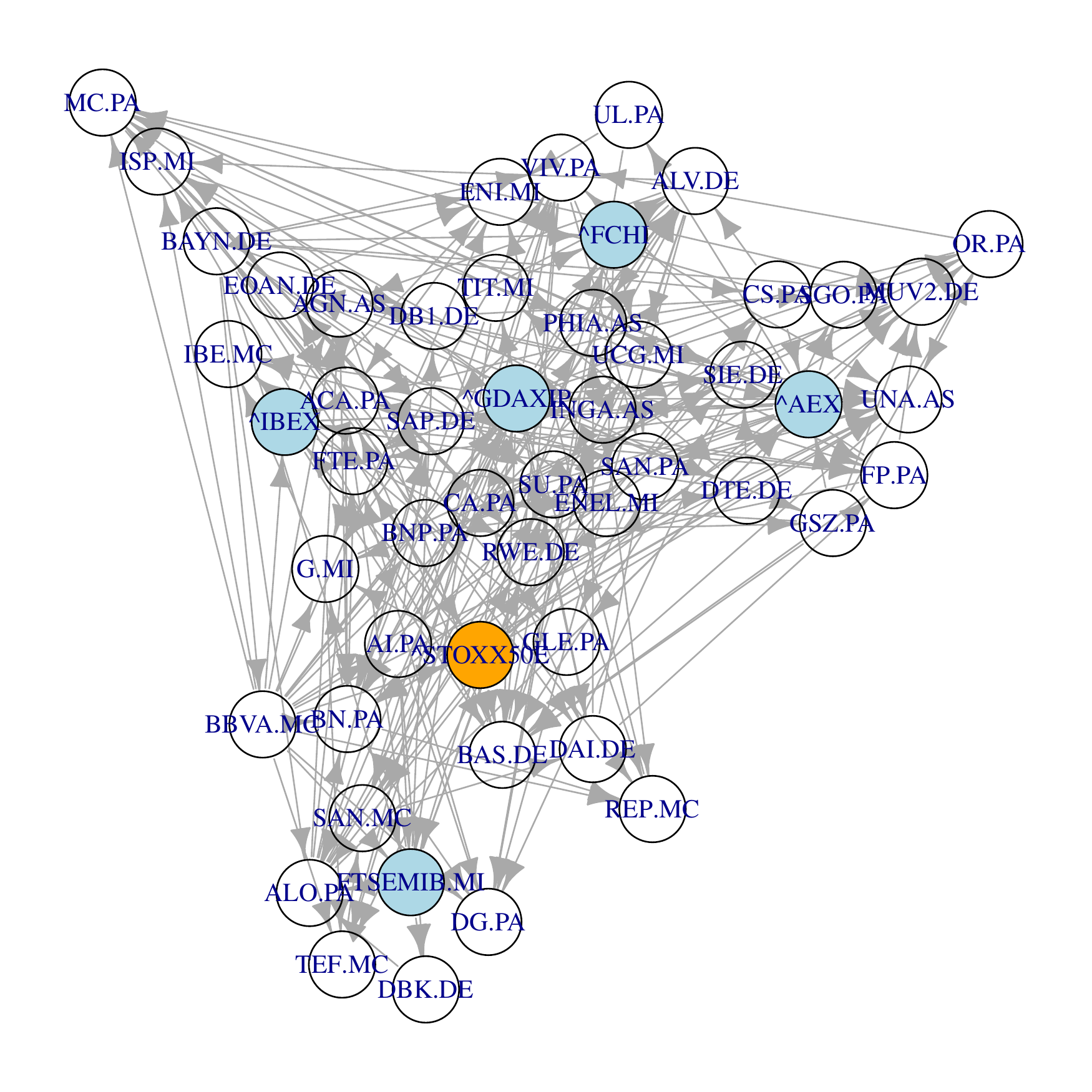}
	\caption{DAGs estimated on Euro Stoxx 50 with at most $k=9,10$ parents.}
	\label{fig:stoxx:dags5}
\end{figure}

\newpage

\begin{figure}[ht]
	\centering
	\includegraphics[width=1\textwidth,trim={0cm 1cm 0cm 1cm},clip]{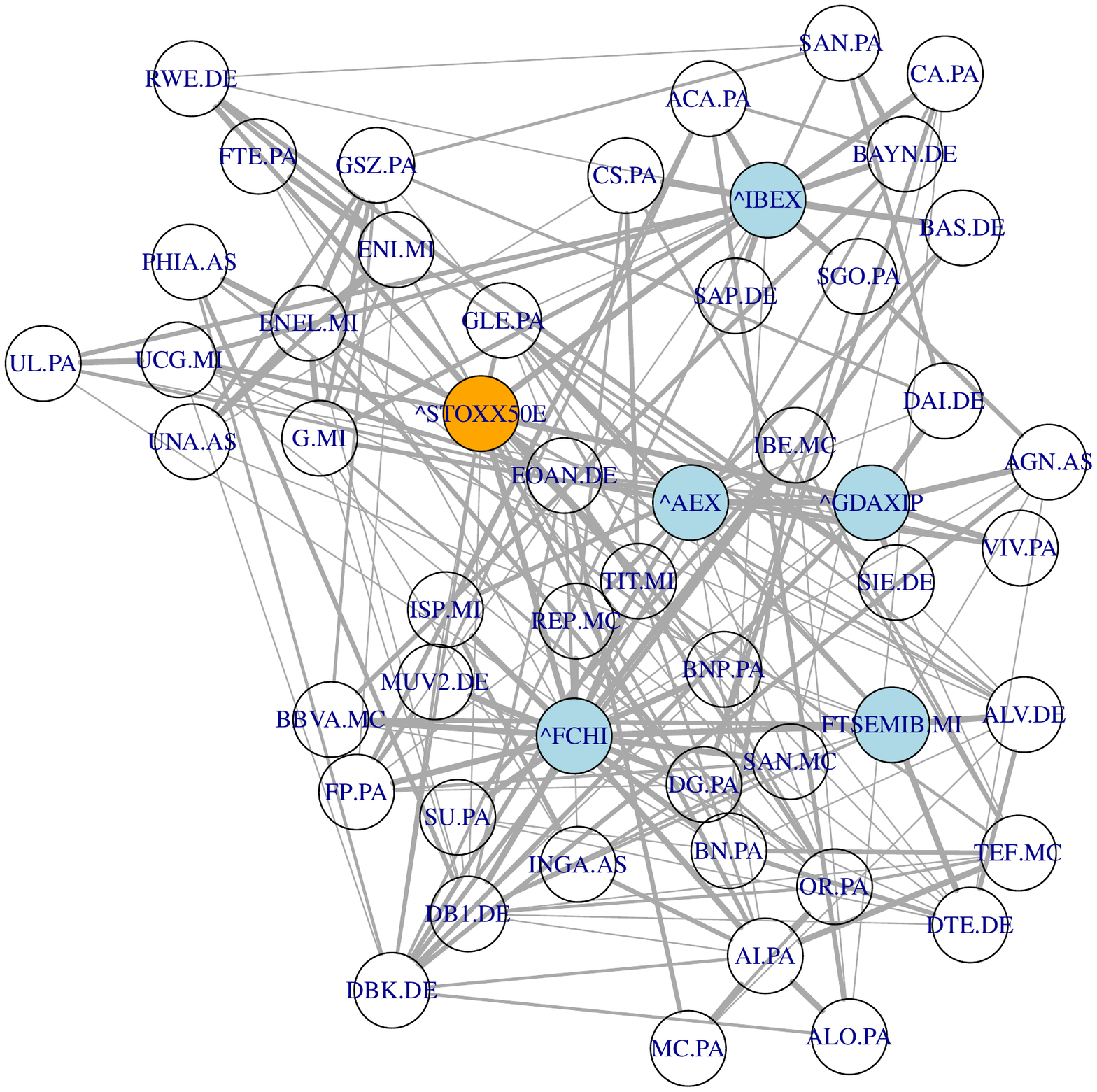}
	\caption{Undirected weighted graph $\HHH$ defined by the union of $\GGG_1^s,\dots,\GGG_4^s$ with weights calculated by \eqref{eq:edgeweight}, see Section \ref{subsec:generaldag2}. The weight of the edges is represented by the line width and illustrates how often these edges occurred in the skeletons $\GGG_1^s,\dots,\GGG_4^s$. We observe strong dependence especially among and between the national stocks indices and the Euro Stoxx 50 index itself. This confirms our expectation that our approach captures the most important relationships in the data. }
	\label{fig:stoxx:hgraph}
\end{figure}

\subsection{Numerical results of fitted models}\label{subsec:appendix_application_tables}
\begin{table}[H]
	\centering
	\begin{tabular}{R{0.7cm}R{1.2cm}R{1.2cm}R{0.8cm}R{1.2cm}R{1.2cm}R{1.2cm}R{1.2cm}R{1.2cm}R{0.8cm}R{0.5cm}R{0.9cm}}
		
		\multicolumn{4}{c}{DAG} & \multicolumn{8}{c}{R-vine representation of DAG}\\
		\hline \hline
		Max. parents & No. par. & log-Lik. & BIC & No. par & No. ni-pc & No. G-pc & No. non-G-pc & log-Lik. & BIC & $k^\prime$ & time (sec.) \\ 
		\hline
		1 & 155 & -47138 & 95344 & 206 & 51 & 0 & 51 & -45880 & 93180 & 1 & 124 \\ 
		2 & 204 & -45365 & 92135 & 405 & 236 & 16 & 220 & -42859 & 88509 & 47 & 197 \\ 
		3 & 250 & -44731 & 91186 & 522 & 401 & 50 & 351 & -41661 & 86919 & 51 & 223 \\ 
		4 & 280 & -44448 & 90826 & 531 & 429 & 59 & 370 & -41492 & 86644 & 47 & 246 \\ 
		5 & 309 & -44224 & 90577 & 536 & 435 & 53 & 382 & -41455 & 86605 & 51 & 255 \\ 
		6 & 330 & -44104 & 90482 & 540 & 438 & 56 & 382 & -41435 & 86592 & 48 & 279 \\ 
		7 & 341 & -44045 & 90440 & 546 & 435 & 52 & 383 & -41418 & 86599 & 47 & 280 \\ 
		8 & 345 & -44026 & 90429 & 542 & 433 & 55 & 378 & -41431 & 86598 & 48 & 276 \\ 
		9 & 349 & -44006 & 90418 & 540 & 427 & 57 & 370 & -41436 & 86594 & 49 & 282 \\ 
		10 & 353 & -43990 & 90412 & 534 & 422 & 58 & 364 & -41437 & 86554 & 47 & 271 \\ 
		\hline
	\end{tabular}
	\caption{Numerical results for DAG and DAG representations. Calculations based on \textit{z-scale}, abbreviations \textit{ni-pc} for \textit{non independence pair copula}, \textit{G-pc} for \textit{Gaussian pair copula}.}
	\label{table:application:dag}
\end{table}

\newpage

%\begin{table}
%	\centering
\begin{longtable}{R{1.2cm}R{1.2cm}R{1.2cm}R{1.2cm}R{1.2cm}R{1.2cm}R{1.2cm}R{1.2cm}}
	trunc. level & No. par & No. ni-pc & No. G-pc & No. non-G-pc & log-Lik. & BIC & time (sec.) \\ 
	\hline \hline
	1 & 206 & 51 & 0 & 51 & -45808 & 93035 & 633 \\ 
	2 & 261 & 99 & 2 & 97 & -44445 & 90688 & 641 \\ 
	3 & 296 & 126 & 5 & 121 & -44120 & 90281 & 674 \\ 
	4 & 320 & 152 & 8 & 144 & -43899 & 90003 & 728 \\ 
	5 & 331 & 170 & 13 & 157 & -43750 & 89781 & 726 \\ 
	6 & 340 & 182 & 14 & 168 & -43628 & 89599 & 741 \\ 
	7 & 358 & 197 & 14 & 183 & -43422 & 89312 & 715 \\ 
	8 & 377 & 217 & 17 & 200 & -43194 & 88986 & 722 \\ 
	9 & 389 & 231 & 18 & 213 & -43101 & 88884 & 723 \\ 
	10 & 402 & 245 & 19 & 226 & -43037 & 88844 & 729 \\ 
	11 & 412 & 259 & 20 & 239 & -42950 & 88741 & 724 \\ 
	12 & 419 & 269 & 23 & 246 & -42882 & 88653 & 739 \\ 
	13 & 425 & 279 & 25 & 254 & -42781 & 88491 & 724 \\ 
	14 & 432 & 290 & 25 & 265 & -42654 & 88286 & 726 \\ 
	15 & 439 & 299 & 26 & 273 & -42566 & 88157 & 727 \\ 
	16 & 451 & 313 & 28 & 285 & -42420 & 87949 & 732 \\ 
	17 & 464 & 321 & 29 & 292 & -42303 & 87805 & 732 \\ 
	18 & 481 & 335 & 30 & 305 & -42170 & 87656 & 731 \\ 
	19 & 491 & 344 & 31 & 313 & -42126 & 87636 & 723 \\ 
	20 & 500 & 353 & 33 & 320 & -42023 & 87491 & 725 \\ 
	21 & 507 & 361 & 33 & 328 & -41969 & 87433 & 728 \\ 
	22 & 508 & 365 & 34 & 331 & -41933 & 87367 & 726 \\ 
	23 & 509 & 371 & 37 & 334 & -41907 & 87323 & 716 \\ 
	24 & 510 & 376 & 39 & 337 & -41886 & 87288 & 737 \\ 
	25 & 517 & 382 & 42 & 340 & -41839 & 87242 & 713 \\ 
	26 & 520 & 387 & 44 & 343 & -41808 & 87200 & 719 \\ 
	27 & 525 & 392 & 47 & 345 & -41730 & 87078 & 711 \\ 
	28 & 526 & 394 & 47 & 347 & -41723 & 87071 & 728 \\ 
	29 & 526 & 396 & 47 & 349 & -41710 & 87046 & 714 \\ 
	30 & 528 & 398 & 47 & 351 & -41695 & 87029 & 726 \\ 
	31 & 531 & 401 & 48 & 353 & -41678 & 87017 & 730 \\ 
	32 & 533 & 404 & 48 & 356 & -41668 & 87009 & 726 \\ 
	33 & 535 & 406 & 48 & 358 & -41657 & 87001 & 738 \\ 
	34 & 539 & 409 & 48 & 361 & -41647 & 87008 & 737 \\ 
	35 & 540 & 410 & 49 & 361 & -41643 & 87009 & 736 \\ 
	36 & 541 & 411 & 50 & 361 & -41638 & 87006 & 743 \\ 
	37 & 544 & 413 & 51 & 362 & -41556 & 86861 & 738 \\ 
	38 & 549 & 417 & 51 & 366 & -41506 & 86797 & 748 \\ 
	39 & 552 & 421 & 51 & 370 & -41464 & 86732 & 740 \\ 
	40 & 552 & 422 & 51 & 371 & -41432 & 86669 & 742 \\ 
	41 & 552 & 424 & 51 & 373 & -41414 & 86633 & 742 \\ 
	42 & 553 & 426 & 51 & 375 & -41390 & 86591 & 757 \\ 
	43 & 555 & 429 & 52 & 377 & -41378 & 86581 & 749 \\ 
	44 & 557 & 431 & 52 & 379 & -41368 & 86576 & 757 \\ 
	45 & 557 & 431 & 52 & 379 & -41368 & 86576 & 755 \\ 
	46 & 557 & 431 & 52 & 379 & -41368 & 86576 & 758 \\ 
	47 & 558 & 432 & 52 & 380 & -41366 & 86578 & 758 \\ 
	48 & 558 & 432 & 52 & 380 & -41366 & 86578 & 759 \\ 
	49 & 558 & 432 & 52 & 380 & -41366 & 86578 & 749 \\ 
	50 & 561 & 434 & 52 & 382 & -41346 & 86559 & 761 \\ 
	51 & 561 & 434 & 52 & 382 & -41346 & 86559 & 762 \\ 
	\hline
	\caption{Numerical results for Dissmann algorithm. Calculations based on \textit{z-scale}, abbreviations \textit{ni-pc} for \textit{non independence pair copula}, \textit{G-pc} for \textit{Gaussian pair copula}.}
	\label{table:application:dissmann}
\end{longtable}
%\end{table}

\subsection{Distribution of non-independence copulas in the Euro Stoxx 50}\label{sec:appendix_heatmap}
To visualize the actual truncation levels of the R-vines based on a DAG with at most $k=2$ parents, we consider the distribution of independence pair copulas. Thus, we plot a $52 \times 52$ matrix indicating which pair copulas are the independence copula in the R-vine representation of the DAG $\GGG_2$, see the lower triangular region of Figure \ref{fig:stoxx:heatmap}, created with the R-package \textit{gplots}, see \citet{gplots}. The upper triangular region encodes which pair copulas are set to the independence copula when we use an additional level $\alpha = 0.05$ independence test. We see the sparsity patterns of the corresponding R-vine models and note that each independence pair copula in the lower triangular is also in the upper triangular, where the upper triangular may also have additional independence pair copulas. It also indicates that an independence test based on the d-separation is not sufficient when dealing with non-Gaussian dependency patterns, as a huge number of pair copulas with small Kendall's are not associated with the independence copula upfront.
\begin{figure}[H]
	\centering
	\includegraphics[width=0.99\textwidth,trim={0.0cm 0.0cm 0.0cm 0.0cm},clip]{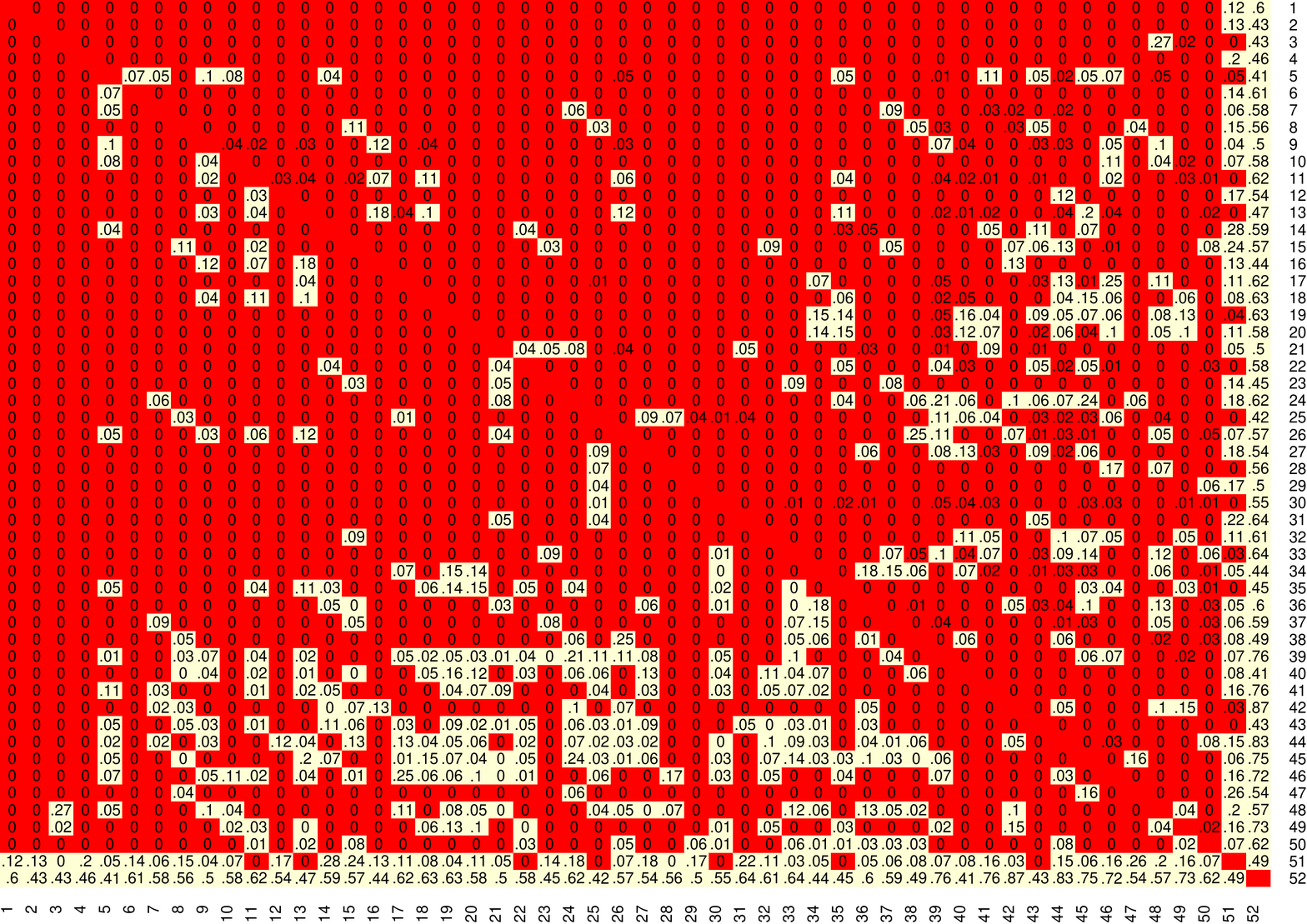}
	\caption{Distribution of independence pair copulas in the R-vine generated by algorithm \texttt{RepresentDAGRVine} for at most $k=2$ parents and absolute values of Kendall's $\tau$ of the corresponding pair copulas. Bright white colour indicates a non independence copula whereas dark red colour indicates an independence copula. The lower triangular describes the R-vine model without additional $\alpha=0.05$ independence test. The last row of the matrix represents the first R-vine tree, the second last the second R-vine tree and so on. The upper triangular represents the same information in transposed form for the R-vine model using an additional $\alpha=0.05$ independence test, i.\,e.\ the last column represents the first R-vine tree, the second last column the second R-vine tree and so on. Thus, the models can be compared along the main diagonal.}
	\label{fig:stoxx:heatmap}
\end{figure}

\newpage

\section{Algorithms}\label{sec:appendix_algorithms}
\subsection{RepresentMarkovTreeRVine}\label{subsec:appendix_representmarkovtreervine}
\IncMargin{0em}
\begin{algorithm}[h]
	\SetAlgoLined
	\SetKwInOut{Input}{input}
	\SetKwInOut{Output}{output}
	\SetKwFunction{Create}{create}
	\SetKwFunction{Calculate}{calculate}
	\SetKwFunction{Define}{define}
	\SetKwFunction{Order}{order}
	\SetKwFunction{Sample}{sample}
	\SetKwFunction{Return}{return}
	\SetKwFunction{List}{list}
	\SetKwFunction{Complete}{complete}
	\SetKwFunction{Set}{set}
	\SetKwFunction{Add}{add}
	\SetKwFunction{Exit}{exit}
	\SetKwFunction{Donothing}{do nothing}
	\Input{DAG $\GGG=\left(V=\left(v_1,\ldots,v_d\right),E\right)$ with topological ordering $v_i >_\GGG v_j$, truncation level $k=\max_{v \in V}\left|\parents\left(v\right)\right|=1$.}
	\Output{R-vine tree sequence $T_1,\ldots,T_{d-1}$ given by a R-vine matrix $M$ and an independence matrix $F \in \left\{0,1\right\}^{d \times d}$, indicating which pair copula families can be set to the independence copula.}
	\BlankLine
	\Set $M = \diag\left(d,\ldots,1\right)$\;
	\Set $F = \left(0\right)^{d \times d}$\;
	\For{$i = d-1$ \KwTo $1$}{
		\Set $M_{d,i}=\parents\left(M_{i,i}\right)$\;
		\Set $F_{d,i} = 1$\;
	}
	\Complete $M$ according to the \textit{proximity condition}\;
	\Return $M,F$\;
	\caption{\texttt{RepresentMarkovTreeRVine}: Construction of a R-vine tree matrix $M$ and independence matrix $F$ obtained from a DAG $\GGG = \left(V,E\right)$ with at most one parent.}
	\label{algorithm:markovtreervineselect}
\end{algorithm}

\subsection{RepresentDAGRVine - Structure estimation}\label{subsec:appendix_representdagrvine_1}
\IncMargin{0em}
\begin{algorithm}[H]
	\SetAlgoLined
	\SetKwInOut{Input}{input}
	\SetKwInOut{Output}{output}
	\SetKwFunction{Create}{create}
	\SetKwFunction{Assign}{assign}
	\SetKwFunction{Calculate}{calculate}
	\SetKwFunction{Define}{define}
	\SetKwFunction{Order}{order}
	\SetKwFunction{Sample}{sample}
	\SetKwFunction{Return}{return}
	\SetKwFunction{Weight}{weight}
	\SetKwFunction{List}{list}
	\SetKwFunction{Complete}{complete}
	\SetKwFunction{Estimate}{estimate}
	\SetKwFunction{Set}{set}
	\SetKwFunction{Add}{add}
	\SetKwFunction{Delete}{delete}
	\SetKwFunction{Skip}{skip}
	\SetKwFunction{Exit}{exit}
	\SetKwFunction{Donothing}{do nothing}
	\Input{DAGs $\GGG_i, i = 1,\dots,k$ with at most $i$ parents, weighting function $g\left(i\right)$.}
	\Output{R-vine matrix $M$ and independence matrix $F$ indicating which pair copulas are the independence copula, truncation level $k'$.}
	\BlankLine
	\Calculate skeletons $\GGG_i^s$ of DAGs $\GGG_i$ for $i=1,\dots,k$\;
	\Create $\HHH=\left(V,E^{\HHH}_1\right) \coloneqq \bigcup_{i=1}^k~\GGG_i^s$\;
	\Set weights $\mu_1\left(v,w\right) = \sum_{i=1}^{k}~g\left(i\right)\mathds{1}_{\left(v,w\right) \in E_i^s}\left(v,w\right)$ for each edge $\left(v,w\right) \in E^{\HHH}_1$\;
	\Calculate maximum spanning tree $T_1=\left(V,E_1^T\right)$ on $\HHH$\;
	\For{$i = 2$ \KwTo $d-1$}{
		\Create full undirected graph $\HHH_i=\left(V_i = E_{i-1},E_i^{\HHH}\right)$\;
		\Delete edges not allowed by the proximity condition\;
		\For{$e \in E_i^{\HHH}$}{
			\eIf{$\mu_1\left(j\left(e\right),\ell\left(e\right)\right) \neq 0$}{
				\Assign DAG weights $ \mu_i\left(e\right)=\mu_1\left(j\left(e\right),\ell\left(e\right)\right)$\;}{
				\eIf{$\condindep{j\left(e\right)}{\ell\left(e\right)}{D\left(e\right)}$ according to d-separation in $\GGG_k$}{
					\Assign independence weight $ \mu_i\left(e\right)=\mu_0$\;}{}}
		}
		\Calculate maximum spanning tree $T_i=\left(V_i,E_i^T\right)$ on $\HHH_i$\;}
	\Create R-vine matrix $M$ from R-vine trees $T_1,\ldots,T_{d-1}$\;
	\caption{\texttt{RepresentDAGRVine}: Calculation of an R-vine tree matrix $M$ obtained from DAGs $\GGG_1,\ldots,\GGG_k$.}
	\label{algorithm:dagvineselect1}
\end{algorithm}

\newpage

\subsection{RepresentDAGRVine - Inference of independence copulas}\label{subsec:appendix_representdagrvine_2}
\IncMargin{0em}
\begin{algorithm}[H]
	\SetAlgoLined
	\SetKwInOut{Input}{input}
	\SetKwInOut{Output}{output}
	\SetKwFunction{Create}{create}
	\SetKwFunction{Assign}{assign}
	\SetKwFunction{Calculate}{calculate}
	\SetKwFunction{Define}{define}
	\SetKwFunction{Order}{order}
	\SetKwFunction{Sample}{sample}
	\SetKwFunction{Return}{return}
	\SetKwFunction{Weight}{weight}
	\SetKwFunction{List}{list}
	\SetKwFunction{Complete}{complete}
	\SetKwFunction{Estimate}{estimate}
	\SetKwFunction{Set}{set}
	\SetKwFunction{Add}{add}
	\SetKwFunction{Delete}{delete}
	\SetKwFunction{Skip}{skip}
	\SetKwFunction{Exit}{exit}
	\SetKwFunction{Donothing}{do nothing}
	\Input{R-vine matrix $M$ and DAG $\GGG_k$.}
	\Output{Independence matrix $F$ indicating which pair copulas are the independence copula and truncation level $k'$.}
	\BlankLine
	\Set $F = \left(0\right)^{d \times d}$\;
	\For{$i = d$ \KwTo $2$}{
		\For{$j = i-1$ \KwTo $1$}{
			\eIf{$\condindep{M_{i,j}}{M_{j,j}}{M_{i+1,j},\ldots,M_{d,j}}$ according to d-separation in $\GGG_k$}{
				\Set $F_{i,j} = 0$\;}{
				\Set $F_{i,j} = 1$\;}
		}
	}
	\Set $k^\prime$ such that in the R-vine trees $T_{k^\prime+1},\ldots,T_{d-1}$ only the independence copula occurs\;
	\caption{\texttt{RepresentDAGRVine}: Calculation of an independence matrix $F$ and truncation level $k^\prime$ for an R-vine matrix obtained from Algorithm \ref{algorithm:dagvineselect2}.}
	\label{algorithm:dagvineselect2}
\end{algorithm}

\end{document}